\documentclass[runningheads, a4paper]{llncs}
\usepackage{amssymb}

\usepackage{geometry}
\usepackage{hyperref}
\usepackage[norsk, english]{babel} 
\usepackage[latin1]{inputenc} 
\usepackage[T1]{fontenc}
\usepackage{etex}
\usepackage{epsfig,stmaryrd}
\usepackage{xspace}
\usepackage{courier}
\usepackage{color}
\usepackage{xcolor}
\usepackage{tikz,tikzsymbols}
\usepackage[underline=false]{pgf-umlsd}
\usepackage{floatflt} 
\usepackage[normalem]{ulem}%
\usepackage{wrapfig}
\usepackage{placeins}
\usepackage{listings}
\usepackage{caption}
\usepackage{bm}
\usetikzlibrary{calc}
\usetikzlibrary{arrows,shapes,trees,positioning,decorations,shapes,automata}
\usetikzlibrary{decorations.markings,automata}
\usepackage{relsize}
\usepackage{syntax}
\usepackage{float}
\usepackage{mathtools}
\usepackage[inline]{enumitem}
\usepackage[caption=false]{subfig}

\usepackage{pgf}
\usepackage{verbatim}

\def\codesize{\fontsize{9}{10}}
\definecolor{mycolor}{rgb}{0.94,0.94,0.94}

\DeclareCaptionFormat{listing}{\rule{\dimexpr\textwidth+17pt\relax}{0.4pt}\par\vskip1pt#1#2#3}
\DeclareMathVersion{sans}
\SetSymbolFont{operators}{sans}{OT1}{cmbr}{m}{n}
\SetSymbolFont{letters}  {sans}{OML}{cmbrm}{m}{it}
\SetSymbolFont{symbols}  {sans}{OMS}{cmbrs}{m}{n}

\lstdefinelanguage{simple}{keywords=
{co,oc,spawn,send,goto,receive,call,atomic,delta,new,data,type,def,case,of,cog,class,interface,extends,implements,if,then,else,for,await,get,Fut,return,skip,while,module,input,
import, export, from, this, destiny, suspend,,adds,modifies,removes,original,productline,features,core,corefeatures,optionalfeatures,after,when,product,hasAttribute,hasMethod,root,extension,component,cost,group,allof,oneof,require,exclude,original,ifin,ifout,opt,now,some,fi,var,and,or,assert,InvocEv,InvocrEv,FutureEv,FetchEv,NewEv,Pre,Post, Require,Ensure, absSource,invariants,functions,predicates,rules}, 
  sensitive=true,
  morecomment=[l]{//},
  morestring=[b]"}

\makeatletter
\lstnewenvironment{simple}[1][]{
 \lstset{language=SIMPLE,columns=fullflexible,mathescape=true,%
           keywordstyle=\bf\ttfamily,commentstyle=\sl\ttfamily,%
           basicstyle=\codesize\sffamily,inputencoding=latin1, %
           extendedchars,xleftmargin=2.5em,
           numbers=left,
           escapechar=\%,
           stepnumber=1,
           breaklines=false,
           breakatwhitespace=false,
           escapeinside={(*@}{@*)},
           captionpos=t,#1}%
   \csname lst@SetFirstLabel\endcsname}
  {\csname lst@SaveFirstLabel\endcsname}
\makeatother

\newcommand{\Simple}[1]{\mbox{\lstinline[language=simple,columns=fullflexible,mathescape=true,inputencoding=latin1,extendedchars,keywordstyle=\bf\ttfamily,basicstyle=\codesize\sffamily]|#1|}}

\lstdefinelanguage{abs}{keywords=
{spawn,send,receive,call,atomic,delta,new,data,type,def,case,of,cog,class,interface,extends,implements,if,then,else,for,await,get,Fut,return,skip,while,module,
import, export, from, this, destiny, suspend, delta,adds,modifies,removes,original,productline,features,core,corefeatures,optionalfeatures,after,when,product,hasAttribute,hasMethod,root,extension,component,cost,group,allof,oneof,require,exclude,original,ifin,ifout,opt,now,some,fi,var,and,or,assert,InvocEv,InvocrEv,FutureEv,FetchEv,NewEv,Pre,Post, Require,Ensure, absSource,invariants,functions,predicates,rules}, 
 sensitive=true,
 morecomment=[l]{//},
 morestring=[b]"}

\makeatletter
\lstnewenvironment{abs}[1][]{
\lstset{language=ABS,columns=fullflexible,mathescape=true,%
          keywordstyle=\bf\ttfamily,commentstyle=\sl\ttfamily,%
          basicstyle=\codesize\sffamily,inputencoding=latin1, %
          extendedchars,xleftmargin=2.5em,
          numbers=left,
          escapechar=|,
          stepnumber=1,
          breaklines=false,
          breakatwhitespace=false,
          escapeinside={(*@}{@*)},
          captionpos=t,#1}%
  \csname lst@SetFirstLabel\endcsname}
 {\csname lst@SaveFirstLabel\endcsname}
\makeatother

\DeclareMathOperator{\isfut}{isFut}
\def\promela{\textsc{ProMeLa}}

\newcommand{\TODO}[1]{\textcolor{red}{[TODO] #1}}

\newcommand{\para}{\ensuremath~||~}

\newcommand{\fd}{\ensuremath \mathit{fd}}

\newcommand{\methods}[1]{\method(#1)}
\newcommand{\Gi}{\ensuremath \mathcal{G}}

\newcommand{\stmt}{\ensuremath \mathit{s}}

\newcommand{\guard}{\ensuremath \mathit{guard}}
\newcommand{\scope}{\ensuremath \mathit{sc}}

\newcommand{\last}{\ensuremath \mathrm{last}}
\newcommand{\first}{\ensuremath \mathrm{first}}
\newcommand{\cont}[1]{\ensuremath \mathrm{K}(#1)}

\newcommand{\suspendedcontF}[3]{\ensuremath \mathrm{K}^{#1}(\Simple{await}\,#2;#3)}

\newcommand{\contF}[2]{\ensuremath \mathrm{K}^{#1}(#2)}
\newcommand{\zero}{\ensuremath \circ}
\renewcommand{\zero}{\mbox{\bottle}}
\newcommand{\pc}{pc} 
\newcommand{\pop}{\triangleright} 

\newcommand{\multiunion}{{+}}
\newcommand{\subst}[2]{[#1\leftarrow#2]}
\newcommand{\mapComp}[2]{#1 \bullet #2}
\newcommand{\sh}{\ensuremath \textit{sh}}
\newcommand{\traces}[1]{\ensuremath\mathbf{Tr}(#1)}
\newcommand{\ctraces}{\ensuremath\mathbf{CTr}}

\newcommand{\trueSem}{\ensuremath \mathrm{t\!t}}
\newcommand{\falseSem}{\ensuremath \mathrm{f\!f}}

\newcommand{\valB}[2]{\ensuremath \mathrm{val}_{#1}(#2)}
\newcommand{\valD}[1]{\ensuremath \mathrm{val}_{\sigma}(#1)}
\newcommand{\valP}[2]{\ensuremath \mathrm{val}_{\sigma}^{#1}(#2)}
\newcommand{\valBP}[3]{\ensuremath \mathrm{val}_{#1}^{#2}(#3)}
\newcommand{\valDP}[1]{\ensuremath \mathrm{val}_{\sigma'}(#1)}
\newcommand{\valDnoargs}{\ensuremath \mathrm{val}_{\sigma}}
\newcommand{\chopSem}{\mathbin{\ast\ast}}
\newcommand{\update}[3]{#1[#2 \mapsto #3]}
\newcommand{\chopTrSem}[2]{\ensuremath #1 \chopSem #2}
\newcommand{\concatTr}[2]{\ensuremath #1\cdot #2}
\newcommand{\consTr}[2]{\ensuremath #2 \curvearrowright #1}

\DeclareMathOperator{\dom}{dom}
\DeclareMathOperator{\Class}{class}
\DeclareMathOperator{\method}{methods}
\DeclareMathOperator{\field}{fields}
\DeclareMathOperator{\lp}{lookup}

\DeclareMathOperator{\vars}{vars}
\DeclareMathOperator{\symb}{symb}
\newcommand{\symbvars}[1]{\symb(#1)}
\newcommand{\concreteop}[2]{#1\:\underline{op}\:#2}
\newcommand{\fields}[1]{\field(#1)}
\newcommand{\lookup}[1]{\lp(#1)} 
\newcommand{\clookup}[1]{\Class(#1)}

\newcommand{\evTrioP}[2]{\textit{ev}_{#1}({#2})} 
\newcommand{\evTrioPV}[3]{\textit{ev}_{#1}^{#3}({#2})} 

\newcommand{\inpEvpV}[3]{\textit{inpEv}_{#1}^{#3}(#2)} 
\newcommand{\inpEv}[1]{\textit{inpEv}(#1)}
\newcommand{\invEv}[2]{\textit{invEv}_{#1}(#2)}
\newcommand{\invEvV}[3]{\textit{invEv}_{#1}^{#3}(#2)} 
\newcommand{\invREv}[2]{\textit{invREv}_{#1}(#2)}
\newcommand{\invREvV}[3]{\textit{invREv}_{#1}^{#3}(#2)} 
\newcommand{\sendEv}[4]{\textit{sendEv}_{#1}(#2, #3, #4)}
\newcommand{\receiveEv}[4]{\textit{receiveEv}_{#1}(#2, #3, #4)}
\newcommand{\receiveEvV}[5]{\textit{receiveEv}_{#1}^{#5}(#2, #3, #4)} 

\newcommand{\spawnEvs}[3]{\textit{spawnEv}_{#1}(#2, #3)}
\newcommand{\spawnEvsV}[4]{\textit{spawnEv}_{#1}^{#4}(#2, #3)} 
\newcommand{\newEvs}[2]{\textit{newEv}_{#1}(#2)}
\newcommand{\newEvsV}[3]{\textit{newEv}_{#1}^{#3}(#2)}
\newcommand{\anotherevp}{\textit{ev}'^{P'}(\many{v'})} 
\newcommand{\invEvp}[3]{\textit{invEv}_{#1}^{#2}(#3)}  
\newcommand{\invREvp}[3]{\textit{invREv}_{#1}^{#2}(#3)}  
\newcommand{\sendEvp}[5]{\textit{sendEv}_{#1}^{#2}(#3, #4, #5)} 
\newcommand{\receiveEvp}[5]{\textit{receiveEv}_{#1}^{#2}(#3, #4, #5)} 
\newcommand{\spawnEvp}[4]{\textit{spawnEv}_{#1}^{#2}(#3, #4)} 
\newcommand{\newEvp}[3]{\textit{newEv}_{#1}^{#2}(#3)}
\newcommand{\compEv}[2]{\textit{compEv}_{#1}(#2)} 
\newcommand{\compREv}[2]{\textit{compREv}_{#1}(#2)} 
\newcommand{\compREvV}[3]{\textit{compREv}_{#1}^{#3}(#2)} 
\newcommand{\compEvp}[3]{\textit{compEv}_{#1}^{#2}(#3)} 
\newcommand{\compREvp}[3]{\textit{compREv}_{#1}^{#2}(#3)} 

\newcommand{\ev}[1]{\textit{ev}(#1)} 
\newcommand{\evp}[2]{\textit{ev}^{#1}(#2)} 
\newcommand{\stateorevent}{t}
 
\newcommand{\identifier}{\textit{identifier}} 
\newcommand{\wfT}[2]{\ensuremath \mathit{wf}_{\text{#1}}(#2)}
\newcommand{\wf}[1]{\ensuremath \mathit{wf}(#1)}
\newcommand{\num}[2]{\ensuremath \#_{#1}({#2})}
\newcommand{\many}[1]{\overline{#1}}

\newcommand{\MId}{\textit{MId}}
\newcommand{\PId}{\textit{PId}}
\newcommand{\OId}{\textit{OId}}
\newcommand{\FId}{\textit{FId}}

\newcommand{\Omain}{\ensuremath {o_{\mathit{main}}}}
\newcommand{\Fmain}{\ensuremath {f_{\mathit{main}}}}

\newcommand{\evalruleN}[2]{%
  \begin{equation}
    \label{eq:#1}
    \begin{array}{@{}r@{\,}c@{\,}l@{}}
      #2
    \end{array}
  \end{equation}
}

\newcommand{\gcondrule}[3]{ 
  \begin{equation}
   \label{eq:rule#1}
  \begin{array}{c} 
    #2 \\[1pt] 
    \hline\\[-7pt]
    #3 
  \end{array} 
 \end{equation}
    }

\newcommand{\grule}[2]{ 
  \begin{equation}
   \label{eq:rule#1}
  \begin{array}{c} 
    #2 \\[1pt] 
  \end{array} 
 \end{equation}
    }

\usepackage{mathpartir}

\title{LAGC Semantics of Concurrent Programming Languages} \author{Crystal Chang Din \inst{1} \and Reiner
  H\"ahnle \inst{2} \and Ludovic Henrio \inst{3} \and Einar Broch Johnsen \inst{4} \and \\ Ka I
  Pun \inst{5} \and S. Lizeth Tapia Tarifa \inst{4} }
\institute{
Department of Informatics, University of Bergen, Norway\\
\email{Crystal.Din@uib.no}
\and
Department of Computer Science, Technische Universit\"at
  Darmstadt, Germany  \\
  \email{haehnle@cs.tu-darmstadt.de} 
\and 
CNRS, Inria, LIP, Univ Lyon, EnsL, UCBL, France\\
\email{ludovic.henrio@cnrs.fr}
 \and
  Department of Informatics, University of Oslo, Norway\\
  \email{\{einarj,sltarifa\}@ifi.uio.no} 
\and
Department of Computing, Western Norway University of Applied Sciences,  Norway   \\
  \email{Violet.Ka.I.Pun@hvl.no}
  }
  
  \titlerunning{LAGC Semantics}
  \authorrunning{C.~C. Din \emph{et al.}}

\begin{document}

\pagestyle{headings} 

\newcommand{\red}[1]{\textcolor{red}{#1}}
\newcommand{\emp}{\circ}

\newif\ifadvancedActors
\advancedActorsfalse

 \maketitle 

\begin{abstract}
  Formal, mathematically rigorous programming language semantics are
  the essential prerequisite for the design of logics and calculi that
  permit automated reasoning about concurrent programs.  We propose a
  novel modular semantics designed to align smoothly with program
  logics used in deductive verification and formal specification of
  concurrent programs. Our semantics separates local evaluation of
  expressions and statements performed in an abstract, symbolic
  environment from their composition into global computations, at
  which point they are concretised. This makes incremental addition of
  new language concepts possible, without the need to revise the
  framework.  The basis is a generalisation of the notion of a program
  trace as a sequence of evolving states that we enrich with event
  descriptors and trailing continuation markers. This allows to
  postpone scheduling constraints from the level of local evaluation
  to the global composition stage, where well-formedness predicates
  over the event structure declaratively characterise a wide range of
  concurrency models.  We also illustrate how a sound program logic
  and calculus can be defined for this semantics.
\end{abstract}

\section{Introduction}
\label{sec:introduction}

We propose a trace semantics that scales flexibly to a range of
concurrent, imperative programming paradigms, as found, for example,
in C, \textsc{Java}, \promela\ \cite{Holzmann03}, Actors \cite{HBS73},
or Active Objects~\cite{ActiveObjects17}.  Specifically, given a
program~$P$, running on a collection of communicating processors or
cores, we want to obtain the set of all global system traces that~$P$
can produce starting from some initial state.
Our overall goal is to provide a semantics that aligns well with
contract-based deductive verification \cite{HaehnleHuismann19}.  To
this end, it is essential that the semantics cleanly separates local
computations (one statement on one core) from global ones; i.e., there
must be a suitable notion of \emph{composition} that can generate the
global traces from the local ones.  It is known
\cite{brookes96ic,Mosses06} that compositionality in denotational
semantics shifts the flavour of the semantics in the direction of
operational semantics~\cite{Plotkin04}, by transitioning from state
transformers \cite{Hoare69,apt09vscp} to traces of execution states
\cite{hennessy79mfcs,Brookes02,Brookes07} or communication events
\cite{Dahl77,hoare85,jeffrey05esop,dinJLA12}. We propose a hybrid
notion of trace that combines states and events, and show that the
ensuing semantic framework is well-suited to compositionally capture
different forms of concurrency for imperative languages.

The envisaged semantics should also be \emph{modular} in the following
sense: it must be possible to support a new language construct without
the need to revise the whole framework. Ideally, there is a single
evaluation rule for each construct of the target language that can be
applied \emph{independently} from all other rules: the recursive call
to evaluate subsequent statements is not inside the semantic
evaluation of each statement.  This is not only a good match with
deductive verification rules of program logics
\cite{HaehnleHuismann19}, but also with formal specification languages
for concurrent programs that have a trace semantics
\cite{SistlaClarke85,BeckertBruns13,DHJPT17,SH19b,Kamburjan19}.

A trace semantics for a given target language, satisfying the
requirements sketched above, can be defined in three phases:
\begin{enumerate}
\item\label{item:phase1} Declare \emph{local} evaluation rules for
  each syntax construct of the target language.
\item\label{item:phase2} Declare \emph{composition} rules that combine
  local evaluation and process creation into global traces. Scheduling
  is expressed declaratively as well-formedness constraints over
  traces.
\item\label{item:phase3} Define the \emph{generation} of all global
  traces with the help of the composition rules from a suitable
  initial configuration.
\end{enumerate}

To achieve the desired degree of modularity, we generalise the
standard notion of a program trace~\cite{HuthRyan04}, i.e.\ a sequence
of evolving system configurations, starting in some initial state. We
make two generalisations. The first is that local states $\sigma$ can
be abstract. This means the value $\sigma(x)$ of a memory location $x$
is permitted to be unspecified. One can think of an abstract value
$\sigma(x)$ as a Skolem constant $x_0$ whose interpretation is
determined by an external context. Alternatively, think of
\emph{symbolic execution} \cite{Burstall74,King76}, where the value of
a memory location $x$ may be a \emph{symbolic} term $x_0$ representing
an unknown value. States containing abstract values are called
\emph{symbolic}, otherwise they are called \emph{concrete}.

Symbolic states allow us to evaluate local code independently of its
call context. For example, when evaluating the semantics of a
statement $s$ that receives a message from another process, then the
value of that message cannot possibly be known independently of the
call context. In this case, the local evaluation of $s$ can be
expressed in a trace with symbolic states.  These symbolic states are
\emph{concretised} when the global context of the local computations
(here, the sender) is resolved, i.e.\ in the composition rules during
phase~\eqref{item:phase2} above. Hence, the resulting global traces
are concrete. Thus, the name of our semantic framework: \emph{locally
  abstract, globally concrete} (LAGC).

The second generalisation of traces concerns \emph{scheduling}:
concurrency models differ in the exact locations at which a local
computation can be interrupted (aka preempted) and how exactly the
computation is continued afterwards, i.e.\ which scheduling decision
is taken next. To achieve maximum modularity, we do not build
scheduling into a fixed set of rules. Instead, we use
\emph{continuation markers} to specify when scheduling and preemption
is possible.

As a result, local traces are not merely abstract, but as well contain
\emph{events} used in the composition rules as interaction
points. With the event mechanism, various concurrency models can be
defined easily in two parts: first, ensure that local evaluation rules
generate suitable synchronisation events. Second, define a
\emph{well-formedness} predicate on concrete traces restricting the
global traces that can be generated to those reflecting the targeted
concurrency model. Now it is sufficient to add the well-formedness
predicate as a premise to each composition rule.  Events turn out to
be a flexible and intuitive mechanism, which can be used to express a
range of dependencies between processes beyond synchronous
communication (for example, object generation).

An LAGC-style trace semantics was pioneered for the active object
language \textsc{ABS} \cite{DHJPT17}. Here we show that, due to its
modularity, LAGC semantics constitutes a general \emph{semantic
  framework} for a wide range of concurrent programming models: it is
easy to add new syntactic constructs and to accommodate different
concurrency paradigms. We will demonstrate this as follows: starting
with a simple \Simple{WHILE}-language, we define LAGC semantics for an
increasingly complex series of languages by successively adding new
features, sequential as well as parallel ones. In the end we cover a
representative set of language features and concurrency models and
instantiate our semantic framework to two quite different concurrent
programming languages.

The paper is organised as follows: In Section~\ref{sec:basics} we set up
the formal framework that the LAGC semantics is based upon. As explained
above, we need symbolic states and traces that may contain abstract
values for memory locations. We also need a concretisation operator
that instantiates a symbolic trace to match a concrete context. In
addition, we equip traces with states and continuation markers.
Section~\ref{sec:while} introduces the LAGC framework along
phases~\eqref{item:phase1}--\eqref{item:phase3} for a
\Simple{WHILE}-language.
We substantiate the claim made above, that an LAGC semantics is a good
match for deductive verification calculi, by defining in
Section~\ref{sec:dynamic-logic-simple} a program logic for
\Simple{WHILE} with a concise soundness proof.
In Section~\ref{sec:interleaving}, we gradually extend the semantics to
local parallelism (with atomicity), local memory, and procedure
calls.
In Section~\ref{sec:multicore}, we take the step to multiple processors
that send and receive messages among each other. We show that a wide
range of communication patterns can be intuitively and declaratively
characterised via well-formedness, including synchronous and
asynchronous communication, bounded and unbounded FIFO, as well as
causality.
Section~\ref{sec:case-studies} instantiates the LAGC framework to the
rather different concurrency models found in the languages \promela\
and ABS, respectively. For the latter, we need to add objects and
futures, as well as to change the interleaving semantics, which turns
out to be easily possible.
Related work is discussed in Section~\ref{sec:related-work}, while
Section~\ref{sec:future-work-concl} considers future directions for work
and concludes.

\section{Basics}
\label{sec:basics}

\subsection{States}
\label{sec:states}

We assume given standard basic types, including integers, Booleans and
process identifiers, with standard value domains and associated
operators. In addition, we allow starred expressions, representing
unknown, \emph{symbolic} values.

\begin{figure}[t]\centering
  $\begin{array}[t]{r@{\hspace{2pt}}r@{\hspace{2pt}}l}
x \in  \mathit{Var} & ::= & \identifier\\
v \in \mathit{Val} &::=& \trueSem \mid \falseSem \mid 0 \mid 1 \mid ...\\       
op \in \mathit{Op} & ::= & ==~\mid~<~\mid~>~\mid~\leq~\mid~\geq~\mid
                            \lit*{+} \mid \lit*{-} \mid
              \lit*{*} \mid \lit*{/} \mid \ldots\\
e \in \mathit{Exp}  &::=&  x \mid v \mid e\, op\,  e \\\
se\in \mathit{Sexp} &::=& e \mid \ast 
     \end{array}$
\caption{\label{fig:basiclang}Syntax of expressions.}
\end{figure}

\begin{definition}[Variables, Values, (Starred) Expressions]
  Let $\mathit{Var}$ be a set of program variables, $\mathit{Val}$ a
  set of values, and $\mathit{Op}$ a set of operators, with typical
  elements $x$, $v$, and $op$, respectively.  The sets $\mathit{Val}$
  and $\mathit{Op}$ include the values and operators of the basic
  types (see Figure~\ref{fig:basiclang}).
  The set $\mathit{Exp}$ contains expressions, with typical element
  $e$, obtained from variables, values, and by applying operators to
  expressions.
  The set $Sexp$ contains starred expressions, with typical element
  $se$, obtained by extending expressions with an additional symbol
  $\ast$.
\end{definition}

We assume that all expressions are well-typed; i.e., in expressions,
operators are only applied to subexpressions which can be reduced to
values of appropriate types.  Let $\concreteop{}{}$ denote an
evaluation function for the operators $\mathit{op}$ defined over the
values of the basic types, such that $\concreteop{v\:}{\:v'}$ is a
value in the basic types. A Boolean expression is an expression that
evaluates to a Boolean value when its arguments are values of the
basic types, and similarly for expressions of other basic types.
Overline notation is used for lists of different syntactic categories;
e.g., $\many{v}$ and $\many{e}$ represent lists of values and
expressions, respectively.  Let $\vars(e)$ denote the set of variables
in an expression $e$, which has a straightforward inductive
definition.

We now define computation states.  Usually, states are mappings from
variables to concrete values. To permit \emph{symbolic} expressions
(i.e., expressions containing variables) occurring as \emph{values} in
states, the starred expression $\ast$ is used to represent a value for
symbolic variables that cannot be further evaluated. The $\ast$ symbol
does not occur in programs, it is part of the semantic domain.

\begin{definition}[Symbolic State, State Update]
  A \emph{symbolic state} $\sigma$ is a partial mapping
  \[\sigma: \mathit{Var} \to \mathit{Sexp} \]
  from variables to starred expressions. The notation
  $\sigma[x\mapsto se]$ expresses the \emph{update} of state $\sigma$
  at $x$ with expression $se$:
  \[
    \sigma[x\mapsto se](y)=\left\{
      \begin{array}{ll}
        se & x=y\\
        \sigma(y)  & x\neq y \ .
      \end{array}
      \right.
  \]  
\end{definition}

In a symbolic state, a \emph{symbolic variable} is defined as a
variable bound to an unknown value, represented by the starred
expression $\ast$. Symbolic variables play a different role than
ordinary variables: they act as parameters, relative to which a local
computation is evaluated. To distinguish them syntactically, we adopt
the convention of using capital letters for symbolic variables as much
as possible. Note that the set of symbolic and non-symbolic variables
are not two distinct categories and a variable that is symbolic can
become non-symbolic after substitution (see
Example~\ref{exmpl:wf-state} below).

\begin{definition}[Symbolic Variable]\label{def:symbolic-variables}
  \emph{Symbolic variables} in a state $\sigma$ are variables mapped
  to $\ast$:
  \[\symbvars{\sigma}=\{X\in \mathit{Var} \mid\, \sigma(X)=\ast\}\ .\]
\end{definition}

We will also need the concept of the extension of a state:

\begin{definition}[State Extension]\label{def:state-extension}
  We say that a state $\sigma'$ \emph{extends} a state $\sigma$ if
  \begin{enumerate*}[label=(\roman*)]
  \item $\dom(\sigma)\subseteq\dom(\sigma')$ and
  \item $\sigma(x)=\sigma'(x)$ for all $x\in\dom(\sigma)$.
  \end{enumerate*}
  We overload the subset symbol and write $\sigma\subseteq\sigma'$.
\end{definition}

\begin{example}
  \label{exmpl:state}
  Consider a state
  $\sigma_0=[x_0\mapsto Y_0+w_0,\,Y_0\mapsto\ast,\,w_0\mapsto
  42,\,x_1\mapsto Y_1]$.
  Observe that (1)~the expressions in the range of $\sigma_0$ can be
  simplified
  and (2) there are dangling references, such as $Y_1$, not in the
  domain of $\sigma_0$.
\end{example}

The example shows that symbolic states are slightly too general for
our purpose, motivating the following definition, which constrains the
variables that may occur in value expressions of states to symbolic
variables:
  
\begin{definition}[Well-Formed State]
  A state $\sigma$ is \emph{well-formed} if it fulfils the following
  condition:
  \[
    \{ x\in \vars(\sigma(y))\mid y\in\dom(\sigma)\} \subseteq\symbvars{\sigma}\enspace.
  \]
\end{definition}

A well-formed state $\sigma$ is \emph{concrete} if
$\symbvars{\sigma}=\{\}$. For a concrete, well-formed state $\sigma$
and all $x\in\dom(\sigma)$, there is a value $v$ such that
$\sigma(x)=v$.

\begin{example}
  \label{exmpl:wf-state}
  The state $\sigma_0$ in Example~\ref{exmpl:state} can be turned into
  a well-formed state $\sigma_1$ by simplifying the expression
  $Y_0+w_0,$ and binding $Y_1$ to a star, obtaining
  $\sigma_1=[x_0\mapsto Y_0+42,\,Y_0\mapsto\ast,\,w_0\mapsto
  42,\,x_1\mapsto Y_1,\,Y_1\,\mapsto\ast]$.  We have that
  $\sigma_1=\sigma_0[x_0\mapsto Y_0+42,\,Y_1\mapsto\ast]$ and
  $\symbvars{\sigma_1}=\{Y_0,\,Y_1\}$. Let
  $\sigma_2=\sigma_1[Y_0\mapsto 3,\,Y_1\mapsto 2]$. After
  simplification, $\sigma_2$ is a concrete, well-formed state that can
  be written as
  $\sigma_2=[x_0\mapsto 45,\,Y_0\mapsto 3,\,w_0\mapsto 42,\,x_1\mapsto
  2,\,Y_1\mapsto 2]$.
\end{example}

Henceforth, all states are assumed to be well-formed.
\emph{We always assume states to be simplified by propagation of
  concrete values}, similar to $\sigma_2$ in
Example~\ref{exmpl:wf-state}, otherwise non-symbolic variables might
occur as values.  Symbolic states, i.e.\ states with symbolic
variables, are close to states in symbolic execution as used in path
exploration \cite{Godefroid12}, but simplified and with an additional
restriction (Definition~\ref{def:wellformed-trace} below).

\subsection{Evaluation}
\label{sec:evaluations}

The evaluation of expressions in the context of a symbolic state
reduces known variables to their values and keeps the symbolic
variables inside the expression, reducing an expression as much as
currently possible. The evaluation function is defined as follows:

\begin{definition}[Evaluation Function]
  Let $\sigma$ be a symbolic state.  The evaluation function
  $\valDnoargs: Exp\rightarrow Exp$ for expressions in the context of
  $\sigma$ is defined inductively:
  \[
    \begin{array}{r@{\,}c@{\,}ll}
      \valD{x} & = & \left\{
                     \begin{array}{ll}
                       x & \text{ if } \sigma(x)=*\\
                       \sigma(x) &\text{ otherwise}
                     \end{array} \right.\\
      \valD{v} & = & v\\
      \valD{e_1 \:op\: e_2} & = & \left\{
                                  \begin{array}{ll}
                                    \concreteop{\valD{e_1}}{\valD{e_2}} & \text{ if } \valD{e_1}\in \mathit{Val} \text{ and } \valD{e_2}\in \mathit{Val}\\[.5ex]
                                    \valD{e_1} \: op\:  \valD{e_2} & \text{ otherwise}
                                  \end{array}\right.
    \end{array}
  \]
\end{definition}

In the following we will ensure that $\vars(e)\subseteq\dom(\sigma)$
holds, such that $\valDnoargs$ is a total function.

\begin{example}
  Using state $\sigma_1$ of Example~\ref{exmpl:wf-state}, we evaluate
  $\valB{\sigma_1}{x_0+Y_0+Y_1}=(Y_0+42)+Y_0+Y_1$.
\end{example}

Let $\valD{\many{e}}$ denote the element-wise application of the
evaluation function to all expressions in the list $\many{e}$, and
likewise for sets.

\subsection{Traces and Events}
\label{sec:traces}

Traces are sequences over states and structured events. The presence
of events makes it easy to ensure global properties of traces via
well-formedness conditions over events.  Since states may be symbolic,
traces will be symbolic as well and it is necessary to constrain
traces by symbolic path conditions.  We start with a general
definition of events. Their specific structure will be added later.

\begin{definition}[Event Marker]
  Let $\ev{\many{e}}$ be an \emph{event marker} over expressions
  $\many{e}$.
\end{definition}

\begin{definition}[Path Condition]
  A \emph{path condition} $\pc$ is a finite set of Boolean
  expressions. If $\pc$ contains no variables, then we can assume it
  to be fully evaluated; i.e.\ it is either $\emptyset$,
  $\{\falseSem,\,\trueSem\}$, $\{\falseSem\}$, or $\{\trueSem\}$. A
  fully evaluated path condition is \emph{consistent} if and only if
  it does not contain $\falseSem$.  
\end{definition}

For any concrete state $\sigma$, $\valD{pc}$ is a path condition with
no variables that can be fully evaluated.
  
\begin{definition}[Conditioned Symbolic Trace]
  A \emph{symbolic trace} $\tau$ is defined co-inductively by the
  following rules ($\varepsilon$ denotes the empty trace):
  \[
    \begin{array}{l@{\;::=\;}l}
      \tau & \varepsilon~|~\consTr{t}{\tau}\\
      \stateorevent & \sigma~|~\ev{\many{e}}
    \end{array}
  \]
  A \emph{conditioned symbolic trace} has the form $\pc\pop\tau$, where
  $\pc$ is a path condition and $\tau$ is a symbolic trace. If $\pc$
  is consistent, we simply write $\tau$ for $\pc\pop\tau$.
\end{definition}

Traces can be finite or infinite. For simplicity, let
$\langle\sigma\rangle$ denote the \emph{singleton trace}
$\consTr{\sigma}{\varepsilon}$.  \emph{Concatenation} of two traces
$\tau_1$, $\tau_2$ is written as $\concatTr{\tau_1}{\tau_2}$ and only
defined when $\tau_1$ is finite.  The final state of a non-empty,
finite trace $\tau$ is selected with $\last(\tau)$, the first state of
a non-empty trace $\tau$ with $\first(\tau)$.

\begin{example}
  \label{exmpl:trace}
  Continuing Example~\ref{exmpl:wf-state}, we define a conditioned
  symbolic trace
  $\pc_0\pop\tau_0=\{Y_0>Y_1\}\pop \consTr{\sigma_1[x_0\mapsto
    17]}{\langle\sigma_1\rangle}$.
\end{example}

\paragraph{Sequential Composition}

It is important that traces can semantically model the sequential
composition of program statements.  Assume that $\tau_1$ is a trace of
a statement $r$ and $\tau_2$ a trace of another statement $s$.  To
obtain the trace corresponding to the sequential composition of $r$
and $s$, 
the traces corresponding to $r$ and $s$ should be concatenated, but the last state of the first trace and the first state of the second trace are generally identical. We generalize such a concatenation in case the states are not exactly the same:
 to concatenate two traces $\tau_1$ and $\tau_2$
 the first state of $\tau_2$ should be an extension of the last state of $\tau_1$, but the resulting trace should only contain the largest of the two states.
This motivates the \emph{semantic chop}
$\chopTrSem{}{}$ on traces (following \cite{NakataUustalu15}, who in
turn were inspired by interval temporal logic \cite{HMM83}):

\begin{definition}[Chop on Traces]
  Let $pc_1, pc_2$ be path conditions and $\tau_1,\,\tau_2$ be symbolic
  traces, and assume that $\tau_1$ is a non-empty, finite trace. The
  semantic chop $\chopTrSem{(\pc_1\pop\tau_1)}{(\pc_2\pop\tau_2)} $ is
  defined as follows:
  \begin{equation*}
    \chopTrSem{(\pc_1\pop\tau_1)}{(\pc_2\pop\tau_2)} =
    (\pc_1\cup\pc_2)\pop\concatTr{\tau}{\tau_2}
    ~\mbox{where}~\tau_1=\consTr{\sigma}{\tau},\
    \tau_2=\concatTr{\langle\sigma'\rangle}{\tau'}~\mbox{and}~\sigma\subseteq\sigma'\ .
  \end{equation*}
\end{definition}

For infinite and empty traces $\tau_1$, the chop operator is
undefined.  Whenever the final state of the first trace cannot be
extended to the first state of the second trace, the operator is
undefined.

The definition can be specialised to traces without path conditions in
the obvious way: $\chopTrSem{\tau_1}{\tau_2}$ is
$\chopTrSem{(\emptyset\pop\tau_1)}{(\emptyset\pop\tau_2)}$.

\paragraph{Traces with Events.}

Events will be uniquely associated with the state in a trace at which
they occurred. The events do not update the values in a state, but
they may extend a state with fresh symbolic variables. To this aim, an
event $\ev{\many{e}}$ is inserted into a trace after a state~$\sigma$
and the state is then augmented by a set $\many V$ of symbolic
variables. The notation we use for this operation is an \emph{event
  trace} $\evTrioPV{\sigma}{\many{e}}{\many V}$ of length three:
\[
  \evTrioPV{\sigma}{\many{e}}{\many V} =
  \consTr{\consTr{\sigma'}{\ev{\valB{\sigma'}{\many{e}}}}}{\langle\sigma\rangle}\ \text{ where } \sigma'=\sigma[\many V\,\mapsto\, \ast].
\]
Given a trace of the form $\consTr{\sigma}{\tau_1}$ and event
$\ev{\many{e}}$, appending the event is achieved by the trace
$\tau_1\cdot\evTrioPV{\sigma}{\many{e}}{\many V}$.  The preceding
definition ensures that events in traces are ``choppable'';
$\chopTrSem{\tau}{\evTrioPV{\sigma}{e}{\many V}}$ is well-defined
whenever $\last(\tau)=\sigma$.  If $\many V$ is empty then the state
is unchanged, in this case we omit the set of symbolic variables:
\[
\evTrioP{\sigma}{\many{e}} = \evTrioPV{\sigma}{\many{e}}{\emptyset}
\]

\begin{example} \label{exmpl:trace-event} Event traces can be inserted
  at the middle of a trace.  To insert an event $\ev{Y_0}$ that does
  not introduce symbolic variables at $\sigma_1$ to trace $\tau_0$ in
  Example~\ref{exmpl:trace}, we use the event trace
  $\evTrioP{\sigma_1}{Y_0} =
  \consTr{\consTr{\sigma_1}{\ev{\valB{\sigma_1}{Y_0}}}}{\langle\sigma_1\rangle}=\consTr{\consTr{\sigma_1}{\ev{Y_0}}}{\langle\sigma_1\rangle}
  $.
  This results in the trace:
  $\tau_2=\{Y_0>Y_1\}\pop\consTr{\sigma_1[x_0\mapsto17]}
  {\consTr{\consTr{\sigma_1}{\ev{Y_0}}} {\langle\sigma_1\rangle}}$.
  To insert an event $\ev{Y_2}$ that introduces the symbolic variable
  $Y_2$, we use the event trace
  $\evTrioPV{\sigma_1}{Y_2}{\{Y_2\}} =
  \consTr{\consTr{\sigma_1[Y_2\,\mapsto\,
      \ast]}{\ev{Y_2}}}{\langle\sigma_1\rangle}$.
  The trace in Example~\ref{exmpl:trace} results in the trace:
  $$\tau_2'=\{Y_0>Y_1\}\pop\consTr{\sigma_1[Y_2\,\mapsto\, \ast,
    x_0\mapsto 17]} {\consTr{\consTr{\sigma_1[Y_2\,\mapsto\,
        \ast]}{\ev{Y_2}}} {\langle\sigma_1\rangle}}.$$
\end{example}

\paragraph{Well-Formed Traces.}

Similar to the values of well-formed states, the expressions in events
and path conditions of a well-formed trace should only contain
symbolic variables. This requires all states in a trace to agree upon
which variables are symbolic. We also impose an additional
well-formedness condition on events: any event occurring in a trace
must be immediately preceded and followed by the same state $\sigma$;
i.e., it can be obtained by inserting an event trace at $\sigma$.
This implies that a trace always starts and ends with a state, never
with an event. Well-Formed traces are formalised as follows:

\begin{definition}[Well-Formed Trace]
  \label{def:wellformed-trace}
  Let $\pc\pop\tau$ be a conditioned symbolic trace and let
  $V=\bigcup_{\sigma\in\tau} \symbvars{\sigma}$. The trace
  $\pc\pop\tau$ is well-formed if the following conditions hold:
  \begin{eqnarray}
    & \forall \sigma \in \tau.\, \sigma\text{ is well-formed} \label{eq:wft0}\\
    & \forall \sigma \in \tau.\, (\dom(\sigma) \setminus \symbvars{\sigma})\cap V= \emptyset \label{eq:wft1}\\
    & \vars(\pc)\subseteq V \label{eq:wft2} \\ 
    & \forall \ev{\many e} \in \tau.\, \vars(\many e) \subseteq V \label{eq:wft3} \\
    & \forall \ev{\many e}, \tau_1, \tau_2.\,\left(\tau=\concatTr{\consTr{\ev{\many e}}{\tau_1}} {\tau_2} \implies \exists \sigma.\,\last(\tau_1) =  \sigma \land \first(\tau_2) = \sigma\right) \label{eq:wft4}
  \end{eqnarray}
\end{definition}

Equation~\eqref{eq:wft1} states that a variable that is symbolic in a
state cannot be non-symbolic in another state of the
trace. Equations~\eqref{eq:wft2}--\eqref{eq:wft3} ensure that any variable
to occur in a path condition or event is a symbolic variable from some
state. Equation~\eqref{eq:wft4} guarantees that any event in a trace is
surrounded by identical states.

\begin{example}\label{exmpl:wf-trace}
  The traces $\tau_2$ and $\tau_2'$ in Example~\ref{exmpl:trace-event}
  have well-formed states. Only symbolic variables occur in path
  conditions and events, and the events were added by inserting event
  traces. Hence, they are well-formed.
\end{example}

\begin{definition}[Concrete Traces]
  A \emph{concrete trace} is a well-formed trace containing only
  concrete states and events, as well as a fully evaluated, consistent
  path condition.  We use $sh$ for concrete traces, where the letters
  stand for ``shining'' trace.
\end{definition}

\begin{example}
  \label{trace-concr}
  Let
  $\sigma_3=[x_0\mapsto 45,\,Y_0\mapsto 3,\,w_0\mapsto 42,\,x_1\mapsto
  2,\,Y_1\mapsto 2]$, then a concrete trace of length two is
  $\tau_3=\consTr{\sigma_3[x_0\mapsto
    17]}{\langle\sigma_3\rangle}$. The path condition $Y_0>Y_1$ is
  evaluated to $3>2=\trueSem$ in $\sigma_3$ and thus consistent.
\end{example}

Observe that $\tau_3$ can be obtained from $\tau_0$ in
Example~\ref{exmpl:trace} by a suitable instantiation of its symbolic
variables. The precise definition of this operation is the purpose of
the following subsection.

\subsection{Concretisation}

A concretisation mapping is defined relative to a state. It associates
a concrete value to each symbolic variable of the state.

\begin{definition}[State Concretisation Mapping]
  A mapping $\rho: \mathit{Var} \to \mathit{Val}$ is a
  \emph{concretisation mapping} for a state $\sigma$ if
  $\dom(\rho)\cap\dom(\sigma)=\symbvars{\sigma}$.
\end{definition}

A concretisation mapping $\rho$ may additionally define the value of variables that
are not in the domain of $\sigma$.

\begin{example}
  \label{exmpl:concr-mapping}
  Consider $\sigma_1$ of Example~\ref{exmpl:wf-state} where
  $\symbvars{\sigma_1}=\{Y_0,\,Y_1\}$. We define the concretisation
  mapping $\rho_1=[Y_0\mapsto 3,\,Y_1\mapsto 2]$ for $\sigma_1$.

  Let $\sigma_4=[X\mapsto\ast,\,z\mapsto 3]$. Then a concretisation
  mapping must give a value to $X$, but \emph{not} to $z$, for
  example, $\rho_2=[X\mapsto 2,\,Y\mapsto 0]$ is a concretisation
  mapping for $\sigma_4$.
\end{example}

A symbolic state can be concretised using the values assigned to its
symbolic variables in a concretisation mapping to evaluate the
symbolic expressions.

\begin{definition}[State Concretisation]
  \label{def:concretisation.state}
  Let $\sigma$ be a state and $\rho$ a concretisation mapping for
  $\sigma$. The concretisation of $\sigma$ with $\rho$ is defined as
  follows:
  \[
    \mapComp{\rho}{\sigma}=\rho \cup \{x\mapsto
    \valB{\rho}{\sigma(x)}\mid\, x\in
    \dom(\sigma)\setminus\dom(\rho)\}\ .
  \]
\end{definition}

\begin{example}
  \label{exmpl:state-concr}
  Continuing Examples~\ref{trace-concr} and~\ref{exmpl:concr-mapping},
  we obtain $\mapComp{\rho_1}{\sigma_1}=\sigma_3$.
\end{example}

To concretise a symbolic trace $\pc\pop\tau$, we must apply a
concretisation mapping $\rho$ to all states and events of the
trace. This means the domain of $\rho$ must comprise all symbolic
variables that occur in the trace.  This is the case for well-formed
traces, whose states must agree on which variables are symbolic
(equation~\eqref{eq:wft1} in Definition~\ref{def:wellformed-trace}).

\begin{definition}[Trace Concretisation Mapping]
  A mapping $\rho$ is a \emph{trace concretisation mapping} for~$\tau$
  if it is a concretisation mapping for all the states in $\tau$.  We
  say that \emph{$\rho$ concretises $\tau$}.
\end{definition}

\begin{definition}[Symbolic Trace Concretisation]
  \label{def:concretisation.symbolic.trace}
  Let $\pc\pop\tau$ be a well-formed trace and $\rho$ a concretisation
  mapping for $\tau$.  The \emph{concretisation} $\rho(\pc\pop\tau)$
  of $\pc\pop\tau$ is obtained as $\rho(\pc)\pop\rho(\tau)$, where
  $\rho(\pc)=\valB{\rho}{\pc}$ and $\rho(\tau)$ is defined as follows:
  \[\rho(\stateorevent_1\cdots\stateorevent_n\cdots) = \rho(\stateorevent_1)\cdots\rho(\stateorevent_n)\cdots \]
  \[\rho(\sigma) = \mapComp{\rho}{\sigma}\]
  \[\rho(\ev{\many{e}}) = \ev{\valB{\rho}{\many{e}} } \]
\end{definition}

\begin{proposition}\label{prop:concretisation}
  The concretisation of a state is a concrete state. The
  concretisation of a well-formed trace is a concrete, well-formed
  trace.
  Any concrete state or trace is a concretisation of itself.
\end{proposition}

\begin{example}
  \label{exmpl:trace-concr}
  We continue to use the symbolic trace $\pc_0\pop\tau_0$ of
  Example~\ref{exmpl:trace} and concretise it with the mapping
  $\rho_1$ of Example \ref{exmpl:concr-mapping}. The resulting
  concrete trace is $\tau_3$ from Example \ref{trace-concr} with path
  condition $\{\trueSem\}$.
\end{example}

\subsection{Continuations}

To capture the local semantics of a language, below we define an
evaluation function $\valD{s}$ that evaluates a single statement $s$
in a---possibly symbolic---state $\sigma$ to a set of conditioned,
symbolic traces.  \emph{Compositional} local evaluation rules can then
be defined for each statement by adding a \emph{continuation marker}
at the end. 
Continuation markers are needed, if $s$ is a composite statement that
requires separate evaluation of its constituent parts. In particular, $s$
might not terminate, but this should not jeopardise parallel
computations when computing a sequential trace for a global system. We
avoid this issue by stopping the evaluation after a finite number of steps
by means of a continuation marker, defined as follows:

\begin{definition}[Continuation Marker]
  Let $s$ be a program statement.  The \emph{continuation marker}
  $\cont{s}$ expresses that a given trace is extended by the traces
  resulting from computing $s$.  The empty continuation, denoted
  $\cont{\zero}$, expresses that nothing remains to be
  computed.\footnote{The mnemonics of the symbol is that from an empty
    bottle nothing can be consumed.}
\end{definition}

The argument to the continuation marker is the code on which
evaluation is continued.  When atomic statements are evaluated, there
is no continuation code but only a return of control. In this case, we
end the trace with an empty continuation (see Section~\ref{sec:local}).
Local evaluation (corresponding to phase~\eqref{item:phase1} in
Section~\ref{sec:introduction}) is defined below such that for each
statement~$s$ and symbolic state $\sigma$ the result of $\valD{s}$ is
a set of conditioned, symbolic traces of the form
$\pc\pop\concatTr{\tau}{\cont{s'}}$ where $\tau$ is a \emph{finite}
trace. These traces are called \emph{continuation traces}; we denote
by $\ctraces$ the type of continuation traces.

\section{LAGC Semantics of WHILE}
\label{sec:while}

\begin{figure}[t]\centering
$\begin{array}[t]{r@{\hspace{2pt}}r@{\hspace{2pt}}l}
         s\in \mathit{Stmt}&::=&  \Simple{skip} \mid  x := e \mid
                   \Simple{if}~e~\{ ~s~\} \mid s;s \mid
                   \Simple{while}\ e\, \{\, s\, \} 
     \end{array}$
\caption{\label{fig:basiclangWhile}The syntax for statements in \Simple{WHILE}.}
\end{figure}

We define a LAGC semantics for \Simple{WHILE}, a language with basic
sequential constructs that can be seen as a kernel of imperative
programming languages \cite{apt09vscp}.  The statements of
\Simple{WHILE} consist of \Simple{skip}, assignment, conditional,
sequential composition, and \Simple{while}-loops. The syntax for
statements is given in Figure~\ref{fig:basiclangWhile}, where we
assume given standard expressions $e$. Assignment binds the value of
an expression to a variable. Conditionals and \Simple{while}-loops
depend on the value of a Boolean expression $e$. We assume the
standard semantics of this language to be known, and use it to
illustrate our trace semantics. \emph{Local rules}
(phase~\eqref{item:phase1}) unfold the traces until the next possible
scheduling point, marked by a continuation. \emph{Composition rules}
(phase~\eqref{item:phase2}) select the next trace to be unfolded at
the scheduling point; for the sequential language, there is only one
trace so the latter selection is deterministic.  Due to its
simplicity, \Simple{WHILE} merely needs state-based traces, the use of
events is covered in Section~\ref{sec:user-input}.

\subsection{Local Evaluation}
\label{sec:local}

Local evaluation rules a single statement in the context of a symbolic
state, and return a set of finite continuation traces.  Each
evaluation rule represents the execution of a single statement,
producing a set of continuations traces. The evaluation rules are
reminiscent of small-step reduction rules, but work in a denotational
setting and with a symbolic context.  We overload the symbol
$\valDnoargs$ and declare it with the type
$\valDnoargs:\mathit{Stmt}\rightarrow2^\ctraces$.

The rule for \Simple{skip} generates an empty path condition, returns
the state it was called in, and continues with the empty continuation.
The result is a set containing one singleton trace:
\evalruleN{skip}{%
  \valD{\Simple{skip}} & = & \{\emptyset\pop
  \concatTr{\langle\sigma\rangle}{\cont{\zero}}\}\ .
}

The assignment rule generates an empty path condition and a trace from
the current state $\sigma$ to a state
which updates $\sigma$  at $x$, and continues with an empty
continuation. The result is one trace of length two:
\evalruleN{assign}{%
  \valD{x := e} & = &
  \{\emptyset\pop\concatTr{\consTr{\update{\sigma}{x}{\valD{e}}}{\langle\sigma\rangle}}{\cont{\zero}}\}\ .
}

The rule for the conditional statement branches on the value of the
condition, resulting in two traces with different path conditions. The
first trace is obtained from the current state and the continuation
with the statements in the \Simple{if}-branch, and the second trace has
the empty continuation (corresponding to the empty else-branch):
\evalruleN{If}{%
  \valD{\Simple{if}~e~\{ ~s~\}} &=&\{\{\,\valD{e}=\trueSem\}\pop \concatTr{\langle\sigma\rangle}{\cont{s}},~~\{\valD{e}=\falseSem\}\pop
  \concatTr{\langle\sigma\rangle}{\cont{\zero}}\,\}\ .
}

The semantics of \Simple{while} is similar to the semantics for the
conditional statement. This gives us the opportunity to illustrate
that the semantics of a statement can easily be expressed in terms of
the semantics of another statement, without having to expose
intermediate states:
\evalruleN{While}{%
  \valD{\Simple{while}~e~\{~s~\}} & = & \valD{\Simple{if}~e~\{
    ~s;~\Simple{while}~e~\{~s~\}\}}\ .
}

The rule for sequential composition $r;s$ is obtained by first
evaluating $r$ to traces of the form
$\pc\pop\concatTr{\tau}{\cont{r'}}$ with continuation $r'$.  The
statement $s$ can simply be added to this continuation:
\evalruleN{Seq}{%
  \valD{r;s}
  =\{\pc\pop\concatTr{\tau}{\cont{r'; s}}\mid\pc\pop
  \concatTr{\tau}{\cont{r'}} \in\valD{r}\}\ .
}

A subtle point concerns the propagation of empty continuations: $r'$
might be the empty continuation $\zero$, which should be ignored. This
behavior is captured in the following rewrite rule, which is
exhaustively applied to statements occurring inside continuations:
\[
  s; s' \rightsquigarrow s'  \text{ if }s=\zero\ .
\]

This rewrite rule reflects that the empty continuation is an identity
element for sequential composition.  Similar rewrite rules will be
added to handle identity elements for other composition operators in
the sequel.

\begin{example}\label{exmpl:seq}
  Consider the sequential statement $s_{seq}=(x:=1;\ y:=x+1)$. To
  explain how its evaluation is performed, we start from an arbitrary
  symbolic state $\sigma$ (to be instantiated later by a composition
  rule).
  The equation for sequential composition yields
  \begin{equation}
    \valD{s_{seq}}  =
    \{\emptyset\pop\concatTr{\consTr{\update{\sigma}{x}{1}}{\langle\sigma\rangle}}{\cont{y:=x+1}}\}\ .
    \label{eq:simple.val1}
  \end{equation}
  To perform this evaluation, we need the result of evaluating the
  first assignment in the context of~$\sigma$:
  \[
      \valD{x := 1}  =
      \{\emptyset\pop\concatTr{\consTr{\update{\sigma}{x}{1}}{\langle\sigma\rangle}}{\cont{\zero}}\}\ .
  \]
\end{example}

In the composition rules it might be necessary to evaluate the empty
continuation, so we must define it. It is simply the empty set of
traces:
\evalruleN{empty}{%
  \valD{\zero} =\{\}\ .}

\begin{proposition}\label{prop-concretestates}
  Given a concrete state~$\sigma$ and a program~$s$ such that
  $vars(s)\subseteq\dom(\sigma)$, then $\valD{s}$ is a set of concrete
  continuation traces of the form
  $\pc \pop \concatTr{\tau}{\cont{s'}}$.
  There is exactly one continuation trace with a consistent path
  condition.
\end{proposition}

\subsection{Trace Composition}
\label{sec:trace-composition}
Local traces are composed into concrete global ones. As \Simple{WHILE}
is sequential and deterministic, we expect to obtain exactly one
trace, provided that the execution starts in a concrete state that
assigns values to all the variables of a
program. Proposition~\ref{prop-concretestates} ensures that all local
evaluation rules produce concrete traces in this case.

The task of the composition rule for \Simple{WHILE}-programs is to
repeatedly evaluate one statement at a time in a concrete state until
the next continuation, then stitch the resulting concrete traces
together.  Given a concrete trace $sh$ with final state $\sigma$ and a
continuation $\cont{s}$, we evaluate $s$ starting in $\sigma$.  The
result is a set of conditioned traces from which one trace with a
consistent path condition and a trailing continuation $\cont{s'}$ is
chosen.\footnote{For a deterministic language like \Simple{WHILE},
  there is exactly one trace, but the rule is designed to work for the
  non-deterministic extension below as well.}
The chosen trace $\tau$ is chopped with the given trace~$sh$.
Afterwards, the composition rule can be applied again to the extended
concrete trace and $\cont{s'}$.
\gcondrule{SeqExc-1}{
 \sigma = \last(sh) \quad 
\pc \pop \concatTr{\tau}{\cont{s'}} \in \valD{s} \quad
\pc \text{ consistent}
 }{
 sh, \cont{s} \to
  \chopTrSem{sh}{\tau}, \cont{s'}
}

The rule assumes that $s$ is evaluated to a concrete trace so that
$\chopTrSem{sh}{\tau}$ stays concrete. At this stage, symbolic traces
do not yet figure.  This works as long as
Proposition~\ref{prop-concretestates} ensures that $\pc\pop\tau$ is a
concrete trace, however, in general the proviso
$vars(s)\subseteq\dom(\sigma)$ does not hold. There are several
approaches to address this. For example, it is sufficient to consider
the \emph{non-initialised} variables of $s$, instead of all of
them. We avoid a lengthy definition to address this technicality by
simply assuming that $\last(sh)$ defines all variables of $s$. This
can be easily achieved by initialising all variables to default values
at program start, as done below.

\subsection{Global Trace Semantics}

Let $sh,\cont{s}\stackrel{*}{\to} sh',\cont{s'}$ denote the transitive
closure of applying rule~\eqref{eq:ruleSeqExc-1}, expressing that $sh',\cont{s'}$ can
be reached from $sh,\cont{s}$ in zero or more steps.

\begin{definition}[Program Semantics]
  \label{def:prog.semantics}
Given a program \Simple{s} and a state $\sigma$, let
 $$sh_0,\,\cont{s_0}\rightarrow sh_1,\,\cont{s_1}\rightarrow\cdots$$
 be a maximal sequence obtained by the repeated application of
 rule~\eqref{eq:ruleSeqExc-1}, starting from
 $\langle \sigma\rangle,\cont{s}$. If the sequence is finite, then it
 must have the form\footnote{Observe that $sh,\cont{\zero}$ is the end
   of the execution because the evaluation of~$\cont{\zero}$ returns
   $\{\}$ such that the composition rule~\eqref{eq:ruleSeqExc-1} is no
   longer applicable. The empty continuation is obtained once the
   whole program has been evaluated.}
  \[\langle \sigma\rangle,\cont{s} \stackrel{*}{\to} sh,\cont{\zero}\ .\]
  If the sequence is infinite, let $sh=\lim_{i\rightarrow\infty}sh_i$.
  The set of all such traces $sh$ for a program $s$ starting
  from a state $\sigma$ is denoted $\traces{s,\sigma}$.
  \label{ex:sequential.global.trace}
\end{definition}

For any statement $s$, let $I_s$ be the state, where $[x\mapsto 0]$
for all $x\in vars(s)$ (for simplicity, assume all variables are of
integer type---the generalisation is obvious). To help readability, we
sometimes omit such default values from the states in the following
examples.

\begin{example}
  We apply rule~\eqref{eq:ruleSeqExc-1} to
  $\langle I_{seq}\rangle,\cont{s_{seq}}$, the program from
  Example~\ref{exmpl:seq}. To obtain the premise, we instantiate
  equation~\eqref{eq:simple.val1} with $\sigma= I_{seq}$.
  \gcondrule{SeqExc-1ex}{
     I_{seq} = \last(\langle I_{seq}\rangle) \qquad 
    \emptyset \text{ consistent}\\
    \emptyset\pop\concatTr{\consTr{\update{ I_{seq}}{x}{1}}{\langle I_{seq}\rangle}}{\cont{y:=x+1}} \in \valB{I_{seq}}{x:=1;y:=x+1} \quad
  }{
    \langle I_{seq}\rangle,\cont{x:=1;y:=x+1} \to
    \consTr{I_{seq}[x\mapsto 1]}{\langle I_{seq}\rangle},\,
    \cont{y:=x+1}
  }

  For the subsequent rule application, it is necessary to evaluate the
  program $y:=x+1$ in the continuation. Again, we do this for a
  general state, while in the rule we use
  $\sigma=I_{seq}[x\mapsto 1]$ (from now on we omit $I_{seq}$):
  \begin{equation*}
    \begin{array}[h]{c}
      \valD{y := x+1}  =
      \{\emptyset\pop\concatTr{\consTr{\update{\sigma}{y}{\valD{x+1}}}{\langle\sigma\rangle}}{\cont{\zero}}\}
    \end{array}
  \end{equation*}

  \gcondrule{SeqExc-1exb}{
    [x\mapsto 1] = \last(\consTr{ [x\mapsto 1]}{\langle I_{seq}\rangle})
    \qquad \emptyset \text{ consistent}\\
    \emptyset\pop\concatTr{\consTr{\update{ [x\mapsto 1]}{y}{\valB{[x\mapsto 1]}{x+1}}}{\langle [x\mapsto 1]\rangle}}{\cont{\zero}} \in \valB{[x\mapsto 1]}{y:=x+1}
  }{
    \consTr{ [x\mapsto 1]}{\langle I_{seq}\rangle}, \cont{y:=x+1} \to
    \consTr{ [x\mapsto 1,y\mapsto 2]}{\consTr{ [x\mapsto 1]}{\langle I_{seq}\rangle}}, \cont{\zero}
  }\smallskip

  Hence,
  $\traces{\Simple{x:=1; y:=x+1},I_{seq}}=\{\consTr{ [x\mapsto 1,y\mapsto
    2]}{\consTr{ [x\mapsto 1]}{\langle I_{seq}\rangle}}\}$.
\end{example}

\subsection{Discussion}
\label{sec:discussion}

As mentioned, as long as we start semantic evaluation in a
sufficiently initialised concrete state,
Proposition~\ref{prop-concretestates} ensures that only concrete
traces will be generated by local rules. Consequently, we have a
semantics that can be aptly called modular and compositional (exactly
one independent rule per language construct and a uniform composition
rule), but the overhead introduced with \emph{symbolic} traces is not
yet justified.

The advantages offered by symbolic traces and states are realised in
the following two sections. In Section~\ref{sec:dynamic-logic-simple} we
define a program logic for \Simple{WHILE} that employs symbolic traces
at the level of the \emph{calculus}. This close correspondence between
semantics and deduction system is the basis for an intuitive soundness
proof for the deductive system.

In Sections~\ref{sec:interleaving}--\ref{sec:case-studies} we extend
\Simple{WHILE} with a number of complex instructions, in particular,
for parallel programming. It will be seen that a wide range of
concurrency paradigms fit naturally into the LAGC semantic framework.

\section{A Program Logic and Sound Calculus for WHILE}
\label{sec:dynamic-logic-simple}

We provide a dynamic logic (DL) \cite{HKT00} as well as a calculus for
reasoning about the correctness of \Simple{WHILE}-programs that is
sound relative to our semantics.  In deductive verification dynamic
logic \cite{BRSST00,keybook} offers technical advantages over Hoare
logic \cite{Hoare69}: it is syntactically closed with respect to
first-order logic, more expressive, and cleanly separates first-order
(``logical'') from program variables \cite[p.~50]{keybook}.  Below we
define program formulas of the form
$\psi\rightarrow\tau\left[s\right]\phi$, where $\tau$ is a finite
symbolic trace, $s$ any \Simple{WHILE}-statement, $\psi$ a first-order
formula, and $\phi$ a formula that in turn may contain programs. The
intuitive meaning is that any terminating execution of $s$ continuing
a trace that concretises $\tau$ and started in a state that satisfies
$\psi$, must end in a state that satisfies $\phi$.
The modality $\left[s\right]$ corresponds to a \emph{continuation} in
the semantics of \Simple{WHILE}, represented symbolically. 

The unusual aspect of this setup is the presence of a symbolic trace
$\tau$ inside a formula.
It aligns with our locally abstract semantics, but it is also
justified, because---unlike a semantics---rule schemata in calculus
rules \emph{necessarily} deal with symbolic values: a verification
calculus aims at proving a property that holds for \emph{all} inputs
of a program, not merely for a single run.
Nevertheless, the presence of symbolic traces inside formulas may
appear as insufficiently syntactic or as an inappropriate intrusion of
the semantics into the calculus. However, efficient syntactic
representations of symbolic assignments are fairly common and well
understood: for example, the deductive verification system KeY uses
symbolic \emph{updates}\footnote{Updates can be viewed as a syntactic,
  efficient, lazy representation of symbolic traces: concrete values
  are not eagerly substituted and assignments are kept in single
  static (SSA) shape.} \cite{keybook} and in the B-method explicit
generalised substitutions play a comparable role \cite{Abrial96a}.  To
keep the calculus as general as possible, we do not commit to a
particular implementation of symbolic traces.

The calculus rules given below will symbolically execute a program $s$
in a sequent of the form $\psi\rightarrow\tau\left[s\right]\phi$ and
produce verification conditions of the form
$\Gamma\Rightarrow \tau\,\phi$, where $\phi$ is a first-order formula
and $\tau$ a symbolic trace.

\subsection{Dynamic Logic}
\label{sec:dynamic-logic}

Given a signature $\Sigma$ with typed function and predicate symbols
and a set $V$ of logical variables which is disjoint from the symbols
in $\Sigma$, let $\textbf{Terms}(\Sigma,V)$ denote the well-formed
terms over $\Sigma$ and $V$ (respecting type compatibility). Note that
the logical variables in $V$ are disjoint from program variables
$\mathit{Var}$. Unlike the latter, the logical variables can be bound
by quantifiers and do not change their value during program execution.

\begin{definition}[DL Formula]
  Let $\Sigma$ be a signature and $V$ a set of  logical
  variables disjoint from $\Sigma$. The language $\textbf{DL}(\Sigma,V)$ of formulas in
  dynamic logic is defined inductively as follows:
  \begin{enumerate}
  \item $B\in \textbf{DL}(\Sigma,V)$ if
    $B\in \textbf{Terms}(\Sigma,V)$ and the type of $B$ is Boolean
  \item
    $\neg \phi_1,\, \phi_1\land \phi_2,\, \phi_1\lor \phi_2,\,
    \phi_1\rightarrow \phi_2,\, \phi_1\leftrightarrow
    \phi_2\in\textbf{DL}(\Sigma,V)$ if
    $\phi_1,\,\phi_2\in \textbf{DL}(\Sigma,V)$
  \item
    $\exists x \cdot \phi,\,\forall x \cdot
    \phi\in\textbf{DL}(\Sigma,V)$ if $x\in V$ and
    $\phi\in\textbf{DL}(\Sigma,V)$
  \item\label{item:modality} $\left[ \Simple{s}\right]\phi
    \in\textbf{DL}(\Sigma,V)$ if $s$ is a, possibly empty,
    \Simple{WHILE}-program and $\phi\in\textbf{DL}(\Sigma,V)$
  \end{enumerate}
  We often omit the signature and variable set of
  $\phi\in\textbf{DL}(\Sigma,V)$ and simply write
  $\phi\in\textbf{DL}$.
\end{definition}

There is a subtle point about empty programs. As explained above, the
program $s$ in the modality can be viewed as the continuation of the
current trace. The rules of the calculus defined below will
symbolically execute $s$ from left to right, until the program
remaining to be executed is empty. For this reason, we allow empty
programs, denoted $\emp$, in
clause~\eqref{item:modality}. The empty program is the identity
  element for the composition operators (i.e., $\emp;s=s;\emp=s$).
However, the continuation with the empty program is not the same as
the empty continuation; it requires a semantics that amounts to an
empty trace, which is reflected in Definition~\ref{def:sat} below.

\begin{definition}[Substitution]\label{def:subst}
  Given a language $\textbf{DL}(\Sigma,V)$ and a set of logical
  variables $V'=\{x_1,\ldots,x_n\}$ such that $V'\subseteq V$, a
  \emph{ substitution} $[x_1/t_1,\ldots,x_n/t_n]$ is a function
  $V'\rightarrow \textbf{Terms}(\Sigma,V)$ which associates with every
  $x_i\in V'$ a type-compatible term
  $t_i\in \textbf{Terms}(\Sigma,V)$.
\end{definition}

Denote by $\phi[x_1/t_1,\ldots,x_n/t_n]$ the application of a
substitution $[x_1/t_1,\ldots,x_n/t_n]$ to a logical formula
$\phi\in \textbf{DL}(\Sigma,V)$.  Observe that applying a substitution
removes occurrences of \emph{logical variables} in a formula and does
not affect programs.  The application of the substitution has a
straightforward inductive definition over $\phi$ (omitted here).

We write $\tau \models \phi$ and $\sigma\models \phi$ to
denote\footnote{To simplify presentation, we assume a fixed standard
  interpretation of the symbols in $\Sigma$. It is straightforward to
  accommodate undefined symbols and relativise satisfiability to
  first-order models.} that a formula $\phi$ is valid for a
\emph{concrete} trace $\tau$ and in a \emph{concrete} state $\sigma$,
respectively (equivalently, $\tau$ and $\sigma$ satisfy $\phi$).
Formally, satisfiability can be defined as follows:

\begin{definition}[Satisfiability of DL Formulas]\label{def:sat}
  Let $\phi\in\textbf{DL}(\Sigma,V)$ be a DL formula,
  $B\in \textbf{Terms}(\Sigma,V)$ a Boolean term, $\sigma$ a concrete
  state and $\tau$ a concrete trace.
  $$
  \begin{array}{l}
    \sigma \models B \iff \valD{B}=\trueSem \\
    \sigma \models \neg \phi \iff \sigma \not\models \phi\\
    \sigma \models \phi_1\land \phi_2 \iff  \sigma \models \phi_1\
    \text{and}\  \sigma \models \phi_2 \quad \text{(analogous for the remaining propositional connectives)}\\
    \sigma \models \exists x \cdot \phi \iff \sigma \models \phi[x/t]\
    \text{for some substitution}\ [x/t]\ \text{where}\ t\in\textbf{Terms}(\Sigma,V)\text{, variable-free}\\
    \sigma \models \forall x \cdot  \phi \iff \sigma \models \phi[x/t]\
    \text{for all substitutions}\ [x/t]\ \text{where}\ t\in\textbf{Terms}(\Sigma,V)\text{, variable-free}\\
    \sigma\models \left[ \Simple{s}\right] \phi \iff \tau\models \phi\
    \text{for all finite}\ \tau\in\traces{\Simple{s},\sigma}\\
    \sigma\models \left[\emp\right] \phi \iff \sigma\models \phi\\
    \tau\models \phi \iff \tau\text{ finite, non-empty},\ \last(\tau)=\sigma,\text{ and }\sigma\models \phi 
  \end{array}$$
\end{definition}

\noindent
Observe that the third but last clause implies a partial correctness
semantics.

\subsection{Calculus}
\label{sec:dl-calculus}

As usual in deductive verification \cite{HaehnleHuismann19}, we define
a calculus operating on \emph{sequents}. These have the form
$\Gamma\Rightarrow\tau\;\phi$, where $\Gamma$ is a set of formulas
thought to be implicitly conjoined (called the \emph{antecedent}), the
symbol $\Rightarrow$ can be read as implication, $\tau$ is a
non-empty, finite symbolic trace, and $\phi$ is a DL formula. Without
loss of generality, we can assume that $\Gamma$ consists of a single
formula $\psi$.

\begin{definition}[Valid Sequent]
  A sequent is \emph{valid}, denoted
  $\models \psi \Rightarrow \tau\;\phi$, iff for \emph{all}
  concretisation mappings $\rho$ such that
  $\first(\rho(\tau))\models\psi$, we have that
  \[
    \last(\rho(\tau))\models\phi\enspace.
  \]
\end{definition}

To make the correspondence with the semantics of \Simple{WHILE} more
immediate, we represent path conditions as a separate set in the
antecedent of the sequents. The rules for programs are shown in
Figure~\ref{fig:calculus}. Observe that this rule set is incomplete as
it leaves open how a sequent of the form
$\Gamma, pc \Rightarrow \tau\,\phi$ should be derived, where $\phi$ is
a first-order formula. For generality, we do not commit to a specific
representation of symbolic traces and how they are applied to
first-order formulas. For example, if symbolic traces are represented
with updates as in KeY, then the rules in~\cite[Ch.~2]{keybook} can be
used. In the soundness proof below we assume that there is a sound
calculus able to derive $\Gamma, pc \Rightarrow \tau\,\phi$.

\begin{figure}[t]
  \centering
    $\begin{array}{c}
      \begin{array}{c}
        \textsc{\footnotesize (Assign)}\\
        \sigma'=\update{\sigma}{x}{\valD{e}}\\
        \last(\tau)=\sigma \qquad \Gamma, pc \Rightarrow \consTr{\sigma'}{\tau} \left[\Simple{s}\right]\phi 
        \\\hline
        \Gamma, pc \Rightarrow \tau \left[\Simple{x:=e; s}\right]\phi 
      \end{array}
      \qquad
      \begin{array}{c}
        \textsc{\footnotesize (Cond)}\\
        \last(\tau)=\sigma \\
        \Gamma, pc\cup\{\valD{e}=\falseSem\}  \Rightarrow \tau\left[\Simple{s'}\right]\phi \\
        \Gamma, pc\cup\{\valD{e}=\trueSem\} \Rightarrow \tau \left[\Simple{s; s'}\right]\phi 
        \\\hline
        \Gamma, pc \Rightarrow \tau \left[\Simple{if e \{ s \}; s'}\right]\phi 
      \end{array}
      \\\\
      \begin{array}{c}
        \textsc{\footnotesize (skip)}\\
        \Gamma, pc \Rightarrow \tau \left[\Simple{s}\right]\phi 
        \\\hline
        \Gamma, pc \Rightarrow \tau \left[\Simple{skip; s}\right]\phi 
      \end{array}
      \qquad
      \begin{array}{c}
        \textsc{\footnotesize (While)}\\
        \Gamma, pc \Rightarrow \tau \left[\Simple{if e \{ s; while e \{ s \} \}; s' }\right]\phi 
        \\\hline
        \Gamma, pc \Rightarrow \tau \left[\Simple{while e \{ s \}; s'}\right]\phi 
      \end{array}
      \qquad
      \begin{array}{c}
        \textsc{\footnotesize (Empty)}\\
        \Gamma, pc \Rightarrow \tau\,\phi 
        \\\hline
        \Gamma, pc \Rightarrow \tau \left[\emp\right]\phi 
      \end{array}
    \end{array}$
  \caption{Dynamic logic sequent rules for \Simple{WHILE}.}
  \label{fig:calculus}
\end{figure}

We denote by
$\vdash \psi\Rightarrow \tau \left[\Simple{s}\right] \phi$ that the DL
sequent $\psi\Rightarrow \tau \left[\Simple{s}\right] \phi$ can be
derived in the proof system. To prove that a postcondition
$\text{Post}$ holds after termination of program $s$ under a
precondition $\text{Pre}$, it is sufficient to derive the sequent
\[
  \text{Pre}\Rightarrow\langle\sigma_\ast\rangle\left[\Simple{s}\right]\text{Post}\enspace,
\]
where $\sigma_\ast(x)\mapsto\ast$ for all $x\in\mathit{Var}$.  By
construction, the only symbolic variables that occur in a well-formed
symbolic trace $\tau$ are introduced in the initial state of
$\tau$. Thus, all path conditions are expressed in terms of the
initial symbolic state $\sigma_\ast$.

\begin{theorem}[Soundness]\label{thm:soundness}
  Let $\tau$ be a well-formed symbolic trace and let $\phi,\,\psi$ be
  DL formulas.  If $\vdash \psi \Rightarrow \tau\;\phi$, then
  $\models\psi \Rightarrow \tau\;\phi$.
\end{theorem}

\begin{proof}
  The proof is by induction over the structure of $\phi$. We assume
  the soundness of first-order (or ``program-free'') formulas; i.e.,
  if $\vdash B$ then $\models B$ for $B\in \textbf{Terms}(\Sigma,V)$.

  \noindent
  {\bf Case \textsc{Empty}}. %
  Assume $\vdash \psi\Rightarrow \tau \left[\emp\right] \phi$.  By rule
  (\textsc{Empty}), we get $\vdash\psi\Rightarrow\tau\,\phi$.
  
  For the base case, we may assume that $\phi$ is a first-order
  formula. Then, by assumption, $\models\psi\Rightarrow\tau\,\phi$,
  which means that for any concretisation mapping $\rho$, if
  $\first(\rho(\tau))\models\psi$, then
  $\last(\rho(\tau))\models\phi$. The latter is, by
  Definition~\ref{def:sat}, equivalent to
  $\last(\rho(\tau))\models\left[\emp\right]\phi$. Since $\rho$ was
  arbitrary, this yields
  $\models\psi\Rightarrow\tau\left[\emp\right]\phi$.

  For the inductive
  case, we may assume that $\phi$ is any DL formula. Then,
  $\models\psi\Rightarrow\tau\,\phi$ follows from the induction
  hypothesis (hereafter, IH), and the rest of the proof is analogous to the base case.
    
  \smallskip
  
  \noindent
  {\bf Case \textsc{Skip}}. %
  Assume $\vdash\psi\Rightarrow\tau\left[\Simple{skip;
      s}\right]\phi$. By rule \textsc{(Skip)}, we get
  $\vdash\psi\Rightarrow\tau\left[\Simple{s}\right]\phi$, and by IH we
  have $\models\psi\Rightarrow\tau\left[\Simple{s}\right]\phi$.
  Hence, for any concretisation mapping $\rho$,
  if $\first(\rho(\tau))\models\psi$,
  then $\last(\rho(\tau))\models\left[s\right]\phi$. Then
  $\tau'\models\left[\emp\right]\phi$
  for all $\tau'\in\traces{s,\last(\rho(\tau))}$. By
  equations~\eqref{eq:ruleSeqExc-1},~\eqref{eq:Seq} and~\eqref{eq:skip}, we
  have $\traces{s,\sigma}=\traces{\Simple{skip; s},\sigma}$ for any $\sigma$.
  Therefore, for all
  $\tau''\in\traces{\Simple{skip; s},\last(\rho(\tau))}$ we have
  $\tau''\models\left[\emp\right]\phi$, and it follows that
  $\models\psi\Rightarrow\tau\left[\Simple{skip; s}\right]\phi$.

  \smallskip

  \noindent
  {\bf Case \textsc{Assign}}. %
  Assume
  $\vdash \psi \Rightarrow \tau \left[\Simple{x:=e; s}\right]\phi$.
  By rule \textsc{(Assign)}, we get
  $\vdash \psi \Rightarrow \consTr{\sigma'}{\tau}
  \left[\Simple{s}\right]\phi$ where $\last(\tau)=\sigma$ and
  $\sigma'=\update{\sigma}{x}{\valD{e}}$.
  By IH we have
  $\models \psi \Rightarrow \consTr{\sigma'}{\tau}
  \left[\Simple{s}\right]\phi$. Hence, for any concretisation mapping
  $\rho$, if $\first(\rho(\consTr{\sigma'}{\tau}))\models\psi$,
  then $\last(\rho(\consTr{\sigma'}{\tau}))\models\left[s\right]\phi$,
  i.e.\ $\rho(\sigma')\models\left[s\right]\phi$.
  By equations~\eqref{eq:ruleSeqExc-1},~\eqref{eq:Seq} and~\eqref{eq:assign}, we have
  $\traces{\Simple{x:=e; s},\sigma} =
  \chopTrSem{\{\consTr{\sigma'}{\langle \sigma\rangle}}{\tau'} \mid \tau'
  \in \traces{s,\sigma'}\}$.
  Therefore,
  $\last(\rho(\tau))\models\left[\Simple{x:=e; s}\right]\phi$ and,
  since $\rho$ was arbitrary and
  $\first(\rho(\consTr{\sigma'}{\tau}))=\first(\rho(\tau))$, we have
  $\models\psi \Rightarrow \tau \left[\Simple{x:=e; s}\right]\phi$.

  \smallskip

  \noindent
  {\bf Case \textsc{Cond}}. %
  Assume
  $\vdash \psi \Rightarrow \tau
  \left[\Simple{if}~\Simple{e}~\{\Simple{s'}\};
    \Simple{s}\right]\phi$.  By rule \textsf{(Cond)}, we get
  $\vdash \psi,\,\valD{e}=\falseSem \Rightarrow
  \tau\left[\Simple{s}\right]\phi$ and
  $\vdash \psi,\,\valD{e}=\trueSem \Rightarrow \tau \left[\Simple{s';
      s}\right]\phi$, where $\last(\tau)=\sigma$.  Correspondingly, we
  have two induction hypotheses:
  $\models \psi,\,\valD{e}=\falseSem \Rightarrow
  \tau\left[\Simple{s}\right]\phi$ (hereafter, IH1) and
  $\models \psi,\,\valD{e}=\trueSem \Rightarrow \tau \left[\Simple{s';
      s}\right]\phi$ (hereafter, IH2).  Let $\rho$ be any
  concretisation mapping and assume $\first(\rho(\tau))\models\psi$.
  Depending on the value of $\valB{\rho(\sigma)}{e}$, exactly one of
  the subcases applies.
  
  \smallskip  
  
  \textbf{Subcase IH1}. %
  We can assume $\valB{\rho(\sigma)}{e}=\falseSem$ and, by IH1,
  $\models \psi,\,\valD{e}=\falseSem \Rightarrow
  \tau\left[\Simple{s}\right]\phi$, which together gives
  $\last(\rho(\tau))\models\left[s\right]\phi$. By definition,
  $\tau'\models\left[\Simple{s}\right]\phi$ for all
  $\tau'\in\traces{\Simple{s},\last(\rho(\tau))}$. By
  Equations~\eqref{eq:ruleSeqExc-1},~\eqref{eq:Seq} and~\eqref{eq:If},
  we have
  $\traces{\Simple{s},\sigma}=\traces{\Simple{if}~\Simple{e}~\{\Simple{s'}\};
      \Simple{s},\,\sigma}$, observing that the path condition is
  consistent. Therefore,
  $\tau'\models\left[\Simple{if}~\Simple{e}~\{\Simple{s'}\};
    \Simple{s}\right]\phi$ for all
  $\tau'\in\traces{\Simple{if}~\Simple{e}~\{\Simple{s'}\};
    \Simple{s},\,\last(\rho(\tau))}$ and so
  $\models\psi\Rightarrow\tau\left[\Simple{if}~\Simple{e}~\{\Simple{s'}\};
    \Simple{s}\right]\phi$.

  \smallskip  

  \textbf{Subcase IH2}. %
  We can assume $\valB{\rho(\sigma)}{e}=\trueSem$ and proceed by a
  completely analogous argument as in the other subcase.

  \smallskip
 
 \noindent
 { \bf Case \textsc{While}}. Follows directly from the case for \textsc{Cond}.
\end{proof}

\section{Semantics for a Shared-Variable Parallel Programming
  Language}
\label{sec:interleaving}

In this and the following section we show that the LAGC semantics
naturally extends to cover advanced language constructs.  We gradually
extend \Simple{WHILE} with parallel programming constructs:
parallel execution, procedure calls, distributed memory, dynamic
process creation, and communication between procedures. To cater for
interleaved execution of parallel instructions, and to allow for the
compositional definition of semantics, the traces of the sequential
language are gradually unfolded by means of continuations.

\subsection{Local Parallelism}
\label{sec:local.parallelism}

We extend \Simple{WHILE} with a statement for parallel execution, such
that the syntax for statements from Figure~\ref{fig:basiclangWhile}
becomes:
\[
  \begin{array}{r@{\,}l@{\,}l}
    s\in \mathit{Stmt} &::=& \Simple{co}\ s \para s\ \Simple{oc}\ \ |\
                             \ldots\ .\\
  \end{array}
\]
The semantics of this statement consists in interleaving the
evaluation of the two parallel branches at the granularity of atomic
statements, as exemplified below.

\begin{example}\label{exmpl:co}
  The execution of the program
  $s_{co}=\Simple{co}\ x:=1;\ y:=x+1 \para x:=2\ \Simple{oc}$ produces
  one of three possible traces that correspond to the traces of the
  following sequential programs: $x:=1;\ y:=x+1;\ x:=2$, or
  $x:=1;\ x:=2;\ y:=x+1$, or $x:=2;\ x:=1;\ y:=x+1$.
\end{example}

\subsubsection{Local Evaluation}

The evaluation rule for the parallel execution statement \Simple{co}
$r\para s$ \Simple{oc} branches on the statement which first gets to
execute, resulting in two sets of traces where the first set contains
traces with the path condition of $r$ and the continuation of $r$ in
parallel with $s$ and the second set contains traces with the path
condition of $s$ and the continuation of $s$ in parallel with r. The
valuation rule is formalised as follows:

\evalruleN{Par}{%
  \valD{\Simple{co}\ r\para s\ \Simple{oc}} 
  & = & \{\pc_r\!\pop\!\concatTr{\tau_r}{\cont{\Simple{co}\ r'\para s\ \Simple{oc}}}\!\mid\!\pc_r\!\pop\!
  \concatTr{\tau_r}{\cont{r'}} \!\in\!\valD{r}\}\\
  & & \cup  \{\pc_s\!\pop\!\concatTr{\tau_s}{\cont{\Simple{co}\ r\para
      s'\ \Simple{oc}}}\!\mid\!\pc_s\!\pop\!
  \concatTr{\tau_s}{\cont{s'}} \!\in\!\valD{s}\}\ .
}

\noindent
As with sequential composition in rule~\eqref{eq:Seq}, we need to ensure that
empty continuations are not propagated. To this end we define the
following rewrite rule for parallel composition inside continuations
(observe that the case $s=s'=\zero$ cannot occur):
\begin{equation}  
  \Simple{co}\ s\,\para s'\ \Simple{oc} \rightsquigarrow
  \left\{\begin{array}{rl}
            s' & \text{ if }s=\zero,\,s'\neq\zero \\
            s & \text{ if }s\neq\zero,\, s'=\zero\ .
    \end{array}\right.
  \label{eq:rewrite}
\end{equation}

\subsubsection{Trace Composition}

The composition rule~\eqref{eq:ruleSeqExc-1} can be kept unchanged
(although the abstract variant in rule~\eqref{eq:ruleGlParExc-a} below can be
used as well). Its effect is that all combinations of execution
sequences of atomic statements in the parallel branches can be
generated. This behaviour corresponds to the classical interleaving
semantics (e.g., \cite{Andrews99}), as demonstrated below.

\begin{example}
  \label{exmpl:parallel}
  Consider program $s_{co}$ from Example~\ref{exmpl:co}, its evaluation is
  \[
    \valD{s_{co}} 
    \begin{array}[t]{l}
      =\{ 
      \emptyset\pop
      \concatTr{\consTr{\update{\sigma}{x}{1}}{\langle\sigma\rangle}}{\cont{\Simple{co}\ y:=x+1 ~\para~ x:=2\ \Simple{oc}}}
      \}
      \cup\\\quad
      \{ 
      \emptyset\pop
      \concatTr{\consTr{\update{\sigma}{x}{2}}{\langle\sigma\rangle}}{\cont{\Simple{co}\  x:=1;y:=x+1~\para~\zero ~ \Simple{oc}}}
      \}\\
      =\{ 
      \emptyset\pop
      \concatTr{\consTr{\update{\sigma}{x}{1}}{\langle\sigma\rangle}}{\cont{\Simple{co}\ y:=x+1 \para x:=2\ \Simple{oc}}}
      \}
      \cup\\\quad
      \{ 
      \emptyset\pop
      \concatTr{\consTr{\update{\sigma}{x}{2}}{\langle\sigma\rangle}}{\cont{ x:=1;y:=x+1}}
      \}\ ,
    \end{array}
\]

\noindent
using the following sub-evaluations:
$\valD{x:=1;y:=x+1}=\{\emptyset\pop\concatTr{\consTr{\update{\sigma}{x}{1}}{\langle\sigma\rangle}}{\cont{y:=x+1}}\}$
and
$\valD{x:=2}=\{\emptyset\pop\concatTr{\consTr{\update{\sigma}{x}{2}}{\langle\sigma\rangle}}{\cont{\zero}}\}$.
For trace composition, we first need to evaluate the continuation
$\Simple{co}\ y:=x+1 \para x:=2\ \Simple{oc}$:
  \[
    \valD{\Simple{co}\ y:=x+1 \para x:=2\ \Simple{oc}} 
    \begin{array}[t]{l}
      =\{ 
      \emptyset\pop
      \concatTr{\consTr{\update{\sigma}{y}{\valD{x+1}}}{\langle\sigma\rangle}}{\cont{\Simple{co}\ \zero ~\para~  x:=2\ \Simple{oc}}}
      \}
      \cup\\\quad
      \{ 
      \emptyset\pop
      \concatTr{\consTr{\update{\sigma}{x}{2}}{\langle\sigma\rangle}}{\cont{\Simple{co}\ y:=x+1 ~\para~  \zero\ \Simple{oc}}}\}
      \\
      =\{ 
      \emptyset\pop
      \concatTr{\consTr{\update{\sigma}{y}{\valD{x+1}}}{\langle\sigma\rangle}}{\cont{ x:=2}}\}
      \cup\\\quad
      \{ 
      \emptyset\pop
      \concatTr{\consTr{\update{\sigma}{x}{2}}{\langle\sigma\rangle}}{\cont{y:=x+1}}\}\ ,
    \end{array}
  \]
  using the following sub-evaluations:
  $\valD{y:=x+1}=\{\emptyset\pop\concatTr{\consTr{\update{\sigma}{y}{\valD{x+1}}}{\langle\sigma\rangle}}{\cont{\zero}}\}$
  and
  $\valD{x:=2}=\{\emptyset\pop\concatTr{\consTr{\update{\sigma}{x}{2}}{\langle\sigma\rangle}}{\cont{\zero}}\}$.

  To illustrate trace composition, let us consider the trace where
  statement $x:=1$ is executed first and explore the remaining
  possible traces. We start from the state $I_{co}$. The first
  application of rule~\eqref{eq:ruleSeqExc-1} results in the following
  concrete trace and continuation:
  $$\consTr{ [x\mapsto 1]}{\langle I_{co}\rangle}, \cont{\Simple{co}\
    y:=x+1 \para x:=2\ \Simple{oc}}\ .$$
  At this point, two different instances of the composition rule are
  applicable, corresponding to the two possible interleavings; i.e.,
  there is a choice between the two continuations in
  $\valB{[x\mapsto 1]}{\Simple{co}\ y:=x+1 \para x:=2\ \Simple{oc}}$
  as shown above. The first possible instance of the composition rule is:
  \gcondrule{SeqExc-1exc}{
    [x\mapsto 1] = \last(\consTr{ [x\mapsto 1]}{\langle
      I_{co}\rangle}) \qquad     \emptyset \text{ consistent}\\
  \hspace{3mm}  \emptyset\pop \concatTr{\consTr{\update{[x\mapsto
          1]}{y}{\valB{[x\mapsto 1]}{x+1}}}{\langle[x\mapsto
        1]\rangle}}{\cont{ x:=2}} \\
\hspace{-17mm} \in \valB{[x\mapsto 1]}{\Simple{co}\
      y:=x+1 \para x:=2\ \Simple{oc}}
  }{
    \consTr{ [x\mapsto 1]}{\langle I_{co}\rangle}, \cont{\Simple{co}\  y:=x+1 \para x:=2\ \Simple{oc}} \to\\
    \consTr{ [x\mapsto 1,y\mapsto 2]}{\consTr{ [x\mapsto 1]}{\langle I_{co}\rangle}}, \cont{ x:=2}
  }

  One further rule application is possible, which results in the final step:
  \[
    \begin{array}{l}
      \consTr{ [x\mapsto 1,y\mapsto 2]}{\consTr{ [x\mapsto 1]}{\langle
      I_{co}\rangle}}, \cont{ x:=2}\to\\ \consTr{[x\mapsto 2,y\mapsto
      2]}{\consTr{ [x\mapsto 1,y\mapsto 2]}{\consTr{ [x\mapsto
      1]}{\langle I_{co}\rangle}}}, \cont{\zero}\enspace .
    \end{array}
  \]

  The second possible instance of the composition rule is similar
  and yields:
  \[\begin{array}{l}
      \consTr{ [x\mapsto  2]}{\consTr{ [x\mapsto 1]}{\langle I_{co}\rangle}}, \cont{y:=x+1}
      \to\\ 
      \consTr{[x\mapsto 2,y\mapsto 3]}{\consTr{ [x\mapsto  2]}{\consTr{ [x\mapsto 1]}{\langle I_{co}\rangle}}}, \cont{\zero}\enspace .
    \end{array}
\]
  
The third possible trace is obtained analogously by starting with the
second branch. After this, only one continuation is
possible. Altogether, the following traces are obtained:
  \[
    \begin{array}{lcl}
      \traces{s_{co}, I_{co}} & =\;\{ & \consTr{[x\mapsto 2,y\mapsto 2]}{\consTr{
      [x\mapsto 1,y\mapsto 2]}{\consTr{ [x\mapsto 1]}{\langle
        I_{co}\rangle}}},\\
      & & \consTr{[x\mapsto 2,y\mapsto 3]}{\consTr{
      [x\mapsto 2]}{\consTr{ [x\mapsto 1]}{\langle I_{co}\rangle}}},\\
      & & \consTr{[x\mapsto 1,y\mapsto 2]}{\consTr{ [x\mapsto 1]}{\consTr{
      [x\mapsto 2]}{\langle I_{co}\rangle}}}\;\}\ .
    \end{array}
  \]
\end{example}

\subsubsection{Calculus}
\label{sec:calculus-localpar}

Let $ar$ and $as$ range over the atomic statements of
  \Simple{WHILE} (i.e., assignment and skip).
It is straightforward to realise a set of DL rules directly inspired
by rule~\eqref{eq:Par}. Local parallelism introduces a new composition
operator, and the rules need to cover the different cases of nested
statements.
\[
\begin{array}{c}
    \begin{array}{c}
      \textsc{\footnotesize (Par)}\\
      \Gamma,\,pc\Rightarrow\tau\left[ar;\Simple{co}\;r'\para as;\,s'\;\Simple{oc}\right]\phi
      \\
      \Gamma,\,pc\Rightarrow\tau\left[as;\Simple{co}\;ar;\,r'\para s'\;\Simple{oc}\right]\phi 
      \\\hline
      \Gamma,\,pc\Rightarrow\tau\left[\Simple{co}\;ar;\,r'\para as;\,s'\;\Simple{oc}\right]\phi 
    \end{array}
    \quad
    \begin{array}{c}\\
      \textsc{\footnotesize (Par-Empty$_1)$}\\
      \Gamma,\,pc\Rightarrow\tau\left[s\right]\phi 
      \\\hline
      \Gamma,\,pc\Rightarrow\tau\left[\Simple{co}\emp\!\!\para s\;\Simple{oc}\right]\phi 
    \end{array}
    \quad
    \begin{array}{c}\\
      \textsc{\footnotesize (Par-Empty$_2$)}\\
      \Gamma,\,pc\Rightarrow\tau\left[r\right]\phi 
      \\\hline
      \Gamma,\,pc\Rightarrow\tau\left[\Simple{co}\;r \para\!\!\emp\Simple{oc}\right]\phi 
    \end{array}\\\\
       \begin{array}{@{}c@{}}
      \textsc{\footnotesize (Par-If$_1$)}\\
\last(\tau)=\sigma \\
        \Gamma, pc\cup\{\valD{e}=\falseSem\} \Rightarrow\tau\left[\Simple{co}\;\,r'\para s\;\Simple{oc}\right]\phi 
      \\
             \Gamma, pc\cup\{\valD{e}=\trueSem\}  \Rightarrow\tau\left[\Simple{co}\;r;\,r'\para s\;\Simple{oc}\right]\phi 
      \\\hline
      \Gamma,\,pc\Rightarrow\tau\left[\Simple{co}\;\Simple{if}\ e\ \{ r\};\,r'\para s\;\Simple{oc}\right]\phi 
    \end{array}
    \qquad
    \begin{array}{@{}c@{}}
      \textsc{\footnotesize (Par-If$_2$)}\\
\last(\tau)=\sigma \\
        \Gamma, pc\cup\{\valD{e}=\falseSem\} \Rightarrow\tau\left[\Simple{co}\;\,r\para s'\;\Simple{oc}\right]\phi 
      \\
             \Gamma, pc\cup\{\valD{e}=\trueSem\}  \Rightarrow\tau\left[\Simple{co}\;r\para s;s'\;\Simple{oc}\right]\phi 
      \\\hline
      \Gamma,\,pc\Rightarrow\tau\left[\Simple{co}\;r\para \Simple{if}\ e\ \{ s\};\,s'\;\Simple{oc}\right]\phi 
    \end{array}
\end{array}
\]

The base cases are covered by rules (\textsc{Par)},
(\textsc{Par-Empty$_1$)} and (\textsc{Par-Empty$_2$)}.  Rule
(\textsc{Par}), where $ar$ and $as$ match atomic statements, splits
the proof into two cases, depending on the first executed atomic
statement in the interleaved execution. The two rules
(\textsc{Par-Empty}$_1$) and (\textsc{Par-Empty}$_2$) correspond to
rewriting with empty continuations in rule~\eqref{eq:rewrite}.
The remaining rules deal with the case when the statement inside
  the parallel statement is not atomic. We illustrate these rules with
  (\textsc{Par-If}$_1$) and (\textsc{Par-If}$_2$), which resolve a
  conditional inside the parallel statement. Similar rules are required
  for inner loops, which unfold a \Simple{while}-statement in either
  branch into a conditional, and for nested parallel statements, which split
  the proof by recursively unfolding different branches.  While the
soundness of these rules, relative to the LAGC semantics, is evident,
they are impractical because they quickly lead to path
explosion. Numerous approaches to mitigate this effect have been
proposed in the literature
(e.g.,~\cite{owicki76acta,apt80toplas,ohearn07tcs,godefroid97popl}),
but here is not the place to dwell on them.

\subsection{Atomic}
\label{sec:atomic}

We extend \Simple{WHILE} with atomic blocks to control
interleaving in the execution of parallel statements, such that
the syntax for statements from Figure~\ref{fig:basiclangWhile} becomes
\[
  \begin{array}{r@{\,}l@{\,}l}
    s\in \mathit{Stmt} &::=& \Simple{co}\ s \para s\ \Simple{oc}\
                             \mid \Simple{atomic}({st})\ \mid\ \ldots\enspace,\\
  \end{array}
\]
where $st$ is a statement without \Simple{while} loops.

\subsubsection{Local Evaluation}

The \Simple{atomic} statement protects its argument against
interleaving computations from other branches of a parallel execution
operator. A (loop-free) statement $st$ can be made to execute
atomically, i.e.\ without preemption, written $\Simple{atomic}(st)$.
Atomic execution requires interleaving points to be \emph{removed}
during trace composition.
\evalruleN{Atom}{%
  \valD{\Simple{atomic}(st)} =
  & & \{ \pc_1\cup\pc_2\pop \concatTr{\chopTrSem{\tau_1}{\tau_2}}{\cont{\zero}}
  \mid \pc_1 \pop \concatTr{\tau_1}{\cont{st'}} \in\valD{st} \land st'
  \not = \zero \ \land\\
  & & \qquad\qquad \pc_2\pop \concatTr{\tau_2}{\cont{\zero}} \in\valDP{\Simple{atomic}(st')}
  \land \sigma' = \last(\tau_1)\} \ \cup\\
  & & \{\pc\pop\concatTr{\tau}{\cont{\zero}} \mid
  \pc\pop\concatTr{\tau}{\cont{\zero}}\in\valD{st}\}
}

The main idea behind rule~\eqref{eq:Atom} is to recursively unfold the
execution of its atomic argument~$st$, while removing continuation
markers. During trace composition, this will prevent the atomic code
from being interleaved.
The rule has two cases, depending on whether the evaluation of $st$
contains a non-empty continuation or not: if the continuation is not
empty, we need to evaluate it immediately, as the execution cannot be
interrupted before the end of the atomic statement.
Note the structural similarity in the definition of the first trace
set above and the composition rule~\eqref{eq:ruleSeqExc-1}. The
difference is that the consistency check and concretisation are
deferred until the actual trace composition.

\paragraph{Discussion.}
In this execution model, the semantics $\valDnoargs$ of
non-terminating atomic statements is \emph{undefined}, which is the
reason for excluding loops.
This is a design choice and not a principal limitation: It is possible
to define the semantics of non-terminating atomic programs with
suitable scheduling events that allow to process non-terminating code
piece-wise~\cite{DHJPT17}. However, that approach is more complex than
the adopted solution and breaks with the modularity we aim at with
only one rule per syntactic construct.
A different solution is to equip \Simple{atomic} with a path condition
argument (using $\pc_1$ in the recursive call) and only keep
consistent traces. When $\sigma$ is concrete, this would suffice for
terminating loops.

Another solution, presented in Section~\ref{sec:actor.future} below,
is to define a trace composition rule with run-to-completion
semantics: atomic execution of \emph{any} statement is the default,
which can be interrupted only at explicit suspension points.

Observe that the semantics of non-terminating \emph{non-atomic}
programs is well-defined: infinite traces are produced by an infinite
number of applications of the composition rule.

\subsubsection{Trace Composition}

The trace composition rule~\eqref{eq:ruleSeqExc-1} is unchanged, but
it is worthwhile to observe how it works in the presence of
$\Simple{atomic}$. When rule~\eqref{eq:ruleSeqExc-1} processes a
continuation of the form $\Simple{atomic}(st)$, it needs to evaluate
$\valD{\Simple{atomic}(st)}$ for the concrete state
$\sigma=\last(sh)$. Rule~\eqref{eq:Atom} evaluates $s$ until it
reaches a continuation $s'$, then recursively evaluates $s'$, until
$s$ is completely executed.  No interleaving can occur during the
evaluation, which reflects the intended semantics of atomic
statements.

\begin{example}
  To illustrate the evaluation of interleaved execution with
  $\Simple{atomic}$, we modify Example~\ref{exmpl:parallel} as
  follows: let
  $s_{at}=\Simple{co}\ \Simple{atomic}(x:=1;y:=x+1) \para x:=2\
  \Simple{oc}$.  The evaluation branches into either
  $\Simple{atomic}(x:=1;y:=x+1)$ and then $x:=2$, or $x:=2$ and then
  $\Simple{atomic}(x:=1;y:=x+1)$. In either case, the final value of
  $y$ is $2$, as shown in the following evaluation:
\[
  \begin{array}[t]{l@{}l}
    \valD{s_{at}} 
    = & \{ 
        \emptyset\pop
        \concatTr{\consTr{\update{\update{\sigma}{x}{1}}{y}{\valB{\update{\sigma}{x}{1}}{x+1}}} {\langle\sigma\rangle}}{\cont{ x:=2}}
        \}
        \ \cup\\
      & \{ 
        \emptyset\pop
        \concatTr{\consTr{\update{\sigma}{x}{2}}{\langle\sigma\rangle}}{\cont{ \Simple{atomic}(x:=1;y:=x+1)}}
        \}\ .
  \end{array}
\]

\noindent
Now the trace in which $y$ has the final value $3$, is no longer
produced; starting
from state $I_{at}=I_{co}$ we have
$\traces{s_{at}, I_{at}}=\{\consTr{[x\mapsto 2,y\mapsto 2]}{\consTr{ [x\mapsto
    1,y\mapsto 2]}{\consTr{ [x\mapsto 1]}{\langle I_{at}\rangle}}},\
\consTr{[x\mapsto 1,y\mapsto 2]}{\consTr{ [x\mapsto 1]}{\consTr{
      [x\mapsto 2]}{\langle I_{at}\rangle}}}\}$.

\end{example}

\subsubsection{Calculus}
\label{sec:calculus-atomic}

Rules for \Simple{atomic} are remarkably simple, because
non-interleaved execution is the \emph{default} in the calculus while
interleaving is modeled with a case distinction.  We let
\Simple{atomic(st)} be an atomic statement which matches $ar$ and $as$
in Section~\ref{sec:calculus-localpar}, such that rule (\textsc{Par})
applies, and introduce a rule which captures that an atomic statement
is equivalent to the statement itself in a sequential setting:
\[
  \begin{array}{@{}c@{}}
  \textsc{\footnotesize (Atomic)}\\
      \Gamma,\,pc\Rightarrow\tau\left[st;s\right]\phi 
      \\\hline
      \Gamma,\,pc\Rightarrow\tau\left[\Simple{atomic}(st);s\right]\phi 
  \end{array}
\]

\subsection{Local Memory}
\label{sec:localmemory}

We extend \Simple{WHILE} with local variable declarations
by introducing a syntactic category of variable declarations
$\mathit{VarDecl}$ and add blocks to the syntax for statements from
Figure~\ref{fig:basiclangWhile}:
\[
  \begin{array}[t]{r@{\hspace{2pt}}r@{\hspace{2pt}}l}
    s\in \mathit{Stmt} &::=& \{\ d \ s\ \} \mid \ldots\\
    d\in\mathit{VarDecl} &::=& \varepsilon\ |\ x; \ d 
  \end{array}
\]
For simplicity, we omit an empty list of variable declarations in a
scope $\{\ \varepsilon\ s\ \}$ and just write $\{\ s\ \}$.

\subsubsection{Local Evaluation} 

A block $\{\ d\ s\ \}$ introduces a local variable scope such that the
program variables $d$ should only be accessible for the statement
$s$. To avoid interference among variable declarations in different
scopes, possibly introducing the same variable name, the evaluation
rule renames the variables in $d$ to unused names and adds the renamed
program variables to the state. The local program variables are
correspondingly renamed in $s$. The local evaluation rules cover a
non-empty and empty list of local variable declarations, respectively.
\evalruleN{scope1}{%
  \valD{\{\ x;d\ s\ \}}  & = & 
  \{ \emptyset \pop \concatTr{\consTr{\sigma[x' \mapsto 0]}{\langle
      \sigma\rangle}}{\cont{\{\ d\ s\subst{x}{x'}\ \}}}
 \mid x' \notin \dom(\sigma)  \}}
\evalruleN{scope2}{%
\valD{\{\  s\ \}} & = & \valD{s} 
}
Here, $s\subst{x\!}{\!x'}$ denotes the textual substitution of program variable
$x$ by $x'$ in the statement $s$.

\subsubsection{Trace Composition}

Trace composition in rule~\eqref{eq:ruleSeqExc-1} is unchanged. It
will gradually extend the state with fresh variables, bound to the
default value $0$, until the local variable declarations of the scope
have been reduced to an empty list, at which point the scope is
removed and execution can continue as normal.

\subsubsection{Calculus}
\label{sec:calculus-localmem}

The rules are straightforward. We model initialisation by an
assignment to the default value, for which a rule already
exists. Actually, this requires a slight generalisation of the
language syntax, such that variable declarations can occur anywhere in
a block, not just at the beginning.
\[
  \begin{array}{@{}c@{}}
  \textsc{\footnotesize (local$_1$)}\\ 
    \Gamma,\,pc\Rightarrow\tau\left[x':=0;\, \{s'\}\right]\phi
    \quad
    x'\text{ fresh}
    \quad
    s'=s[x\leftarrow x']
    \\\hline
    \Gamma,\,pc\Rightarrow\tau\left[\{x;\,s\}\right]\phi 
  \end{array}
  \qquad \qquad
  \begin{array}{@{}c@{}}
  \textsc{\footnotesize (local$_2$)}\\
    \Gamma,\,pc\Rightarrow\tau\left[s\right]\phi
    \\\hline
    \Gamma,\,pc\Rightarrow\tau\left[\{s\}\right]\phi 
  \end{array}
\]

\subsection{Introducing Events and Symbolic Traces}
\label{sec:user-input}

Symbolic traces can contain states in which variables may be bound to
unknown values. A simple case where this become necessary is when
\Simple{WHILE} is extended with a statement to express user input,
such that the syntax for statements from
Figure~\ref{fig:basiclangWhile} becomes
\[s\in \mathit{Stmt} ::= \Simple{input}(x) \mid\ \ldots\ .\]

For the \Simple{input}-statement,
Proposition~\ref{prop-concretestates} does not hold, because we do not
know the value of the user input $x$ when locally evaluating the
statement.  To represent these unknown values in the local semantics,
we use symbolic variables introduced in
Definition~\ref{def:symbolic-variables}.  Events can be used to
introduce a point of interaction; we define an event type $\inpEv{Y}$
that captures the introduction of a symbolic variable $Y$ by an
\Simple{input}-statement. The local evaluation rule for
\Simple{input($x$)} can be formalised as follows:
\evalruleN{input}{%
  \valD{\Simple{input}(x)} = \{
  \emptyset \pop \concatTr{\consTr{\sigma[Y\mapsto * ] [\, x\mapsto Y]}{\inpEvpV{\sigma}{Y}{\{Y\}}}
  }{\cont{\zero}}\mid Y\not\in \dom(\sigma)
  \}\ .
}

\noindent
The trace is well-formed, because the variable $Y$ inside the event is
symbolic.

The composition rule accommodates
Proposition~\ref{prop-concretestates} by using a concretisation
mapping to concretise the emitted trace. Concretisation captures the
actual user input when the statement is executed in a concrete trace.
\gcondrule{GlParExc-a}{
    \sigma = \last(sh)
\quad
 pc\pop \concatTr{\tau }{\cont{s'}} \in \valB{\sigma}{s}
\quad \rho\text{ concretises }\tau\quad\rho(pc) \text{ consistent}
}{
 sh, \cont{s} \to
  \chopTrSem{sh}{\rho({\tau})}, \cont{s'}
}

The concretisation mapping $\rho$ ensures that the resulting trace
$\chopTrSem{sh}{\rho({\tau})}$ is concrete. 

The previous composition rule~\eqref{eq:ruleSeqExc-1} is a special
case of this one, when all the states are concrete (this follows from
Proposition~\ref{prop:concretisation}).

\begin{example}
$\traces{\Simple{input}(x),I_{seq}}=\{
  \chopTrSem{\langle I_{seq}\rangle}
  {\consTr{\sigma[Y\mapsto v] [\, x\mapsto v]}{\inpEvpV{\sigma}{Y}{\{Y\}}}} 
\mid v\in \mathit{Val}\} =
\{\consTr{I_{seq}[Y\mapsto v,\, x\mapsto v]} {\consTr{
    I_{seq}[Y\mapsto v]}{\consTr{\inpEv{v}}{\langle
      I_{seq}\rangle}}}\mid v\in \mathit{Val}\}$.
All possible values $v$ may occur as an input.  Variable $Y$ has a
concrete value now; $x,\,Y$ are needed to propagate the input value
correctly, for example, in subsequent code of the form ``\Simple{if}
($x>0$) \Simple{then ...}''.
\end{example}

\paragraph{Calculus}

The calculus rule directly follows the local evaluation
rule~\eqref{eq:input} by introducing a fresh variable $Y$ to model the
unknown user input. By definition we have $\sigma_\ast(Y)=\ast$, so we
do not need this assignment.
\[
  \textsc{\footnotesize (Input)}\quad
  \begin{array}{c}
    \last(\tau)=\sigma
    \quad
    \Gamma,\, pc \Rightarrow \chopTrSem{\tau}{\inpEvpV{\sigma}{Y}{\{Y\}}} \left[x:=Y;\,s\right]\phi
    \quad
    Y\text{ fresh}
    \\\hline
    \Gamma,\, pc \Rightarrow \tau \left[\Simple{input}(x);\,s\right]\phi 
  \end{array}
\]

\noindent
Interestingly, the abstract composition rule~\eqref{eq:ruleGlParExc-a}
has no effect on the calculus, which is already symbolic, but the
soundness proof needs to be adapted.

\subsection{Procedure Calls}
\label{sec:procedure}

\begin{figure}[t!]\centering
$ \begin{array}[t]{r@{\hspace{2pt}}r@{\hspace{2pt}}l}
    P\in\mathit{Prog}&::=&\overline{M}\,\{s\}\\
    M\in\mathit{Meth}&::=&m(x)\{s\} \\
         s\in\mathit{Stmt}&::=&  \Simple{skip} \mid  x := e \mid
                   \Simple{if}~e~\{ ~s~\} \mid s;s \mid
                   \Simple{while}\ e\, \{\, s\, \} \\
&& \mid \Simple{co}\ s \para s\ \Simple{oc} 
 \mid  \Simple{atomic}(s)   
\mid   \Simple{call}(m,e) 
     \end{array}$
\caption{\label{fig:calls}Program Syntax with Procedure Calls.}     
\end{figure}

Figure~\ref{fig:calls} introduces procedure\footnote{Historically, it is
  more common to use the term \emph{procedure}, whereas in
  object-oriented programming the term \emph{method} is usual, which
  is where the LAGC concept originated from. We use both terms
  interchangeably.} calls with the statement $\Simple{call} (m, e)$: A
program is a set of method declarations $\many{M}$ with main block
$\{s\}$.
We consider $\many{M}$ to be a \emph{method table} that is implicitly
provided with a program.
Each method declaration associates a unique name $m$ with the
statement in its body. To reduce technicalities, we assume without
loss of generality that a method call has one argument and no return
value on which the caller needs to
synchronise.\footnote{Synchronisation on return values is discussed in
  Section~\ref{sec:actor.future}.}

In contrast to parallel $\Simple{co}\ s \para s\ \Simple{oc}$
statements, where parallelism is explicit in the program syntax,
procedure calls introduce implicit parallel execution: there is a
context switch between the syntactic call site and the processor,
whereupon the call is executed in parallel. This decoupling has an
important consequence: according to the local evaluation principle,
method bodies are evaluated independently of the call context. This
necessitates a new composition rule that starts the execution of a
called method. We also need to ensure that a method is not executed
before it is called, requiring the introduction of suitable
\emph{events}.

\subsubsection{Local Evaluation}
\label{sec:procedure.local.evaluation}
The new events are $\textit{invEv}(m, \valD{e})$ and
$\textit{invREv}(m, \valD{e})$ for a given state $\sigma$. These
events denote the \emph{invocation} and the \emph{activation} (also
called the \emph{invocation reaction}) of a method $m$ with argument
$e$, respectively.  Recall from Section~\ref{sec:traces} that events
inserted into a trace at a state $\sigma$ are preceded and succeeded
by that state.  The local evaluation rule~\eqref{eq:Call} for a call
to $m$ with argument $e$ has an empty path condition, inserts an
\emph{invocation} event $\invEv{}{m,\valD{e}}$ into the trace, and
continues with the empty continuation. This means that the call is
\emph{non-blocking}: the code following the call could be executed
immediately, if scheduled.  The rule is formalised as follows:
\evalruleN{Call}{\valD{\Simple{call}(m,e)}
  =\{\emptyset\pop\concatTr{\invEv{\sigma}{m,\valD{e}}}{\cont{\zero}}\}\ .
}

\subsubsection{Trace Composition:  Concrete Variant}
\label{sec:trace-comp-concr}

The concurrency model so far is still very simple: method invocations
and processors are anonymous, there is no way to distinguish between
two calls to the same method: these execute identical code and all
executions can be interleaved. For this reason, it is sufficient to
represent continuation candidates for the next execution step as a
multiset $q$.  We denote with ``$\multiunion$'' the \emph{disjoint}
multiset union.  We add a rewrite rule to simplify empty continuations
in the multisets of tasks; this rule is applied exhaustively when
adding tasks in multisets.
\[
  q \multiunion \{ \cont{\zero} \} \rightsquigarrow q .
\]

The judgment representing trace extension by one step
has the form: $sh,\,q \to sh',\, q'$ where $q$, $q'$ are multisets of
continuations of the form $\cont s$.

\gcondrule{GlParExc-1}{
  \sigma = \last(sh) \quad 
  \pc \pop \concatTr{\tau}{\cont{s'}} \in \valD{s} 
  \quad
  \pc \text{ consistent}
}{
  sh, q \multiunion \{ \cont{s} \} \to
  \chopTrSem{sh}{\tau}, q \multiunion \{ \cont{s'} \}
} 

The difference between rule~\eqref{eq:ruleSeqExc-1} and
rule~\eqref{eq:ruleGlParExc-1} is that we no longer commit to one
possible continuation, because the continuation code is now located in
a method.  Therefore, the rule selects and removes one matching
continuation from the pool $q$, executes it to the next continuation
marker, exactly like rule~\eqref{eq:ruleSeqExc-1}, and then puts the
code remaining to be executed back into $q$.  We also need a new
composition rule that adds method bodies to the pool $q$.
\gcondrule{GlParExc-2}{
m(x)\{s\} \in \many{M} 
 \quad
    \sigma = \last(sh)
\quad
  \wf{\chopTrSem{sh}{\invREv{\sigma}{m,v}}}
\quad y \not\in\dom(\sigma)
}{
  sh, q \to
  \chopTrSem{sh}{\consTr{\update{\sigma}{y}{v}}{\invREv{\sigma}{m,v}}}, q \multiunion \{ \cont{s\subst{x}{y}} \}
}

\noindent
The rule handles the creation of a new execution thread.  We select a
method with body $s$ from table $\many{M}$ and create a new
continuation $\cont{s\subst{x}{y}}$, where the call parameter $x$ is
substituted with a fresh variable $y$ for disambiguation. The new
continuation is added to the pool $q$. Next, we need to record that a
new method with argument $v$ has started to execute, as a consequence
of a previous method call. The existence of this call is ensured by
the premise expressing the well-formedness of the extended trace,
which requires an invocation event of the form $\invEv{}{m,v}$ to be
present in $sh$.  We extend the current trace with the
\emph{invocation reaction} event\footnote{Recall that
  $\invREv{\sigma}{m,v}$ is a triple with two copies of $\sigma$
  around $\invREv{}{m,v}$.}  $\invREv{\sigma}{m,v}$ to mark the start
of method execution and record the extension of the current state in
which the parameter $y$ has the value $v$.  It remains to formalise
well-formedness.

\subsubsection{Well-Formedness}

A trace is well-formed if its events obey certain ordering
restrictions.  Whenever a trace $sh$ is extended with an invocation
reaction event of the form $\invREv{}{m,v}$, there must be a
corresponding invocation event $\invEv{}{m,v}$ in $sh$.  This ordering
restriction can be captured by counting the number of occurrences of
both event forms in $sh$ using the comprehension expression
$\num{sh}{\ev{\many{e}}}$. In all other cases, the trace stays
well-formed when it is extended with a new event.  We formalise
well-formedness as a predicate on traces.

\begin{definition}[Well-Formedness]
  The \emph{well-formedness} of a concrete trace \textit{sh} is
  formalised by a predicate $\wf{\textit{sh}}$, defined inductively
  over the length of \textit{sh}:
  \[
    \begin{array}[t]{r@{\hspace{2pt}}r@{\hspace{2pt}}l}
      \wf{\varepsilon}  &=& true \\
      \wf{\consTr{\sigma}{sh}} &=& \wf{sh}\\
      \wf{\consTr{\invEv{}{m,v}}{sh}}& =& \wf{sh}\\
      \wf{\consTr{\invREv{}{m,v}}{sh}}& =& \wf{sh} \wedge \num{sh}{\invEv{}{m,v}} > \num{sh}{\invREv{}{m,v}}
    \end{array}
  \]
  \label{def:wellformedness-concrete-trace}
\end{definition}

The well-formedness predicate is used in the global composition rules
to ensure that only a valid concrete trace for a given program can be
generated during the trace composition. Any trace emitted by the trace
composition rules is well-formed, because it is an invariant
established each time a trace is extended. It is sufficient to define
well-formedness on concrete traces, because in the end it has to hold
for them. It follows that well-formedness is always a simple,
decidable property.

\subsubsection{Global Trace Semantics}

\begin{definition}[Program Semantics with Procedure Calls]
  \label{def:trace-semantics-proc}
  Given a finite program $P$ with a method table $\many{M}$ and a main
  block~$s_{main}$. Let
  $sh_0,\,q_0\rightarrow sh_1,\,q_1\rightarrow\cdots$\ be a maximal
  sequence obtained by the repeated application of the composition
  rules~\eqref{eq:ruleGlParExc-1}--\eqref{eq:ruleGlParExc-2}, starting
  with $\langle I_P\rangle,\,\{\cont{s_{main}}\}$. If the sequence is
  finite, then it must have the form
  \[
    \langle I_{P}\rangle,\{\cont{s_{main}}\} \to \dots \to
    sh,\emptyset 
\enspace.
  \]
  
  %
  %
  If the sequence is infinite, let $sh=\lim_{i\rightarrow\infty}sh_i$.
  In either case, $sh$ is a trace of $P$. The set of all such traces
  is denoted by $\traces{P,I_P}$.
  \label{ex:sequential.global.trace.procedure}
\end{definition}

\begin{example}
  \label{ex:call.concrete}
  Consider the following program $P$:
    $$
    \begin{array}[t]{l}
      m(x)\ \{ y := x;\, x := x+1 \}\\
      \{ \Simple{call}(m, 1);\, z := 2 \}
    \end{array}
    $$

    \noindent
    The method table for this program is
    $\many{M}=\{ m(x)\ \{ y := x;\, x := x+1 \}\}$.  The evaluation of
    the \Simple{main} method is as follows:
  \[
    \begin{array}[h]{r@{\hspace{2pt}}r@{\hspace{2pt}}l}
      \valD{\Simple{call}(m,1)}
      &=&
            \{\emptyset \pop 
      \concatTr{\consTr{\consTr{\sigma}{\invEv{}{m,1}}}{\langle\sigma\rangle}}{\cont{\zero}}\}\\
      \valD{z:=2}&=&\{\emptyset\pop\concatTr{\consTr{\update{\sigma}{z}{2}}{\langle\sigma\rangle}}{\cont{\zero}}\}\\
      \valD{\Simple{call}(m,1);z:=2}&=&\{\emptyset\pop
      \concatTr{\consTr{\consTr{\sigma}{\invEv{}{m,1}}}{\langle\sigma\rangle}}{\cont{z := 2}}\}
    \end{array}
  \]

  \noindent
  To prepare the evaluation of the method body of $m(x)$ (where
  $x$ is any integer):
  \begin{equation}
    \begin{array}{r@{\hspace{2pt}}r@{\hspace{2pt}}l}
      \valD{y := x} &=&
      \{\emptyset\pop\concatTr{\consTr{\update{\sigma}{y}{\valD{x}}}{\langle\sigma\rangle}}{\cont{\zero}}\}\\
      \valD{x := x+1} &=&
      \{\emptyset\pop\concatTr{\consTr{\update{\sigma}{x}{\valD{x+1}}}{\langle\sigma\rangle}}{\cont{\zero}}\}\\
      \valD{y:=x;x:=x+1} &=&\{\emptyset\pop\concatTr{\consTr{\update{\sigma}{y}{\valD{x}}}{\langle\sigma\rangle}}{\cont{x:=x+1}}\}
    \end{array}\label{eq:proc.call.val}
  \end{equation}
  
  Let $I_P$ be the initial state of $P$.  We consider the state where
  the statement $\Simple{call}(m,1)$ in \Simple{main} has already been
  executed, and explore the possible traces.  The concrete trace and
  the continuation pool at this point are
  $sh=\consTr{\consTr{ I_{P}}{\invEv{}{m,1}}}{\langle I_{P}\rangle}$
  and
  $q=\{\cont{z:=2}\}$,
  respectively.  Let $s_m$ be the method body of $m$.
  Both composition rules~\eqref{eq:ruleGlParExc-1}
  and~\eqref{eq:ruleGlParExc-2} are applicable and yield different
  interleavings: the execution continues either with the \Simple{main}
  method or with the creation of a new execution thread. In the
  following, we assume that the execution continues with thread
  creation, i.e.\ we apply rule~\eqref{eq:ruleGlParExc-2}:
  \gcondrule{ProcCall-1exa-1}{
    m(x) \{ s_m \} \in \many{M}\quad
    I_P=\last(\sh)\quad
    \wf{\chopTrSem{sh}{\invREv{I_P}{m,1}}}\quad w\not\in\dom(I_P)
  }{
    sh, q \to
    \chopTrSem{sh}{\consTr{[w \mapsto 1]}{\invREv{I_P}{m,1}}}, q \multiunion \{ \cont{s_m\subst{x}{w}} \}
  }

  At this point, both composition rules are again applicable, but only
  rule~\eqref{eq:ruleGlParExc-1} is useful: since only one invocation
  event is present in $sh$, well-formedness will not allow more than
  one instance of $m$ to execute.
  We use some abbreviations:
  \[
    \begin{array}[h]{r@{\hspace{2pt}}r@{\hspace{2pt}}l}
      \sh'&=&
              \chopTrSem{sh}{\consTr{[w \mapsto 1]}{\invREv{I_P}{m,1}}}
      \\
          &=&\consTr{\consTr{[w \mapsto
              1]}{\consTr{I_P}{\invREv{}{m,1}}}}{\consTr{\consTr{I_P}{\invEv{}{m,1}}}{\langle
              I_P \rangle}}
      \\
      s_m'&=&\{ s_m\subst{x}{w} \} =\{ y := w;\ w := w+1 \} \\
      q'&=&q \multiunion \{\cont{s_m\subst{x}{w}}\} =
            \{\cont{z:=2}, \cont{y : = w;\ w:=w+1}\}
    \end{array}
  \]

  \noindent
  The two continuations in $q'$ indicate two possible interleavings.
  We assume execution continues with the body of~$m$:
  \gcondrule{ProcCall-1exa-2}{
    [w \mapsto 1] = \last(sh')  \qquad \emptyset \text{ consistent}\\
    \emptyset\pop\concatTr{\consTr{\update{[w \mapsto 1]}{y}{\valB{[w \mapsto 1]}{w}}}{\langle [w \mapsto 1]\rangle}}{\cont{w:=w+1}}
    \in \valB{[w \mapsto 1]}{s_m'} 
  }{
    sh', \{\cont{z:=2}\} \multiunion \{ \cont{y : = w; w:=w+1} \}\to\\
    \chopTrSem{sh'}{\consTr{[w\mapsto 1, y \mapsto 1]}{\langle [w \mapsto 1]\rangle}}, \{\cont{z:=2}\} \multiunion \{ \cont{w:=w+1} \}
  }

  \noindent
  where $\valB{[w \mapsto 1]}{s_m'}$ is taken
  from equation~\eqref{eq:proc.call.val}. We introduce the abbreviations:
  \[
    \begin{array}[h]{r@{\hspace{2pt}}r@{\hspace{2pt}}l}
      \sh''& = & \consTr{[w\mapsto 1, y \mapsto 1]}{\consTr{\consTr{[w \mapsto
               1]}{\consTr{I_P}{\invREv{}{m,1}}}}{\consTr{\consTr{I_P}{\invEv{}{m,1}}}{\langle
               I_P \rangle}}}\\
      q''& = & 
             \{\cont{z:=2}\} \multiunion \{\cont{w:=w+1}\}
    \end{array}
  \]
  
  At this point, rule~\eqref{eq:ruleGlParExc-1} is again applicable
  and the two continuations in $q''$ indicate two possible
  interleavings.  We assume the execution continues with the
  \Simple{main} method body.
  \gcondrule{ProcCall-1exa-3}{
    [w \mapsto 1, y \mapsto 1] = \last(sh'') \qquad \emptyset \text{ consistent}\\ 
    \{\emptyset\pop\concatTr{\consTr{[w \mapsto 1, y \mapsto 1,  z \mapsto 2]}{\langle
        [w \mapsto 1, y \mapsto 1]\rangle}}{\cont{\zero}}\}
    \in \valB{[w \mapsto 1,y \mapsto 1]}{z:=2} 
  }{
    sh'',\, q''\to
    \chopTrSem{sh''}{\consTr{[w \mapsto 1,y \mapsto 1, z \mapsto 2]}{\langle [w \mapsto 1,y \mapsto 1]\rangle}}, \{ \cont{w:=w+1} \} 
  }

  At this point only one continuation is possible, which results in
  the concrete trace:
  \[
    \begin{array}[h]{l}
      \consTr{\consTr{[w \mapsto
      1]}{\consTr{I_P}{\invREv{}{m,1}}}}{\consTr{\consTr{I_P}{\invEv{}{m,1}}}{\langle
      I_P \rangle}}\\
      \quad\consTr{[w \mapsto 2, y \mapsto 1,  z \mapsto 2]}{\consTr{[w
      \mapsto 1, y \mapsto 1, z \mapsto 2]}{\consTr{[w \mapsto 1, y \mapsto
      1]}{}}}\ .
    \end{array}
  \]
   
  There are two other possible interleavings that all result in the
  same final state. If, for example, the assignment in the main method
  of $P$ were changed to $y := 2$, then traces with different final
  states would be generated.  \qed
\end{example}

\subsubsection{Trace Composition: Symbolic Variant}
\label{sec:trace-comp-proc-2}

The trace composition in rule~\eqref{eq:ruleGlParExc-2} works ``eagerly'' or
``on demand'' in the sense that a concrete trace extending $\sigma$
and a continuation containing the local parameter $y$, whose value is
fixed in $\update{\sigma}{y}{v}$, are generated on the spot. In the
presence of multiple calls to the same method, this leads to multiple
evaluation of the same code, i.e.\ the method body.  A deductive
verification calculus would avoid this and evaluate each method only
once, but symbolically \cite{keybook}. We demonstrate that abstract
traces allow semantic evaluation that works similarly by defining a
local evaluation function for \emph{methods} which combines the use of
symbolic variables from the \Simple{input}-statement in
Section~\ref{sec:user-input} with the variable renaming used for local
memory in Section~\ref{sec:localmemory}.  We first introduce a symbolic variable
$Y$ to act as a placeholder for the value of the substituted call
parameter $y$ and then substitute the call
parameter $x$ with a fresh variable $y$ for disambiguation, as done in
rule~\eqref{eq:ruleGlParExc-2}:
\evalruleN{method}{%
\multicolumn{1}{l}{\valD{m(x)\{s\}} =\qquad} \\
\qquad\{
  \emptyset \pop \concatTr{\consTr{\sigma[Y\mapsto \ast, \, y\mapsto Y]}{\invREvV{\sigma}{m,Y}{\{Y\}}}
  }{\cont{s\subst{x}{y}}}\mid y,Y\not\in \dom(\sigma)
  \}\ .
}

\noindent
We conventionally write $ y,Y\not\in \dom(\sigma)$ 
to express that no element in a list
of variables is in the domain of $\sigma$.

The concrete composition rule~\eqref{eq:ruleGlParExc-2} needs to be
adapted to a composition rule which takes the evaluation of a method
declaration and concretises the resulting symbolic trace in a
well-formed and consistent way:
\gcondrule{GlParExc-2s}{
m(x)\{s\} \in \many{M} 
\qquad
\rho\text{ concretises }\tau
\qquad
\rho(pc) \text{ consistent}\\
\sigma = \last(sh) 
\qquad 
pc\pop \concatTr{\tau }{\cont{s'}} \in \valB{\sigma}{m(x)\{s\}}
\qquad
\wf{\chopTrSem{sh}{\rho(\tau)}}
}{
  sh,\, q \to
  \chopTrSem{sh}{\rho(\tau)},\, q \multiunion \{ \cont{s'} \}
}

\noindent
The definition of well-formedness stays the same.
Definition~\ref{def:trace-semantics-proc} of trace semantics of $P$ stays
the same, except that rule~\eqref{eq:ruleGlParExc-2s} is used instead of
rule~\eqref{eq:ruleGlParExc-2}.
It is interesting to note that, while we need one more local rule (to
evaluate method declarations), the composition rule becomes simpler,
more uniform with rule~\eqref{eq:ruleGlParExc-1}, and more modular (not
referring to events) as a consequence. 


\begin{example}
  \label{ex:call.symbolic}
  Consider the program~$P$ from Example~\ref{ex:call.concrete}.  Let
  $s_m$ be the method body of $m$. The evaluation of the method
  declaration with rule~\eqref{eq:method} is as follows:
  \[
    \begin{array}[h]{l}
      \valD{m(x)\{s_m\}} \\
      \qquad =
                               \{
                               \emptyset \pop
                               \concatTr{\consTr{\sigma[W\mapsto *, w\mapsto W]}{\invREvV{\sigma}{m,W}{\{W\}}}
                               }{\cont{s_m\subst{x}{w}}\mid w,W\not\in \dom(\sigma)}
                               \}\\
 \qquad    = 
           \{
           \emptyset \pop \concatTr{\consTr{\sigma[W\mapsto *, w\mapsto W]}{\invREvV{\sigma}{m,W}{\{W\}}}
           }{\cont{y:=w;\, w:=w+1}\,\mid  w,W\not\in \dom(\sigma)
           \}\ .}
    \end{array}
  \]
\end{example}

\subsubsection{Calculus}
\label{sec:calculus-proc}

There are several ways to model procedure calls in a calculus. The
expressiveness of DL permits a simple method that suffices for the
language at hand. The basic insight is that the DL formula $\phi$
occurring in a sequent $\Gamma\Rightarrow\tau\;\phi$ may contain more
than one modality. This can be used to model the process pool $q$ as a
conjunction of indexed modal formulas in a sequent of the form
$\Gamma\Rightarrow\tau\;(\left[s_1\right]_1\phi_1\land\cdots\land\left[s_n\right]_n\phi_n)$
($i\in\mathbb{N}$), which is still well-formed.
Each modality of the form $\left[\cdot\right]_i$ in the sequent
represents one process in the pool $q$, distinguished from each other
by a unique index $i$.  All calculus rules introduced so far are
assumed to be equipped with indexed modalities.  This can be done in a
completely uniform manner.

We need a way to model multi-process execution like in
rule~\eqref{eq:ruleGlParExc-1}. Since the computations of the $s_i$
might interfere with each other, there may in the general case be many
interleavings of the $s_i$ that need to be considered.  The meta rule
schema (\textsc{Interleave}) branches into different premises,
depending on which of the $s_i$ is scheduled, thus unfolding all
possible different interleavings:
\[
  \begin{array}{@{}c@{}c@{}c@{}}
  \multicolumn{2}{c}{\textsc{\footnotesize (Interleave)}}\\
    \multicolumn{1}{l}{\textit{For all } 1 \leq i \leq n:} & &
    \multicolumn{1}{l}{\textit{Such that there is a rule}}
    \\
    \bigwedge_{j\in Z_i}\wf{\tau_j} \quad
       \Gamma,\,pc,\, \!\!\bigwedge_{j\in {Z_i}}(pc\cup pc_j \rightarrow \tau_j\,
      \left[s_j\right]_i\phi_i) \Rightarrow  \tau\,
    \left[s_i\right]_i\phi_i   
    & \hspace{1em} & \bigwedge_{j\in Z_i}\wf{\tau_j} \\
    \left\{\Gamma,\, pc\cup pc_j \Rightarrow
    \tau_j\left(\left[s_1\right]_1\phi_1\land\cdots\land\left[s_j\right]_i\phi_i\land\cdots\land\left[s_n\right]_n\phi_n\right)\right\}_{j\in {Z_i}}
    & &
    \left\{\Gamma,\, pc\cup pc_j \Rightarrow \tau_j
      \left[s_j\right]_i\phi\right\}_{j\in {Z_i}} 
    \\\cline{1-1}\cline{3-3}
    \Gamma,\, pc \Rightarrow \tau\left(\left[s_1\right]_1\phi_1\land\cdots\land[s_i]_i\phi_i\land\cdots\land\left[s_n\right]_n\phi_n\right) 
    & &                                                       
       \Gamma,\,\, pc\Rightarrow \tau\,
    \left[s_i\right]_i\phi
  \end{array}
\]

The rule schema (\textsc{Interleave}) takes the form of a context
rule, replacing a conjunct in the conclusion with a formula from which
that conjunct can be derived with an appropriate path condition (thus
generalising the rules for local parallelism in
Section~\ref{sec:calculus-localpar}). The schema unfolds the premises
for each conjunct $[s_i]_i\phi_i$ in the conclusion, such that there
is a rule in the calculus which is applicable to
$\Gamma,\, pc \Rightarrow \tau\,\left[s_i\right]_i\phi$ with
premise(s)
$\Gamma,\, pc\cup pc_j\Rightarrow\tau_j\,\left[s_j\right]_i\phi$ for
$j\in Z_i$ and constraint $\wf{\tau_j}$.\footnote{The discussion of
  the efficiency of the rules for local parallelism in
  Section~\ref{sec:calculus-localpar} also applies to rule
  \textsc{(Interleave)}.}  We use the notation
$\{ \phi_j\}_{j\in {Z_i}}$ to denote a set of formulas indexed by $j$
from the index set $Z_i$. Thus in the rule schema shown on the right,
the bracketed sequent denotes the set of sequent premises to the rule
for $\Gamma,\,\,pc\Rightarrow \tau[s_i]\phi$ (for example, rule
\textsc{(Cond)} has a set containing two premises).

The second premise of \textsc{(Interleave)} is discharged by the rule
corresponding to the schema on the right. Well-formedness becomes a
separate proof obligation for every trace $\tau_j$, that can be dealt
with by a fairly straightforward translation into first-order logic
(see \cite{DinTHJ15} for the details).

This approach rests on two assumptions about the rule displayed above
on the right: first, it must follow the required syntactic form,
second, it must be sound. For the composition of more general calculi,
it is necessary to define syntactic restrictions and soundness
constraints \cite{Kamburjan20}.

It remains to design calculus rules modeling rule~\eqref{eq:Call} and
rules~\eqref{eq:method}--\eqref{eq:ruleGlParExc-2s}.  The first is
straightforward, resulting in rule (\textsc{MtdCall}).
We can combine rules~\eqref{eq:method}--\eqref{eq:ruleGlParExc-2s}
into a single rule (\textsc{MtdRun}) by wrapping the body of $m$ into
a block and declaring its formal parameter $x$ as a local
variable. This allows us to re-use the previously declared rules. We
assume given method specifications in the form of a formula $\phi$
associated with each method. In the new conjunct, the modality must
have a fresh index. The call argument $e$ must be supplied at the time
the rule is applied. One heuristics is to look for matching invocation
events in $\tau$.
\[
  \begin{array}[b]{c}\\
  \textsc{\footnotesize (MtdCall)}\\
    \last(\tau)=\sigma
    \\
    \Gamma,\, pc \Rightarrow \chopTrSem{\tau}{\invEv{\sigma}{m,e}} \left[s\right]_i\phi
    \\\hline
    \Gamma,\, pc \Rightarrow \tau\,\left[\Simple{call}(m,e);\,s\right]_i\phi 
  \end{array}
\qquad
  \begin{array}[b]{c}
  \textsc{\footnotesize (MtdRun)}\\
    \last(\tau)=\sigma \qquad $i$\text{ fresh in }\psi
    \\
    m(x)\{s\} \in \many{M} \qquad
    \wf{\chopTrSem{\tau}{\invREv{\sigma}{m,e}}}
    \\
    \Gamma,\, pc \Rightarrow \chopTrSem{\tau}{\invREv{\sigma}{m,e}} \left(\left[\{x;\,x:=e;\,s\}\right]_i\phi\land\psi\right)
    \\\hline
    \Gamma,\, pc \Rightarrow \tau\,\psi 
  \end{array}
\]

\subsection{Guarded Statements}
\label{sec:guarded.statements}

Guarded statements are a common synchronisation mechanism in
concurrent languages \cite{dijkstra75cacm,Holzmann03,ABSFMCO10}. A
statement is preceded by a Boolean guard expression that blocks
execution of the current process, until the guard is evaluated to
true. This can be used, for example, to ensure that the result of a
computation is ready before it is used, a message has arrived, etc.
Our syntax for guarded commands is inspired by \promela; the syntax
for statements from Figure~\ref{fig:basiclangWhile} becomes (where $g$
is a Boolean expression):
\[
  \begin{array}{r@{\,}l@{\,}l}
    s\in \mathit{Stmt} &::=&  :: g; s  \mid\ \ldots\ .\\
  \end{array}
\]

\subsubsection{Local Evaluation}

A local evaluation rule for a guarded statement is straightforward to
design.  If the Boolean guard $g$ evaluates to true, the execution
continues normally. If the guard $g$ evaluates to false,
the execution is blocked until $g$ evaluates to true. Blocking is
modeled by re-scheduling the entire guarded statement in the
continuation.
\evalruleN{GuardedCommand}{%
  \valD{::  g; s} &=&
  \{\{\valD{g}=\trueSem\}\pop
  \concatTr{\langle\sigma\rangle}{\cont{s}}\},\quad
  \{\valD{g}=\falseSem\}\pop
  \concatTr{\langle\sigma\rangle}{\cont{:: g; s }}\}\}
}

\subsubsection{Trace Composition}

The trace composition rule~\eqref{eq:ruleSeqExc-1} is unchanged, but
it is interesting to observe how it works in the presence of
$:: g; s$. When rule~\eqref{eq:ruleSeqExc-1} processes a continuation
with a guarded statement, it may result in a \emph{deadlocked} trace
due to the second alternative, where from some point onward no
progress is made on some non-empty continuations in the process pool.
How to \emph{detect} and, possibly, to \emph{avoid} deadlocks is a
question outside the scope of this paper.

\subsubsection{Calculus}
\label{sec:calculus-guard}

Structurally, the local rule~\eqref{eq:GuardedCommand} for guarded
statements is very similar to the rule~\eqref{eq:If} for the
conditional (in fact, the case when the guard is true, is
actually identical). Indeed, we can design a calculus rule modeled
after \textsc{(Cond)}:
\[
  \begin{array}{c}
  \textsc{\footnotesize (GrdStmt)}\\
    \last(\tau)=\sigma \qquad
    \Gamma,\, pc\cup\{\valD{g}=\trueSem\} \Rightarrow \tau \left[s\right]\phi \qquad
     \Gamma,\, pc\cup\{\valD{g}=\falseSem\}  \Rightarrow \tau\left[::g;s\right]\phi 
    \\\hline
    \Gamma,\, pc \Rightarrow \tau \left[::g;s\right]\phi 
  \end{array}
\]

\noindent
While this rule is certainly sound, it is clearly incomplete, because
it admits infinitely many possible interleavings. To achieve
completeness, one would need to establish a suitable invariant that
holds whenever the guard is evaluated and that is strong enough to
render the proof finite (for example, by excluding a branch when the
guard is false). For details on such techniques, we refer the reader
to the literature~\cite{DinTHJ15,KDHJ20}.

\section{Semantics For A Shared-Memory Multiprocessor Language }
\label{sec:multicore}

We consider a multiprocessor extension to the programming language
of Section~\ref{sec:interleaving}. This goes some way to realize a
concurrent, object-oriented setting: We will now spawn processes, each
with its own identity and its own task queue to execute, but memory is
still shared among the processes. In addition, processes may exchange
values by sending and receiving messages, which have their own
identity and thus provide the capability to impose ordering
constraints among them.  Tasks are method invocations similar to the
ones defined in the previous section. In particular, we use the second
variant of the semantics, defined in
Section~\ref{sec:trace-comp-proc-2} with the
rules~\eqref{eq:method}--\eqref{eq:ruleGlParExc-2s}.

Expressions on the right-hand side of assignments are extended with
\Simple{spawn($m$, $\,e$)}, which creates a new (virtual) process
where task $m$ is called with an actual value $e$. It returns the
identifier of the newly created process. Furthermore, two statements
\Simple{send}($e_1$,$\,e_2$), to send a value $e_1$ to a process with
identifier $e_2$, and \Simple{receive}($x$,$\,e$), to bind a value
received from $e$ to the variable $x$, are introduced such that
processes may exchange values. The syntax is shown in
Figure~\ref{fig:multisyntax}.

\begin{figure}[th]\centering
$  \begin{array}[t]{r@{\hspace{2pt}}r@{\hspace{2pt}}l}
    P\in \mathit{prog}&::=&\many{M}\{s\}\\
    M\in \mathit{Meth}&::=&m(x)\{s\}\\
         s\in\mathit{Stmt}&::=&  \Simple{skip} \ | \  x := rhs \ | \
                                \Simple{if}~e~\{ ~s~\} \ | \ s;s \ | \
                                \Simple{while}\ e\, \{\, s\, \} \ | \
                                \Simple{co}\  s \para s\ \Simple{oc}\  \\
         &&  \ | \   \Simple{atomic}(st) \ | \   \Simple{call}(m,e) \ | \    \Simple{send}(e,e) \  | \ \Simple{receive}(x,e) \ \\  
       rhs\in\mathit{Rhs} &::=& e  \ | \  \Simple{spawn}(m,e)
  \end{array}$
\caption{\label{fig:multisyntax}Syntax for the Shared-Memory Multiprocessor Language.}
\end{figure}

\subsection{LAGC Semantics}
\label{sec:multicore-semantics}

In a multiprocessor setting it is necessary to keep track of what each
process does, in particular, about the origins and destinations of
messages. This is achieved by introducing process identifiers and
tagging each event in a trace with the identifier of the process that
produced it. remark that this tagging is orthogonal to the local
semantic evaluation.

\begin{definition}[Tagged Trace, Projection]\label{def:tagged-trace}
  Let $\ev{\many e}$ be an event, $\tau,\, \tau_1,\, \tau_2$ traces,
  and $p\in\PId$ a process identifier.  A \emph{tagged trace} $\tau^p$
  is defined inductively as follows:
  $$
  \begin{array}{l}
    (\ev{\many e})^{\,p}=\evp{\,p}{\many{e}}\\
    \sigma^{\,p}=\sigma\\
    (\concatTr{\tau_1}{\tau_2})^{\,p}=\concatTr{\tau_1^{\,p}}{\tau_2^{\,p}}      
  \end{array}
  $$
\end{definition}

\subsubsection{Local Evaluation}

In the following, let $\PId$ be a set of process identifiers with
typical element $p$ and $\MId$ a set of message identifiers with
typical element $i$. We start with the evaluation rule for
\Simple{send}. The rule evaluates the arguments to \Simple{send} and
creates a trace from the current state $\sigma$ with an event $\sendEv{}{v}{p}{i}$ expressing that a value $v$ is
sent to the process $p$ by a message with the identifier $i$. The
message can have any possible identifier, so the rule provides traces for
all of them:
\evalruleN{Send}{%
  \valD{\Simple{send}(e, e')} 
  & = & \{\emptyset\pop\concatTr{\sendEv{\sigma}{{\valD{e}}}{\valD{e'}}{i}}{\cont{\zero}}\}
  \mid i\in \MId\}\ . }

The local semantics of \Simple{receive} and \Simple{spawn} is more
complex, because the values received in these statements are not known
in the local evaluation.  We employ the technique using symbolic
variables introduced in Section~\ref{sec:user-input}. In contrast to
the case for method calls (Section~\ref{sec:trace-comp-concr}), we do
not have the option to compose concrete continuations ``on demand'',
because
the return values for \Simple{receive} and \Simple{spawn} are going to
be resolved later.\footnote{It would be possible to introduce symbolic
  values for message identifiers as well in the rule for \Simple{send}
  and \Simple{receive}. However, since programs do not manipulate
  message identifiers, there is no reason to do so.}  The rules
evaluate to a trace with the event
$\receiveEv{}{v}{p}{i}$, expressing that the value $v$ is received
from a process $p$ by a message with identifier $i$, and a trace with the event
$\spawnEvs{}{m,v}{p}$, expressing that a new process with
identifier $p$ is created to execute the task $m$ with the parameter
value $v$. In each rule, the symbolic variable $Y$ represents the (as
yet) unknown received value.
\evalruleN{Receive}{%
  \valD{\Simple{receive}(x, e)} 
  & = & 
  \{\emptyset\pop\concatTr{\consTr{\sigma[x\mapsto Y,\,Y\mapsto *]}{\receiveEvV{\sigma}{Y}{\valD{e}}{i}{\{Y\}}}}{\cont{\zero}} \\
  & & \multicolumn{1}{l}{\mid Y \not\in \dom(\sigma),\, i\in \MId\}}   
}
\evalruleN{Spawn}{%
  \valD{x :=\Simple{spawn}(m,e)} 
  & = &
  \{\emptyset\pop\concatTr{\consTr{{\sigma}[x\mapsto Y,\,Y\mapsto *]}{\spawnEvsV{\sigma}{m,\valD{e}}{Y} {\{Y\}} }}{\cont{\zero}} \\
  & & \multicolumn{1}{l}{\mid Y\not\in\dom(\sigma)\}}
}
All previous local rules are unchanged.

\subsubsection{Trace Composition}

The trace composition rules generalize rules~\eqref{eq:ruleGlParExc-1}
and~\eqref{eq:ruleGlParExc-2s}. There are three main differences:
First, \emph{both} composition rules now work on symbolic traces,
simply because rules~\eqref{eq:Receive}--\eqref{eq:Spawn} are symbolic.
Second, to account for having several processes, we introduce a
mapping $\Omega$ from process identifiers $p$ to multisets of
continuations $q$, i.e.\ $\Omega(p)$ gives the current task list for a
given process.  As before, we denote with $\multiunion$ the
\emph{disjoint} multiset union on task lists.  The judgment
representing trace extension by one step has the form:
$sh,\,\Omega \to sh',\, \Omega'$.
Third, we tag the generated concrete trace with the process
identifier that created it.

The following rule selects a processor $p$ with non-empty task list
and from there a continuation~$s$, evaluated in $\last(sh)$. The
resulting symbolic trace $\tau$ must be concretizable such that, after
tagging, it extends $sh$ in a well-formed manner. Observe that
  $sh$ is tagged already. The remaining
continuation $\cont{s'}$ is added back to $p$'s task list.
\begin{equation}
  \inferrule
  {p\in\dom(\Omega)\qquad\Omega(p)=q\multiunion\{\cont{s}\} \\  \sigma = \last(sh) \\
    pc\pop \concatTr{\tau }{\cont{s'}} \in \valD{s}\\
    \rho\text{ concretizes }\tau \\ \rho(pc)\ \text{consistent} \\
    \wf{\chopTrSem{sh}{\rho(\tau)}^p}}     
  { sh, \Omega \to 
    \chopTrSem{sh}{\rho(\tau)^p}, \Omega[p\mapsto q\multiunion\{\cont{s'}\}]}
  \label{eq:multi.cont}
\end{equation}

The following rule is similar and follows the pattern established in
rule~\eqref{eq:ruleGlParExc-2s}. It starts to evaluate a method body
on a processor $p$ and adds the remaining continuation $\cont{s'}$ to
its task list.  If $p\not\in\dom(\Omega)$, we use the notational
convention $\Omega(p)=\emptyset$ and avoid a special operation to
create a process. The following rule thus either creates a task inside
an existing process or spawns a new process if necessary, similar to
rule~\eqref{eq:ruleGlParExc-2s}. Well-formedness ensures that the
correct number of tasks is created, relying on the fact that
$\valB{\sigma}{m(x)\{s\}}$ starts with an invocation reaction event.
\gcondrule{multi.exec}{
  m(x)\{s\} \in \many{M} 
  \qquad
  \sigma = \last(sh)
  \qquad
  pc\pop \concatTr{\tau }{\cont{s'}} \in \valB{\sigma}{m(x)\{s\}}\\
  \rho\text{ concretizes }\tau\qquad
  \rho(pc) \text{ consistent}
  \qquad
  \wf{\chopTrSem{sh}{\rho(\tau)^p}}
}{
  sh, \Omega \to
  \chopTrSem{sh}{\rho(\tau)^p} , \Omega[p\mapsto\Omega(p)\multiunion \{ \cont{s'} \}]
}

\subsubsection{Well-formedness}
We need to establish the basic properties of sending and receiving as well
as to ensure that each process has the right number of tasks.
\begin{definition}[Well-formedness]
  \label{def:well-formedness-multi}
  The following rules are \emph{added} to the rules in
  Definition~\ref{def:wellformedness-concrete-trace}, except the final one,
  which \emph{replaces} the final rule in
  Definition~\ref{def:wellformedness-concrete-trace}.
  The first rule ensures no two messages will be sent with the same
  message identifier, the second one states that each message is
  received only once. Similarly, the third rule guarantees that each
  spawn allocates a new process identifier.  The final rule
  reflects the fact that the code executed on a processor $p$ can be
  the reaction to either a call or a spawn event, i.e.\
  $\textit{spawnEv}^{\,p}$ acts as a particular invocation event.
  \[
    \begin{array}[t]{r@{\hspace{2pt}}r@{\hspace{2pt}}l}
      \wf{\consTr{\sendEvp{}{\,p}{e}{p'}{i}}{sh}}  &=&  \wf{sh} \land \not\exists p'', e', p'''.\, \sendEvp{}{\,p''}{e'}{p'''}{i} \in sh \\
      \wf{\consTr{\receiveEvp{}{\,p}{e}{p'}{i}}{sh}}  &=&  \wf{sh}  \land \not\exists p'', e', p'''.\, \receiveEvp{}{\,p''}{e'}{p'''}{i} \in sh  \qquad\qquad~\\
      \wf{\consTr{\spawnEvp{}{\,p}{m,v}{p'}}{sh}}&=&\wf{sh} \land \not\exists p'', m',v'.\,\spawnEvp{}{\,p''}{m',v'}{p'}\in sh\\
      \wf{\consTr{\invREvp{}{\,p}{m,v}}{sh}} &=&
                                                      \wf{sh}\, \land \\
       \multicolumn{3}{@{}r}{\num{sh}{\invEvp{}{\,p}{m,v}}+ {\displaystyle \sum_{p'\in \PId}\num{sh}{\spawnEvp{}{\,p'}{m,v}{p}}} > \num{sh}{\invREvp{}{\,p}{m,v}}}
    \end{array}
  \]
\end{definition}

Message identifiers are used to avoid sending or receiving the same
message (i.e. two messages with the same identifier) twice. This rule
set is incomplete as it does not specify when a message is
received. This aspect is discussed in
Section~\ref{sec:communication-patterns} and will also use message
identifiers.  For example, the basic correctness criterion that states
that ``all received messages have been sent'' is captured by the
well-formedness condition in Definition~\ref{def:AC}; this criterion
should be used in addition to the \textit{wf} predicate defined above.

\subsubsection{Global Trace Semantics}
\begin{definition}[Program Semantics with Multiprocessors]
  Given a program $P$ with a method table $\many{M}$ and
  a main
  block~$s_{main}$, let
  $sh_0,\,\Omega_0\rightarrow sh_1,\,\Omega_1\rightarrow\cdots$ be a
  maximal sequence obtained by the repeated application of the composition
  rules~\eqref{eq:multi.cont}--\eqref{eq:rulemulti.exec}, starting
  with
  \begin{equation}
    \consTr{\spawnEvp{}{0}{main,0}{0}}{\langle I_P\rangle},\,\Omega[0\mapsto\{\cont{s_{main}}\}]\enspace.
    \label{eq:spawn-global-multi}
  \end{equation}
  If the sequence is finite, then it must have the form
  \[\consTr{\spawnEvp{}{0}{main,0}{0}}{\langle I_P\rangle},\,\Omega[0\mapsto\{\cont{s_{main}}\}]
    \to \dots \to sh,\Omega^{\zero}\enspace,\]
  where $\Omega^{\zero}(p)=\emptyset$ 
  for all $p\in\dom(\Omega)$.  If the sequence is infinite, let
  $sh=\lim_{i\rightarrow\infty}sh_i$.  In either case, $sh$ is a trace
  of $P$. The set of all such traces is denoted with
  $\traces{P, I_P}$.
  \label{ex:sequential.global.trace.multi}
\end{definition}

The spawn event at the start of a trace represents the creation of an
initial process that runs the main method. The well-formedness of
traces ensures that subsequently spawned processes will not
erroneously be assigned index $0$, which is reserved for the code
executed in the main method.

\subsection{Communication Patterns}
\label{sec:communication-patterns}

We now unleash the power of traces with events and show that several
well-known communication patterns can be defined simply by adding
well-formedness constraints to
Definition~\ref{def:well-formedness-multi}. Imposing such patterns in
the LAGC semantics is modular in the sense that all other rules are
unaffected by the enforced patterns. The well-formedness constraints
defined in this section are to be added conjunctively to the
well-formedness predicate $\mathit{wf}$ over concrete traces $sh$
(Definition~\ref{def:well-formedness-multi}) to achieve the desired
semantic properties. We recall for the definitions below that $e$ is
an expression, $p,p'\in\PId$ process identifiers, and $i\in \MId$ a
message identifier.

We first consider the basic tenet of \emph{asynchronous communication},
stipulating that received messages must have been sent.
\begin{definition}[Asynchronous
  Communication Constraint]\label{def:AC} Well-formed asynchronous
  communication is captured by the constraint
  $\mathit{wf}_\mathrm{ac}$, defined as follows :
  \[
    \wfT{ac}{\consTr{\receiveEvp{}{\,p}{e}{p'}{i}}{sh} } =
    \wf{sh} \land {\sendEvp{}{\,p'}{e}{p}{i}}\in sh
  \]
\end{definition}

Next, we consider the \emph{FIFO} (first-in, first-out) ordering
principle, which stipulates that messages between the same processors
must be received in the same order as they were sent; i.e.\ messages
may not overtake each other.  Let the notation
$\evp{\,P}{\many v}<_{sh}\anotherevp$ express that
$\evp{\,P}{\many v}$ appears before $\anotherevp$ in a
trace $sh$.
\begin{definition}[FIFO Communication Constraint]\label{def:wf-fifo}
  Well-formed FIFO communication is captured by the constraint
  $\mathit{wf}_\mathrm{fifo}$, defined as follows :
  %
\begin{multline*}
\wfT{fifo}{\consTr{\receiveEvp{}{\,p}{e}{p'}{i}}{sh} } =
   \wfT{ac}{\consTr{\receiveEvp{}{\,p}{e}{p'}{i}}{sh} } ~\land \\
    \bigl(\forall {\sendEvp{}{\,p'}{e'}{p}{i'}}\in sh.\;
    \sendEvp{}{\,p'}{e'}{p}{i'}<_{sh}\sendEvp{}{\,p'}{e}{p}{i}
    \\\implies \receiveEvp{}{\,p}{e'}{p'}{i'}\in sh\bigr)\ .
\end{multline*}
\end{definition}

A \emph{channel} is an ordered pair $(p,p')$ of processor identifiers
(not necessarily distinct) such that $p$ sends messages to $p'$ and
$p'$ receives messages sent by $p$. The messages in a channel
after a trace $sh$ can be given as a set of message identifiers,
defined by the following function:
  
\begin{definition}[Channel]
  \label{def:channel}
  Given a concrete
  trace $sh$ and a channel $ch=(p,p')$, we define by $\textit{inChannel}(sh,ch)$ the messages in $ch$
  that have been sent, but not yet received within $sh$:
  \[
    \textit{inChannel}(sh,\,(p,p'))\,=\,\{i\mid\exists
    e.\sendEvp{}{\,p}{e}{p'}{i}\in
    sh\land\receiveEvp{}{\,p'}{e}{p}{i}\not\in sh\}\ .
  \]
\end{definition}

We can detect that a fixed bound $N>0$, the \emph{capacity} of a FIFO
channel, has been reached by the following condition, which blocks a
send event if the channel is full, i.e.\ the channel already contains
$N{-}1$ messages.
\begin{definition}[Bounded FIFO Communication Constraint]
  \label{def:bfifo}
  Let $N$ be an integer.  Well-formed $N$-bound FIFO, ensured by the
  predicate $\mathit{wf}_{\mathrm{bd}(N)}$, is captured by extending
  the constraint $\mathit{wf}_{\mathrm{fifo}}$ for standard (unbound)
  FIFO by a check on the bound of the channel as
  follows. Well-formedness of a trace ending by a reception event is
  given by $\mathit{wf}_{\mathrm{fifo}}$, and an additional rule
  checks well-formedness of message sending:
  %
  \[
    \wfT{bd\textit{(N)}}{\consTr{\sendEvp{}{\,p}{e}{p'}{i}}{sh}} =
   \wf{sh} \,\land\,
    \left|\,\textit{inChannel}(sh,\,(p,p'))\,\right|<N\ .
  \]
\end{definition}

\emph{Causally ordered} (CO) messages are harder to characterize than
the previous communication patterns, because in the most simple
definition of CO found in~\cite{charron-bost95synchronous}, the
receiving events must occur in the right order provided that sending
events are causally ordered. Our event structure does not record
causal ordering. This could certainly be realized, but it would add
considerable complexity to the evaluation rules. Instead, we use the
information that is already provided by the events in a given trace,
based on an alternative characterization of
CO~\cite{charron-bost95synchronous}: ``a computation is CO if and only
if no message is bypassed by a chain of other messages''.  We first
define a chain predicate such that $\textit{Chain}(e,\,e',\,sh)$ is
true if there is a chain of messages inside $\sh$ that asserts the
causal ordering of $e$ and $e'$ induced by messages.


\begin{definition}[Communication Chain]
  The predicate \textit{Chain} holds for two events $e$ and $e'$ in a
  given trace $sh$, if there is a chain of messages asserting the causal ordering  between $e$ and
  $e'$:
  \[
    \textit{Chain}(e,\,e',\,sh)\, =\,
    \begin{array}{l}
      \exists n, e_1,\ldots,e_n,i_1,\ldots,i_n,p_1,\ldots, p_{n+1}.\\ 
      \big(\forall 1\leq k\leq n.\,
      \sendEvp{}{\,p_k}{e_k}{p_{k+1}}{i_k}\in sh\land  \receiveEvp{}{\,p_{k+1}}{e_k}{p_k}{i_k}\in sh\,\big)\,\land\\
      \big(\forall 1 < k\leq n.\,\receiveEvp{}{\,p_{k}}{e_{k-1}}{p_{k-1}}{i_{k-1}}
      <_{sh}
      \sendEvp{}{\,p_{k}}{e_{k}}{p_{k+1}}{i_k}\big)\\
      \text{such that } e <_{sh}
      \sendEvp{}{\,p_1}{e_1}{p_2}{i_1} \land
      \receiveEvp{}{\,p_{n+1}}{e_{n}}{p_{n}}{i_{n}} <_{sh}
      e'\ .
    \end{array}
  \]
\end{definition}

We now define a well-formedness constraint that checks the absence of
a message sent in the past (according to the definition of
\textit{Chain}) but not yet received. The following definition
generalises Definition~\ref{def:wf-fifo} to messages causally ordered (according to the \textit{Chain} predicate) instead of messages originating from the same process:
  
\begin{definition}[CO Communication Constraint]
 Well-formed CO communication is captured by the constraint
 $\mathit{wf}_\mathrm{co}$, defined as follows :
  %
  \begin{multline*}
    \wfT{co}{\consTr{\receiveEvp{}{\,p}{e}{p'}{i}}{sh}} =
    {\wf{sh}\land\sendEvp{}{\,p'}{e}{p}{i}}\in sh\, \land\\
    \forall
    e'',p'',i''.\,\textit{Chain}(\sendEvp{}{\,p''}{e''}{p}{i''},\sendEvp{}{\,p'}{e}{p}{i},sh)
    \\\implies \receiveEvp{}{\,p}{e''}{p''}{i''}\in\sh\ .
  \end{multline*}
\end{definition}


Finally, we have a look at \emph{synchronous call} patterns.  Several
definitions can be found in~\cite{charron-bost95synchronous}.  We
adopt a definition that constrains well-formed traces rather strongly:
any send event is immediately followed by the corresponding receive
event.
%
%
%
\begin{definition}[Strict Synchronous
  Communication Constraint]\label{def:sync}
   Well-formed strict synchronous communication is captured by the constraint
 $\mathit{wf}_\mathrm{sync}$, defined as follows :
  \[
    \wfT{sync}{\consTr{\receiveEvp{}{\,p}{e}{p'}{i}}{sh}}
    =
    \big(\wf{sh}\land sh=\chopTrSem{sh'}{\sendEvp{\sigma}{\,p'}{e}{p}{i}} \big)
  \]
\end{definition}

%
The above constraint is realized, for example, by rendez-vous channels
in \promela\ \cite{Holzmann03}, see Section~\ref{sec:promela}.

One can define more liberal notions of synchronous communication that
accept the presence of independent events between the sending and the
reception of the message. The crown
criterion~\cite{charron-bost95synchronous}, for example, can be used
for this purpose. Its definition is global over traces and requires to
define \textit{wf} in a non-inductive manner (as an \emph{invariant}
holding for the trace at any time). We omit a detailed elaboration,
because this would distract from the main point of this paper.


\section{Case Studies}
\label{sec:case-studies}

\subsection{\promela}
\label{sec:promela}

\begin{figure}[th]\centering
  \hspace*{-2em}
  \begin{minipage}[t]{.44\linewidth}
    \begin{lstlisting}[language=promela,basicstyle=\codesize\sffamily,keywordstyle=\bf\ttfamily,numbers=left,escapechar=\&,name=prom,firstnumber=auto]
chan request = [0] of {byte};&\label{line:chan}&
       
active proctype Client0() {
  request!0;
}  
active proctype Client1() {
  request!1;
}
    \end{lstlisting}        
  \end{minipage}
  \qquad
  \begin{minipage}[t]{.44\linewidth}
    \begin{lstlisting}[language=promela,basicstyle=\codesize\sffamily,keywordstyle=\bf\ttfamily,numbers=left,escapechar=\&,name=prom,firstnumber=auto]
active proctype Server() {
  byte n;
  
  do&\label{line:g1}&
  :: request?n;&\label{line:g3}&
     printf(&"&client %d\n&"&, n)&\label{line:print}&
  od&\label{line:g2}&
}
    \end{lstlisting}        
  \end{minipage}
  \caption{A simple \promela\ program.}
  \label{fig:promela}
\end{figure}

\promela\ \cite{Holzmann03} is a concurrent modeling language
that has been used to model and analyse numerous industrial systems. It
comes with an industry-strength model checker \cite{Ben-Ari08}. 
We do not give a full LAGC-style semantics of 
\promela, but we discuss its main features and illustrate
that they can be formalised using minor variations of the concepts
discussed above.
We also do not give the DL calculus for \promela, which mostly follows
the calculus developed in
Sections~\ref{sec:dynamic-logic-simple}--\ref{sec:interleaving}, since
\promela\ is designed for model checking.
We assume the reader is familiar with \promela.  A simple example of a
\promela\ program is shown in Figure~\ref{fig:promela}.

\paragraph{Types, Variables, Expressions}

All \promela\ datatypes are mapped into finite integer types,
including Booleans, process and channel identifiers, enumerations,
arrays, etc. Strings occur only as literals in print statements.
Variables declared outside a process are global and can be accessed by
any process, such as the \textsf{request} channel in Figure~\ref{fig:promela},
line~\ref{line:chan}.
Like in \Simple{WHILE}, \promela\ expressions are side effect-free and
their evaluation is standard.

\paragraph{Processes}

Methods and method calls do not exist in \promela, so \Simple{call} is
not present.  All \promela\ processes are declared and started in the
beginning by a---possibly implicit---initial process, so
\Simple{spawn} does not occur in the \promela\ syntax either. In
Figure~\ref{fig:promela}, three processes are declared and started
upfront.
It is easy to create an initial judgment similar to
rule~\eqref{eq:spawn-global-multi} from the process declarations of a
\promela\ program $P$:
\[
  \langle I_P\rangle,\,\Omega_P\enspace,
\]
where $\Omega_P(p)=\{\cont{\{s_p\}}\}$ for each process $p\in P$ with
program code $s_p$. The code is wrapped in a scope to handle local
variable declarations as in Section~\ref{sec:localmemory}.

There is no need for a method table, because there are no method
calls. For the same reason, the range of $\Omega$ is a singleton.
\promela\ processes execute in parallel and interleave on global
variables.
The granularity of interleaving in \promela\ is the same as in
Section~\ref{sec:local.parallelism}, so the continuations in the local
rules are unchanged.
Composition rule~\eqref{eq:multi.cont} is adequate for the \promela\
semantics, while rule~\eqref{eq:rulemulti.exec} is not needed, because
there are no method calls and the program code of all processes is
part of the initial judgment. In consequence, neither spawn, nor
invocation, nor invocation reaction events occur in \promela\ traces.

\paragraph{Statements}

There is a print statement with no effect on states, we decide not to observe the effect on the trace, we can thus give it
the semantics of \Simple{skip} in rule~\eqref{eq:skip}. Assignments are
exactly as in \Simple{WHILE}, discussed in Section~\ref{sec:while}.
Instead of \Simple{if}- and \Simple{while}-statements, \promela \ has
selection and repetition statements over guarded commands, which may
occur only there.  The semantics of guarded statements is as in
Section~\ref{sec:guarded.statements}. The semantics of a selection
\lstinline[language=promela,basicstyle=\codesize\sffamily,keywordstyle=\bf\ttfamily,mathescape=true]{if :: g$_1$; s$_1$ $\cdots$ :: g$_n$; s$_n$ fi }
is a straightforward generalisation of rules~\eqref{eq:If}
and~\eqref{eq:GuardedCommand} with $n+1$ premisses; one for each
guarded statement and one premiss if no guard is true. The path
condition of the $i$-th guarded statement is
\lstinline[language=promela,basicstyle=\codesize\sffamily,keywordstyle=\bf\ttfamily,mathescape=true]{g$_i$}. When more than
one guard is true, any one of them can be taken.

More generally, in \promela, any statement $s$ can serve as a guard,
so a guard can have a side effect. This can be modeled simply by
putting $s$ into the continuation of its premise. All statements
except \Simple{send} and \Simple{receive} (see below) are executable,
their path condition simply becomes $\trueSem$. This can be assumed
even for \Simple{send} and \Simple{receive}, because their execution
is modeled by well-formedness constraints (see below), not by
guards.

A repetition loops infinitely over its body and executes any of the
statements whose guard is true; in the loop on the
lines~\ref{line:g1}--\ref{line:g2} of Figure~\ref{fig:promela},
line~\ref{line:g3} has an empty guard, which defaults to true. A
repetition can only be exited by \lstinline[language=promela,basicstyle=\codesize\sffamily,keywordstyle=\bf\ttfamily]{break}
or \lstinline[language=promela,basicstyle=\codesize\sffamily,keywordstyle=\bf\ttfamily]{goto} statements (see below).  The
local evaluation of repetitions can be reduced to selection in a
similar manner as \Simple{while} is reduced to \Simple{if} in
rule~\eqref{eq:While}.

\paragraph{Jumps}

\promela\ features a \lstinline[language=promela,basicstyle=\codesize\sffamily,keywordstyle=\bf\ttfamily]{goto} statement,
whose argument is a label within the same process. Such unconditional
jumps without a (process) context switch are easy to characterise
semantically via continuations:
\evalruleN{goto}{%
  \valD{\Simple{goto}\;l} & = & \{\emptyset\pop \concatTr{\langle\sigma\rangle}{\cont{s}}\}
  \qquad \text{where $s$ is the code following $l$\ .}
}

The \lstinline[language=promela,basicstyle=\codesize\sffamily,keywordstyle=\bf\ttfamily]{break} statement is equivalent to a
\lstinline[language=promela,basicstyle=\codesize\sffamily,keywordstyle=\bf\ttfamily]{goto} with an implicit label that points
to the code after the repetition it is contained in.

\paragraph{Atomic}

\promela\ has an \Simple{atomic} statement\footnote{Although \promela\
  permits loops to occur inside \texttt{\bfseries atomic}, these are
  actually discouraged; if loops occur, they must always terminate.},
which works as in Section~\ref{sec:atomic}, except it may contain a
guard that can cause blocking. In this case, interleaving \emph{is}
possible.
In case the guard $g$ evaluates to true, statement $s$ is simply
atomically executed according
to rule~\eqref{eq:Atom}.\footnote{Rule~\eqref{eq:Atom} for unguarded atomic
  blocks ensures that the derived trace always ends with an empty
  continuation.}  Otherwise, it puts the whole guarded atomic block
into a continuation, so that the guard can be re-evaluated later.
\evalruleN{GuardedAtomic}{%
  \valP{O,F}{\Simple{atomic}\{:: g;  s\}} &=&
  \{\{\valD{g}=\falseSem\}\pop
  \concatTr{\langle\sigma\rangle}{\cont{\Simple{atomic}\{:: g;  s\}}}\}
  \\  
  &&\hspace{1mm} \cup
\{ \{\valD{g}=\trueSem\}\!\cup\!\pc \pop \concatTr{\tau}{\cont{\zero}}
\mid \pc\!\pop\! \concatTr{\tau}{\cont{\zero}} \!\in\!\valP{O,F}{\Simple{atomic}(s)}\}
}

\paragraph{Channels}

In \promela, the \Simple{send} and \Simple{receive} commands have the
syntax seen in Figure~\ref{fig:promela}: if $c$ is a channel identifier
and $\many{e}$ a list of expressions, then the current process uses
$c\,!\,\many{e}$ to send the value of $\many{e}$ to $c$. Dually, with
$c\,?\,\many{e}$ it receives $\many{e}$ from $c$.\footnote{\promela\
  allows either a variable, which is then set to the received value or
  an expression that is matched against it.}
In contrast to Definition~\ref{def:channel}
(Section~\ref{sec:multicore}), \promela\ channels are explicitly
declared, and the address of a send (the origin of a receive)
statement is not the recipient (sender) process, but a channel. A
message can be read by anyone who has access to the channel identifier
it was sent to. In the case of globally declared channels
(line~\ref{line:chan}), this is any process. For locally declared
channels the receiving process must have received or declared the
channel identifier.

The semantics of \promela's send and receive can be modeled in analogy
to Section~\ref{sec:multicore-semantics}, but one uses send and receive
events of the form $\sendEvp{}{\,p}{\many{e}}{c}{i}$ and
$\receiveEvp{}{\,p}{\many{e}}{c}{i}$, respectively, where the
address is a channel.\footnote{In case one of the $e$ is not a
  variable, but a match expression, the assignment in the local
  evaluation rule~\eqref{eq:Send} has the variables \emph{occurring}
  in $e$ in its domain.}

\paragraph{Well-formedness}

Only the first two rules of Definition~\ref{def:well-formedness-multi} (with
obvious adaptations) are needed, because there are no method calls or
dynamically created processes.
On the other hand, the option of local channels requires to model
the visibility of channels. This can be done by well-formedness
constraints of the following form:
\[
  \wfT{ch}{\consTr{\receiveEvp{}{\,p}{e}{c}{i}}{sh} } =
  \wf{sh} \land
  \left(\text{isLocal}(p,c)\lor\text{isGlobal}(c)\lor{\receiveEvp{}{\,p}{c}{c'}{i}}\in
    sh\right)\ .
\]

The predicate $\text{isLocal}(p,c)$ holds if the process $p$ declares a local
channel named $c$, $\text{isGlobal}(c)$ holds if channel $c$ is
declared globally. These predicates can be easily checked by
inspecting the code of a \promela\ program. The constraint expresses
that each channel used in a receive statement in  process $p$ must be
visible at that point. This is the case if it is either global, or
locally declared in $p$, or it was received earlier. There is a
similar constraint for send statements.

There are two channel variants in \promela: \emph{rendez-vous
  channels} impose strictly synchronous communication. For those
channels the well-formedness constraint in Definition~\ref{def:sync} is
used.
The other channel type are \emph{buffered channels} with a capacity
$N>0$. These are characterised by the bound FIFO pattern in
Definition~\ref{def:bfifo}.

The standard \promela\ semantics stipulates that received messages are
consumed, i.e.\ they can be read only once, which is ensured by the
following constraint:
\[
  \wf{\consTr{\receiveEvp{}{\,p}{e}{c}{i}}{sh}}\,=\,
  \wf{sh}\land\not\exists\, e', p'.\, \receiveEvp{}{\,p'}{e'}{c}{i} \in
  sh\ .
\]

It is possible to change certain features of \promela\ channels, for
example, reception can be made non-consumptive, out-of-order reception
can be permitted, send and receive can be made non-blocking.  All of
these can be characterised by suitable well-formedness constraints,
but we refrain from spelling out details.

\paragraph{Synchronisation}

Channels, like guarded commands, are often used for synchronisation
purposes in \promela. The receive statement on a buffered channel is
only executable if it contains a message; one can only send to a
channel that is not full. Non-executable statements block a process
until they become executable. In contrast to guarded commands, no
specific evaluation rule is required to model blocking senders and
receivers, because the correct communication order is already
guaranteed by the well-formedness constraints.

This consideration suggests that blocking guards might as well be
modeled by suitable ``waiting'' events. Indeed, this is possible,
but less natural in that case.



\subsection{Actors}
\label{sec:actor.class}

We define a pure, object-based actor language based on the language in
Section~\ref{sec:multicore}; Figure~\ref{fig:actor.class} shows the
syntax. There is a one-to-one mapping between objects and actors. Each
actor has its own process with local memory and executes calls to its
methods with run-to-completion semantics (in consequence,
non-terminating method calls render an actor unresponsive). Each actor
can only access its own local memory. The state of another actor can
only be modified via a method call.

\begin{figure}[th]\centering
$  \begin{array}[t]{r@{\hspace{2pt}}r@{\hspace{2pt}}l}
    P \in \mathit{Prog}&::=&\many{CD}\ \scope\\
    CD \in \mathit{ClassDecl} & ::=& \Simple{class}\ C\ \{ \many{\fd}\ \many{M} \}\ 
    \\
    M \in \mathit{MethDecl} & ::=& m(\overline{x})\ sc \ 
    \\
    d\in\mathit{VarDecl} &::=& \varepsilon\ |\ x; \ d
    \\    
    \scope \in \mathit{Scope} &::=& \{d \ \Simple{atomic}(\stmt)\} \\    
    s\in\mathit{Stmt}&::=&  \ \Simple{skip} \ | \  x := \mathit{rhs} 
\ | \ \Simple{if}~e~\{ ~\stmt~\} 
 \ | \ x!m(\many{e})  
\ |\ \stmt;\stmt \\
    \mathit{rhs} \in\mathit{Rhs} &::=& e  \ | \  \Simple{new} \ C\,(\many{e})
   \end{array}$
   \caption{Syntax of the Actor Language in Section~\ref{sec:actor.class}\label{fig:actor.class}}     
\end{figure}

In this actor language, a program $P$ consists of a set of class
declarations $CD$ with a main block $\scope$.  Each class has a name
$C$, which also becomes the type of the class, with a sequence of
field declarations $\fd$ and method declarations~$M$.  Field
declarations have the same syntax as local variable declarations $d$.
For simplicity, we let each method declaration associate a \emph{unique}
name~$m$ to a list of arguments $\many{x}$ and a method body $\scope$.

A method body starts with local variable declarations followed by an
\Simple{atomic} statement wrapping a sequence of statements
$s$. Recall from Section~\ref{sec:atomic} that we excluded loops
inside \Simple{atomic}, and thus exclude the \Simple{while} statement
from the syntax.  This is no fundamental restriction, because actors
may call each other and permit unbounded recursion.

The effect of \Simple{atomic} is to enforce a run-to-completion
semantics. Rather than changing the semantics of the programs
discussed in Sections~\ref{sec:while}--\ref{sec:multicore}, we here
make atomicity explicit.  In most actor
languages~\cite{ActiveObjects17}, the \Simple{atomic} keyword is not
used because actors are single threaded, and thus ``locally atomic''
by default.

We extend the statements of \Simple{WHILE}, given in
Figure~\ref{fig:basiclangWhile},
 with 
asynchronous method calls
$x!m(\many{e})$ on a caller object $x$.
The semantics of asynchronous calls combines the semantics of sending
a message and a method call. These calls are not blocking, so in the
absence of futures or promises, there is no direct way to return a
value to the caller (return values via futures are modeled in
Section~\ref{sec:actor.future} below).\footnote{It is possible to pass
  the caller's identity as a call argument, which the callee could
  then use to return a value to the caller's state via a separate
  callback. But this does not give the caller a handle to access the
  result. }
The right-hand side of assignments includes expressions and
$\Simple{new} \ C\,(\many{e})$. The latter creates a new actor object
of class $C$ with constructor arguments $\many{e}$.

In the following, let $\OId$ be a set of object identifiers with
typical element~$o$.  We represent a program $P$ in terms of a global
lookup table $\Gi$: A set of triples each consisting of a class name
with its field and method declarations:
\begin{equation*}
  \Gi= \{
  \langle C, \many{\fd}, \many{m(\many{x})\ \scope} \rangle
  \mid \fd \in \fields{C},\, m\in \methods{C},\,C\in P
  \}\ .
  \label{eq:actor.G}
\end{equation*}

In the above expression we use selector functions with obvious
definitions: $\fields{C}$ returns the field names of a class $C$,
$\methods{C}$ returns the methods declared in a class $C$. We also
define $\clookup{X}$, where $X$ is either an object identifier or a
variable of type $\OId$, that returns the class of $X$, as well as
$\lookup{m,\Gi}$ that returns the declaration of a method with name
$m$ in $\Gi$.

\subsubsection{Local Evaluation}

When following the principle of local evaluation, we face the problem
that a statement $s$, which is locally evaluated, cannot possibly know
the object it is running on. We address this issue in the same manner
as we dealt with unknown values before, i.e.\ by introducing a
parameter $O$ that is instantiated during trace composition. Since
this parameter must be instantiated consistently for all statements an
object $O$ executes, it becomes a parameter of the semantic evaluation
function, denoted $\valP{O}{s}$, where~$O$ is a symbolic variable.
The evaluation function $\valP{O}{s}$ produces the traces when $s$ is
running on object $O$; in particular, we let
$\valP{O}{\Simple{this}}=O$.

The evaluation rule for \Simple{new} below introduces an event
$\newEvs{\sigma}{o, \many{v}}$ to capture that a new object $o$ with
arguments~$\many{v}$ is created.
Similar to Rule~\eqref{eq:Spawn} that spawns a task, the local
evaluation Rule~\eqref{eq:NewClass} for object creation creates a
fresh symbolic variable $X$ to represent the unknown object identity
that is returned.  Hence, Rule~\eqref{eq:NewClass} has an empty path
condition, extends the trace by the \emph{new object creation} event
followed by an updated state, where~$x$ is mapped to symbolic
variable~$X$ and object fields are mapped to the constructor
arguments.  We suppose that $\Gi$ is pre-populated with enough
variables of each class. Consequently, Rule~\eqref{eq:NewClass} can
pick a symbolic variable $X$ that is fresh in $\sigma$ and such that
$\clookup{X} = C$.

\evalruleN{NewClass}{%
\valP{O}{x :=\Simple{new} \ C\,(\many{e})} 
&=&\\
&&\hspace{-20mm}\{\emptyset\pop\concatTr{\consTr{{\sigma}[x\mapsto X,
    X\mapsto *, \many{X.\fd \mapsto
      \valP{O}{e}}]}{\newEvsV{\sigma}{X,\valP{O}{\many{e}}}{\{X\}}} } {\cont{\zero}}
\\
&&\hspace{-18mm}\mid
X\not\in\dom(\sigma),\ \clookup{X} = C,\ \fd \in \fields{C}\}}

Rule~\eqref{eq:Asynch} for non-blocking asynchronous method calls is
similar to Rule~\eqref{eq:Call}, but the invocation event
$\invEv{\sigma}{\many{v},x,m,i}$ also includes the callee object $x$
and a message identifier $i\in\MId$. Thus, asynchronous method calls
can be viewed as a combination of procedure calls and message sending.
\evalruleN{Asynch}{%
  \valP{O}{x!m(\many{e})} 
  & = & \{\emptyset\pop\concatTr{\invEv{\sigma}{\valP{O}{\many{e}},\valP{O}{x},m,i}}{\cont{\zero}}\}
\mid i\in \MId\} }

The following rule evaluates the body of a method and corresponds to
Rule~\eqref{eq:method} but adapted to a list of arguments instead of a
single one.  The invocation reaction event inserted by the rule
includes symbolic variables which represent the method's actual
parameters, which are unknown in the local context. Since there are no
return values, it does not need to know the caller.
\evalruleN{method-class}{%
  \valP{O}{m(\many{x})\ \scope} &=& \{
  \emptyset \pop \concatTr{\consTr{\sigma[\many{z}\mapsto \many{Z}, \many{Z}\mapsto * ]}{\invREvV{\sigma}{\many{Z},m,i}{\many{Z}}}
  }{\cont{sc\subst{\many{x}}{\many{z}}}}\\
  &&\ \ \mid \many{z},\many{Z}\not\in
  \dom(\sigma),\ i\in \MId
  \}
}

\subsubsection{Trace Composition}

As explained above, the local evaluation is parameterised with the
executing object $O$. During trace composition, this parameter is
instantiated by a concrete object identifier $o$. To associate events
with the object they originate from, traces are tagged with that
object. This works exactly like the process identifier tags in
Definition~\ref{def:tagged-trace} of Section~\ref{sec:multicore}.

We use a mapping $\Omega$ into multisets of possible continuations
like in Section~\ref{sec:multicore}, with the difference that its
domain is object identifiers instead of process identifiers. We use
the same notation for multiset operations as before.

Trace composition in Rule~\eqref{eq:actor.multi.cont} follows
Rule~\eqref{eq:multi.cont} closely , with two small modifications:
instead of process identifiers, objects identifiers are used and the
evaluation of the continuation $s$ is performed on the object $o$,
where it is scheduled. The first two premises of the rule capture the
scheduling decision to continue the trace by executing task $s$ on
object $o$.
\begin{equation}
  \inferrule
  {o\in\dom(\Omega) \qquad
  \Omega[o]=q\multiunion\{\cont{s}\} \\  
    \sigma = \last(sh) \\
    pc\pop \concatTr{\tau }{\cont{s'}} \in \valP{o}{s}
    \\
    \rho\text{ concretizes }\tau \\ 
    \rho(pc) \text{ consistent} \\ 
    \wf{\chopTrSem{sh}{\rho(\tau)}^o}}     
  {sh, \Omega \to
    \chopTrSem{sh}{\rho(\tau)^o}, \Omega[o\mapsto q\multiunion\{\cont{s'}\}]}
  \label{eq:actor.multi.cont}
\end{equation}

Likewise, Rule~\eqref{eq:actor.multi.invREv} \emph{exactly} follows
the pattern established with Rule~\eqref{eq:rulemulti.exec}.  The rule
picks a method $m$ and an object $o$, on which it is supposed to be
executed. If that object is not yet in the domain of $\Omega$, then we
use the notational convention $\Omega(o)=\emptyset$. This corresponds
to object creation.  Either way, the method declaration is evaluated
on $o$. Observe that Rule~\eqref{eq:method-class} issues an
\emph{invocation reaction} event.  Well-formedness ensures that an
\emph{invocation} event of the form $\invEv{}{\many{v},o,m,i}$ is
present in $sh$ that matches the message identifier, method
parameters, and callee.
\begin{equation}
\inferrule
{
  \lookup{m, \Gi} = m(\many{x})\ \scope 
  \qquad
  \sigma = \last(sh)
  \qquad
  pc\pop \concatTr{\tau }{\cont{s'}} \in \valP{o}{m(\many{x})\ \scope} \\
  \rho\text{ concretizes }\tau \qquad 
  \rho(pc) \text{ consistent}
  \qquad
  \wf{\chopTrSem{sh}{\rho(\tau)^o}} \qquad 
}{
  sh, \Omega \to
  \chopTrSem{sh}{\rho(\tau)}^o,\, 
  \Omega[o\mapsto \Omega[o] \multiunion \{\cont{s'}\}]
}
\label{eq:actor.multi.invREv}
\end{equation}

\subsubsection{Well-formedness}
\label{def:wellformedness-class} 

Well-formedness must essentially ensure two aspects:

\begin{enumerate*}
\item the uniqueness of events, for example, there is at most one
  $\newEvs{}{o,\_}$ per object, i.e.\ an object cannot be created
  twice;
\item the event sequence related to a call to a method $m$ on object
  $o$ with parameters $\many{v}$ has the following form:
\end{enumerate*}
\[
  \cdots\quad \newEvp{}{o''}{o,\many{v}}\quad\cdots\quad
  \invEvp{}{o'}{\many{v'},o,m,i}\quad\cdots\quad\invREvp{}{o}{\many{v'},m,i}\quad\cdots
\] 
where $o'$ and $o''$ can be the same object. This is achieved with the
equations in Figure~\ref{fig:well-formedness-actor} (we do not repeat the
first two lines of Definition~\ref{def:wellformedness-concrete-trace}, which
are always assumed to be part of well-formedness).

\begin{figure}[th]\centering
   $ \begin{array}[t]{r@{\hspace{2pt}}r@{\hspace{2pt}}l}
      \wf{\consTr{\newEvp{}{o'}{o,\many{v}}}{sh}}  &=&  \wf{sh} \land \not\exists o'',\many{v'}.\, \newEvp{}{o''}{o,\many{v'}} \in sh
      \\
      \wf{\consTr{\invEvp{}{o'}{\many{v},o,m,i}}{sh}}  &=&
                                                           \wf{sh} \land \ 
                                                           \exists o'', \many{v}'
                                                           . \newEvp{}{o''}{o,\many{v}'}\in sh \ \land
      \\   &&\not\exists o''',\many{v}'',o'''',m'.\, \invEvp{}{o'''}{\many{v}'',o'''',m',i}\in sh	
      \\
      \wf{\consTr{\invREvp{}{o'}{\many{v},m,i}}{sh}}  &=&  \wf{sh} \land \ \exists o.\invEvp{}{o}{\many{v},o',m,i} \in sh 
                                                          \ \land \\
                     && \not\exists o'',\many{v}',m'.\, \invREvp{}{o''}{\many{v'},m',i}\in sh
    \end{array}$
  \caption{Well-formedness for Actor language.\label{fig:well-formedness-actor}}
\end{figure}

\subsubsection{Global Trace Semantics}
\label{sec:global-actors}
\begin{definition}[Program Semantics for Actors]
  \label{def:global-trace-semantics}
  Given a finite program $P$ with a lookup table $\Gi$ and a main block
  $\mathit{main}$, i.e.,
  $ \Omega_{\mathit{init}}( o_{\mathit{main}}) =
  \{\cont{\mathit{main}}\}$.
  Let $\sigma_\varepsilon$ denote the empty state, i.e.\
  $\dom(\sigma_\varepsilon)=\emptyset$.
  Let $sh_0,\,\Omega_0\rightarrow\,sh_1,\,\Omega_1\rightarrow\cdots$\
  be a maximal sequence obtained by the repeated application of
  the composition rules (Rules~\eqref{eq:actor.multi.cont}--\eqref{eq:actor.multi.invREv}),
  starting with\footnote{Note that we denote by $\varepsilon$ an empty list of method arguments.}
  \begin{equation*}
    \consTr{\newEvp{}{o_{\mathit{main}}}{o_{\mathit{main}},\varepsilon}}{\langle\sigma_\varepsilon\rangle},\,\Omega_{\mathit{init}}\enspace.
  \end{equation*}
  If the sequence is finite, then it must have the form
  \[\consTr{\newEvp{}{o_{\mathit{main}}}{o_{\mathit{main}},\varepsilon}}{\langle\sigma_\varepsilon\rangle},\,\Omega_{\mathit{init}}
    \to \dots \to sh,\Omega^{\zero}\enspace,\]
  where $\Omega^{\zero}(o)=\emptyset$ 
  for all $o\in\dom(\Omega)$.
  If the sequence is infinite, let $sh=\lim_{i\rightarrow\infty}sh_i$.
  In either case, $sh$ is a trace of $P$. The set of all such traces
  is denoted with $\traces{P, \sigma_\varepsilon}$.
  \label{ex:actor.global.trace.multi}
\end{definition}

The new event at the start of a trace represents creation, by the
system, of an initial object that runs the main method. The
well-formedness of events ensures that subsequently created objects
will not erroneously be the initial object, which is reserved for the
code executed in the main method.


\subsection{Active Objects}
\label{sec:actor.future}

Active object languages \cite{ActiveObjects17} have a mechanism
like futures or promises that provides a reference to the value
computed by an asynchronous method call. This makes it possible for a
task to free its processor resource while waiting for a result to be
finished (so-called cooperative multi-tasking). 
We modify the syntax from Section~\ref{sec:actor.class} and introduce
futures, a \Simple{return} statement, and \Simple{get}-expressions to
retrieve the value stored in a future; the resulting syntax is given
in Figure~\ref{fig:actor.future}. Note that the latter can block if
this value is not yet available.

\begin{figure}[th]\centering
$  \begin{array}[t]{r@{\hspace{2pt}}r@{\hspace{2pt}}l}
        P \in \mathit{Prog}&::=&\many{CD}\ sc\\
    CD \in \mathit{ClassDecl} & ::=& \Simple{class}\ C\ \{ \many{\fd}\ \many{M} \}\ 
    \\
    M \in \mathit{MethDecl} & ::=& m(\overline{x})\ sc \ 
    \\
    d\in\mathit{VarDecl} &::=& \varepsilon\ |\ x; \ d
    \\
    sc \in \mathit{Scope} &::=& \{d \ s;\ \Simple{return}\ e\} \\ 
    g \in \guard &::= & e? \mid e\\
    s\in\mathit{Stmt}&::=&  \ \Simple{skip} \ | \  x := \mathit{rhs} \ | \ \Simple{if}~e~\{ ~\stmt~\} \ | \ \Simple{while}\ e\, \{\, \stmt\, \} 
 \  |\ \Simple{this}.m(\many{e}) \  |\ \Simple{await}\ g  \ |\ \stmt;\stmt \\  
    \mathit{rhs} \in\mathit{Rhs} &::=& e  \ | \  \Simple{new} \ C\,(\many{e})\  | \ x!m(\many{e}) | \ e.\Simple{get}\   
    \\
  \end{array}$
\caption{Syntax of Active Objects with Futures.\label{fig:actor.future}}     
\end{figure}

We need to capture cooperative multi-tasking and the blocking of future
access into the semantics. This cannot be done with \Simple{atomic},
as in the previous section, because the suspending statements might occur
nested inside loops and recursive calls. It is more appropriate to
switch to a semantics where a task \emph{runs to completion} by
default, unless it suspends or blocks.
As a consequence we will specify a semantics with finer-grained,
explicitly controlled interleaving between tasks. For this reason, we
now consider both \Simple{while}-loops and synchronous self-calls
$\Simple{this}.m(\many{e})$ in the syntax.

Task suspension takes place in $\Simple{await}\,g$ statements, where
the guard $g$ suspends the execution of a task; it works similarly to
the guarded statements of Section~\ref{sec:guarded.statements}.  An
\Simple{await} releases its object's process, such that other tasks
may be executed. The execution of the statements following the guard
can resume when the guard evaluates to true.  Even when the guard
evaluates to true, the execution may be suspended.  Guard expressions
either have the form $x?$, which synchronize with the future
referenced by $x$ receiving a value (the future thereby gets
\emph{resolved}),
or they are Boolean expressions $e$.

The second difference to the Actor language of
Section~\ref{sec:actor.class} is that asynchronous method calls return
values to their associated future.\footnote{Obviously, synchronous
  calls can also be equipped with return values in a straightforward
  manner.} The \Simple{return} statement terminates the execution of
an asynchronous method and returns its argument. The syntax enforces
that the \Simple{return} statement is at the end of a method body.
Asynchronous method calls now appear on the right hand side of an
assignment, the variable on the left is a \emph{future} associated
with the call.  The future's value can be retrieved with \Simple{get},
whose argument expression must evaluate to a future. To avoid
blocking, one can precede $x.\Simple{get}$ with $\Simple{await}\,x?$.
Futures are first-class values and may be passed as method arguments
to other objects.

\subsubsection{Local Evaluation}
\label{sec:local-evaluation-actor-futures}

For similar reasons as for objects in Section~\ref{sec:actor.class},
we need to provide the future towards whose value a local statement is
contributing, as a parameter of the evaluation function. This future,
sometimes called the \emph{destiny} of the executing code
\cite{deboer07esop}, cannot possibly be known locally. Hence, the form
of the semantic evaluation function becomes $\valP{O, F}{s}$, where
$F$ is a symbolic variable of type future with values in $\FId$; the
evaluation function captures that $s$ is running on object $O$ and has
destiny~$F$. To keep track of the destiny, we store it together with
the continuations; in the configurations, continuations now take the
form $\contF{f}{s}$.


The rule for assignment with an asynchronous call on the right-hand
side first emits an invocation event $\invEv{}{\many{e},o,m,f}$,
similar to Rule~\eqref{eq:Asynch}. The difference is that now the
future $F'$, that is the destiny of the call, must be recorded both in
the left-hand-side variable on the left and in the invocation
event. That future can also be used to identify the call, thus
replacing the message identifier. Since the value of $F'$ cannot be
known locally, it is modeled as a fresh symbolic variable, similarly
as in Rules~\eqref{eq:Receive}--\eqref{eq:Spawn}.
\evalruleN{Async-Assign}{%
\multicolumn{1}{l}{\valP{O,F}{x:=e!m(\many{e'})}=}\\ 
\hspace{7mm}\{\emptyset\pop\concatTr{\consTr{{\sigma}[x\mapsto F',F'\mapsto
    *]}{\invEvV{\sigma}{\valP{O,F}{\many{e'}},\valP{O,F}{e},m,F'}{\{F'\}}
  }}{\cont{\zero}}
\\ 
\multicolumn{1}{l}{\hspace{8mm}
\mid
F'\not\in\dom(\sigma)\}
}}

To ensure that only futures with an available value are retrieved, we
introduce new events $\compEv{\sigma}{f, v}$ and
$\compREv{\sigma}{f,v}$.  These events denote the \emph{completion} (the
value of future $f$ is available) and \emph{completion reaction} (the
value of $f$ is retrieved) of an asynchronous method.  The evaluation
rule for returning the completed result simply inserts a
\emph{completion} event for the current future.
\evalruleN{Return}{%
\valP{O,F}{\Simple{return}~e} 
&=&\{\emptyset\pop\concatTr{\compEv{\sigma}{F,
    \valP{O,F}{e}}}{\cont{\zero}} \}
}

To retrieve the returned value stored in a future, Rule~\eqref{eq:Get}
inserts a completion reaction event and extends the current state with
a symbolic variable~$V$ that holds the as yet unknown return
value.
\evalruleN{Get}{%
\valP{O,F}{x:=e.\Simple{get}} 
&=&\{\emptyset\pop\concatTr{\consTr{{\sigma}[x\mapsto V,V\mapsto
    \ast]}{\compREvV{\sigma}{\valP{O,F}{e},V}{\{V\}}}}{\cont{\zero}}\\
&&\hspace{2mm} \mid
 V\not\in\dom(\sigma)\}
}

The rule for evaluating a method body is similar to
Rule~\eqref{eq:method-class}, except that
the future variable~$f$ of the freshly running process is unified with
the variable $f$ of the invocation reaction event. It also must
contain the (as yet unknown) caller $Y$. Well-formedness of the trace
containing $\invREv{}{\many{X}, Y, m,f}$ will ensure that the $f$ of
the invocation reaction is matched with the same $f$ for a matching
invocation event (see the definition of $\wf{sh}$ below).
\evalruleN{method-future}{%
  \multicolumn{1}{l}{\valP{O,F}{m(\many{x})\ \scope} =} \\
 \hspace{7mm} \{ \emptyset \pop
  \concatTr{\consTr{\sigma[\many{x}'\mapsto\many{X}, \many{X}\mapsto *, Y\mapsto * ]}{\invREvV{\sigma}{\many{X}, Y, m,F}{\many{X}\cup\{Y\}} }
  }{\cont{\scope\subst{\many{x}}{\many{x}'}}}\\ 
\multicolumn{1}{l}{ \hspace{8mm}\mid \many{x}', \many{X}, Y \not\in\dom(\sigma)\}}}

For \Simple{while}-loops, we can use Rule~\eqref{eq:While}.
Rule~\eqref{eq:Synch} handles synchronous self-calls (no other
synchronous calls are considered here) by inlining the method body. It
obtains the method declaration of~$m$ with body $sc$ using the
auxiliary function $\lookup$.and turns the formal parameters
$\many{x}$ into local variable declarations bound to the argument
values $\many{e}$. The entire statement sequence is wrapped into a
scope for name disambiguation.
\evalruleN{Synch}{%
  \valP{O,F}{\Simple{this}.m(\many{e})} 
  &=& 
  \valP{O,F}{\{\many{x};\ \many{x}:=\valP{O,F}{\many{e}};\ \scope\}},\ \text{where} \lookup{m, \Gi} = m(\many{x}) \ \scope}%
    
Suspended tasks are introduced by \Simple{await} statements. We only
need to specify how to progress after an \Simple{await}---the trace
composition semantics will deal with actual task suspension by
checking for the presence of an \Simple{await} guard. There are two
syntactically distinct cases. The first corresponds to a guarded
command in Rule~\eqref{eq:GuardedCommand}, the second requires that
the guarded future is resolved and thus introduces a completion
reaction event.  This event must match a previous completion event
involving the same future, which will be ensured by well-formedness in
the global rules.
\evalruleN{Suspend}{%
  \valP{O,F}{\Simple{await}\,e} &=& 
\{ \{\valD{e}=\trueSem\}\pop \concatTr{\langle\sigma\rangle}{ \cont{\zero}}\}  \\[1ex]
 \valP{O,F}{\Simple{await}\,e?} &=&\{\emptyset\pop\concatTr{{\compREvV{\sigma}{\valP{O,F}{e},V}{\{V\}}} }{\cont{\zero}}
\mid
V\not\in\dom(\sigma)\} 	
}

\subsubsection{Trace Composition}
\label{sec:global-evaluation-active-objects}

We extend the continuations in the configuration with a future
identifier that corresponds to the future to be resolved by the
considered task (thus, continuations take the form $\contF{f}{s}$).
With active objects, there is no data-race among concurrent tasks,
because each data item belongs to one single active object and there
is at most a single active task for each active object.  One way to
realize this in the semantics is to assign a single process to each
active object and make sure that among all the tasks of this process,
only one is \emph{not} waiting for its turn.  In other words, for each
object identifier $o$, all continuations, except possibly one, are of
the form $\suspendedcontF{f}{g}{s}$.

We use the symbol $Q$ to denote a multiset of \emph{suspended} or empty
continuations of the form $\suspendedcontF{f}{g}{s}$ or $\contF{f}{\zero}$. The current global
state of a program trace is represented by a mapping $\Sigma$ from
object identifiers to multisets of continuations. For all object
identifiers $o$, $\Sigma[o]$ is either empty or of the form
$Q\multiunion\{\contF{f}{s}\}$, which reflects that there is at most
one active task.
If $s$ is of the form $\Simple{await}\,g;s'$, then no task is
currently executing on object $o$ and any task can be activated;
otherwise, $s$ must be the only continuation that can be executed in
object $o$. It is called the \emph{active} task. The rewrite rule for
empty trace simplification is extended to the new continuations:
$Q+\contF{f}{\zero} \rightsquigarrow Q$.  The trace composition rule
closely follows Rule~\eqref{eq:actor.multi.cont}:
\begin{equation}
  \inferrule
  { 
    \Sigma[o]=Q\multiunion\{\contF{f}{s}\} \\  \sigma = \last(sh) \\
    pc\pop \concatTr{\tau }{\cont{s'}} \in \valP{o,f}{s}\\
    \rho\text{ concretizes }\tau \\ 
    \rho(pc) \text{ consistent} \\ 
    \wf{\chopTrSem{sh}{\rho(\tau)}^{o}}
  }     
  {sh, \Sigma \to
    \chopTrSem{sh}{\rho(\tau)^{o}}, \Sigma[o\mapsto Q\multiunion\{\contF{f}{s'}\}]}
  \label{eq:actor.progress}
\end{equation}

The main difference is that the global state has a different signature
and the nature of $Q$ enforces the selection of the current
task. Observe that this rule either evaluates the current active task
or activates a new task in $o$ when there is no active task.

Rule~\eqref{eq:actor.future.invREv2} follows
Rule~\eqref{eq:actor.multi.invREv} except that the evaluation function
is now tagged with a concrete future identity $f$. The correctness of
the concretisation is ensured by the well-formedness premise that
requires an invocation event of the form $\invEv{}{\many{v},o,m,f}$ to
be present in $sh$.  Note that we only run a new method if there is no
active task ($\Sigma[o] = Q$).  Here, the destiny of the created
thread is the future that was created during the evaluation of the
invocation event $\invEv{}{\many{v},o,m,f}$.  As no other task is
active, the method can start running immediately.\footnote{Running the
  method immediately allows us, for example, to encode a FIFO service
  policy by composing with the communication ordering rules of
  Section~\ref{sec:communication-patterns}}
\begin{equation}
  \inferrule {
	\Sigma[o] = Q \\
    \lookup{m, \Gi} = m(\many{x})\ \scope 
    \\
    \sigma = \last(sh)
    \\
    f\in\FId \\ o\in \OId\\
    pc\pop \concatTr{\tau }{\cont{s'}} \in \valP{o,f}{m(\many{x})\ \scope} 
    \\\\
    \rho\text{ concretizes }\tau \\ 
    \rho(pc) \text{ consistent}
    \\
    \wf{\chopTrSem{sh}{\rho(\tau)^{o}}}  
  }{
    sh, \Sigma \to
    \chopTrSem{sh}{\rho(\tau)}^o,\, 
    \Sigma[o\mapsto Q \multiunion \{\contF{f}{s'}\}]
  }
\label{eq:actor.future.invREv2}
\end{equation}

This last rule can be triggered when $Q$ is empty, i.e. when a newly
created object handles its first invocation or when an object has no
task currently running.

\begin{proposition}
  The property that $\Sigma[o]$ contains at most one active task for
  any $o\in\OId$ is an invariant preserved by applications of
  Rules~\eqref{eq:actor.progress}--\eqref{eq:actor.future.invREv2}.
\end{proposition}

\begin{proof}
  When an object $o$ is created, $\Sigma[o]$ is empty. Every time one
  of the
  Rules~\eqref{eq:actor.progress}--\eqref{eq:actor.future.invREv2} is
  applied to $o$, a continuation is added to $\Sigma[o]$. As long as
  $\Sigma[o]$ contains only suspended continuations, either rule can
  be applied to $o$. Consequently, a suspended task on $o$ can be
  activated or a new method invocation on $o$ can start to be
  executed. If there is one active task $\contF{f}{s}$ in $\Sigma[o]$,
  only Rule~\eqref{eq:actor.progress} can be applied to $o$. In this
  case, execution on $o$ is forced by Rule~\eqref{eq:actor.progress}
  to continue with this specific $\cont{s}$, i.e., the active task
  on $o$. 
\end{proof}

The proposition guarantees that each object has at most one active
task at any time. This captures the semantics of sequential execution
on an object between two suspension points.

\subsubsection{Well-Formedness}

Compared to Section~\ref{def:wellformedness-class}, well-formedness
must reflect the life cycle of asynchronous method calls with
completion and completion reaction events to ensure that return values
are not delivered and retrieved too early:
\[
  \begin{array}{c}
    \cdots\quad \newEvp{}{o''}{o,\many{v}}\quad \cdots \quad
    \invEvp{}{o'}{\many{v'},o,m,f} \quad \cdots \quad \invREvp{}{o}{\many{v'},o',m,f}
    \quad \cdots\\
    \cdots \quad \compEvp{}{o}{f,v}\quad \cdots\quad
    \compREvp{}{o'''}{f,v}\quad \cdots
  \end{array}
\]
where $o'$ and $o''$ can be the same object, and $o'''$ can be any
object except $o$. Up to tagging, the first four events are unique in
each trace, but not so the final completion reaction event. This is,
because any object can retrieve the value stored in a future as long
as it possesses the future's identifier.

The fourth equation in Figure~\ref{fig:well-formedness-future} %
is straightforward, because
\begin{enumerate*}[label=(\roman*)]
\item a \Simple{return} can only be encountered after an invocation
  reaction event, which is guaranteed to be unique by the third
  equation, and
\item Rule~\eqref{eq:actor.progress} makes sure that it is evaluated
  on the matching future and object.
\end{enumerate*}

%
The final equation ensures the value of a future can only be fetched
after the future is resolved.

\begin{figure}[t]\centering
    $\begin{array}[t]{r@{\hspace{2pt}}r@{\hspace{2pt}}l}
      \wf{\consTr{\newEvp{}{o'}{o,\many{v}}}{sh}}  &=&  \wf{sh} \land \not\exists o'',\many{v'}.\, \newEvp{}{o''}{o,\many{v'}} \in sh
      \\
      \wf{\consTr{\invEvp{}{o'}{\many{v},o,m,f}}{sh}}  &=& \wf{sh} \land \ 
															\exists o'', \many{v}'. \newEvp{}{o''}{o,\many{v}'}\in sh \ \land\\   &&\not\exists o''',\many{v}'',o'''',m'.\, \invEvp{}{o'''}{\many{v}'',o'''',m',f}\in sh	      
      \\
      \wf{\consTr{\invREvp{}{o'}{\many{v},o,m,f}}{sh}}  &=&  \wf{sh} \land \ \invEvp{}{o}{\many{v},o',m,f} \in sh 
                                                            \ \land \\
                     && \not\exists o'',\many{v'},o''',m'.\, \invREvp{}{o''}{\many{v'},o''',m',f}\in sh
      \\
      \wf{\consTr{\compEvp{}{o}{f,v}}{sh}}  &=&  \wf{sh} 
	\\
      \wf{\consTr{\compREvp{}{o}{f,v}}{sh}}  &=&  \wf{sh} \land \exists o',  \compEvp{}{o'}{f,v} \in sh 
    \end{array}$
  \caption{Well-Formedness for Actors with
    Futures.\label{fig:well-formedness-future}}
\end{figure}

\subsubsection{Global Trace Semantics}
\label{sec:global-active-objects}

\begin{definition}[Program Semantics for Active Objects]
	\label{def:global-trace-semantics-active-objects}
	Given a program $P$ with a method table $\Gi$ and a main block
	$\mathit{main}$, i.e.,
	$ \Sigma_{\mathit{init}}( o_{\mathit{main}}) =
	\{\contF{f_{\mathit{init}}}{\mathit{main}}\}$.
	Let $\sigma_\varepsilon$ denote the empty state, i.e.\
	$\dom(\sigma_\varepsilon)=\emptyset$.
	Let
        $sh_0,\,\Sigma_0\rightarrow\,sh_1,\,\Sigma_1\rightarrow\cdots$\
        be a maximal sequence obtained by the repeated application of
        the composition rules
        (Rules~\eqref{eq:actor.progress}--\eqref{eq:actor.future.invREv2}),
        starting with
	\begin{equation*}
		\consTr{\newEvp{}{o_{\mathit{main}}}{o_{\mathit{main}},\varepsilon}}{\langle\sigma_\varepsilon\rangle},\,\Sigma_{\mathit{init}}\enspace.
	\end{equation*}
If the sequence is finite, then it must have the form
	\[\consTr{\newEvp{}{o_{\mathit{main}}}{o_{\mathit{main}},\varepsilon}}{\langle\sigma_\varepsilon\rangle},\,\Sigma_{\mathit{init}}
	\to \dots \to sh,\Sigma^{\zero}\enspace,\]
	where $\Sigma^{\zero}(o)=\emptyset$ 
	for all $o\in\dom(\Sigma)$.
	If the sequence is infinite, let $sh=\lim_{i\rightarrow\infty}sh_i$.
	In either case, $sh$ is a trace of $P$. The set of all such traces
	is denoted with $\traces{P,\sigma_\varepsilon}$.
\end{definition}

The new event at the start of a trace represents the creation, by the
system, of an initial object that runs the main
method. The well-formedness of events ensures that subsequently created
objects will not erroneously be the initial object, which is reserved
for code executed in the main method.

\begin{example}
  \label{ex:active.objects}
  Consider the following program $P$:	
 
  \smallskip
  \centering
  
  \begin{minipage}[t]{0.45\linewidth}
    \begin{simple}[name=abs]
class $C$ {
$\label{ex:abs.m}\quad m(n)$ {
$\label{ex:abs.add}\quad\quad n := n+1$;
$\label{ex:abs.return}\quad\quad$return $n$;
  }
}
    \end{simple}
  \end{minipage}
  \qquad
  \begin{minipage}[t]{0.45\linewidth}
    \begin{simple}[name=abs,firstnumber=last]
{ // main block
$\label{ex:abs.declare}\quad a$; $x$; $f$; $y$; 
$\label{ex:abs.assign}\quad a := 1$;
$\label{ex:abs.new}\quad x :=$ new $C$();
$\label{ex:abs.inv}\quad f := x!m(a)$;
$\label{ex:abs.await}\quad$await $f?$;
$\label{ex:abs.get}\quad y := f$.get;
}
    \end{simple}
  \end{minipage}
  
\end{example}

Let the global lookup table $\Gi$ for this program be
$\{ \langle C, \varepsilon, m(n) \{n := n+1; \Simple{return}~n;\} \rangle
\} $. In the following, we use $s_i$ to indicate the sequence of
statements of a method body starting from line~$i$.

We first show the abstract local traces.  The evaluation of the $main$
method is as follows,
where $s_{\it main}$ represents the body in
$main$, and
$s_{\ref{ex:abs.assign}}$ corresponds to $s_{\it main}[a\leftarrow
a', x \leftarrow x', f \leftarrow f', y \leftarrow y']$. 
\[
\begin{array}[h]{@{}r@{\hspace{2pt}}r@{\hspace{2pt}}l}
  	\valP{O,F}{\{\ a; x; f; y; s_{\it main}\ \}}  & = & 
\{ \emptyset \pop \concatTr{\consTr{\sigma[a' \mapsto 0, x' \mapsto 0, f' \mapsto 0, y' \mapsto 0]}{\langle
		\sigma\rangle}}{\cont{\{\ s_9\ \}}}\\
&&\ \ \mid a', x', f', y' \notin \dom(\sigma)  \}		\\
\valP{O,F}{\{\  s_{\ref{ex:abs.assign}}\ \}} & = & \valP{O,F}{s_{\ref{ex:abs.assign}}} \\
  \valP{O,F}{a' := 1} & = & \{\emptyset\pop\concatTr{\consTr{\update{\sigma}{a'}{1}}{\langle\sigma\rangle}}{\cont{\zero}}\} \\
\valP{O,F}{x' :=\Simple{new} \ C()} 
&=&\{\emptyset\pop\concatTr{\consTr{{\sigma}[x'\mapsto X,
		X\mapsto *]}{\newEvsV{\sigma}{X,\varepsilon}{\{X\}}}}{\cont{\zero}} 
\\
&&\hspace{2mm}\mid
X\not\in\dom(\sigma),\ \clookup{X} = C\} \\ 
\valP{O,F}{f':=x'!m(a')} 
&=&\{\emptyset\pop\concatTr{\consTr{{\sigma}[f'\!\mapsto\! F',F'\!\mapsto\!*]}{\invEvV{\sigma}{\valP{O,F}{a'},\valP{O,F}{x'},m,F'}{\{F'\}}}}{\cont{\zero}} \\ 
&&\hspace{2mm} \mid
F'\not\in\dom(\sigma)\} \\
 \valP{O,F}{\Simple{await}\,f'?} &=&\{\emptyset\pop\concatTr{\consTr{{\sigma}[V\mapsto
		\ast]}{\compREvV{\sigma}{\valP{O,F}{f'},V}{\{V\}}}}{\cont{\zero}}
\mid
V\not\in\dom(\sigma)\} 	\\	
\valP{O,F}{y':=f'.\Simple{get}} 
&=&\{\emptyset\pop\concatTr{\consTr{{\sigma}[y'\mapsto V,V\mapsto
		\ast]}{\compREvV{\sigma}{\valP{O,F}{f'},V}{\{V\}}}}{\cont{\zero}}\\
&&\ \ \mid V\not\in\dom(\sigma)\}
\end{array}
\]	
The evaluation of the body of method $m$ is shown below, where $sc$
represents $\{ n := n+1; \Simple{return}~n; \}$ and
$s_{\ref{ex:abs.add}}$ the content of the block
$sc\subst{n}{n'}$.
\[
\begin{array}[h]{@{}r@{\hspace{2pt}}r@{\hspace{2pt}}l}
  \valP{O,F}{m(n)\ \scope} &=& \{ \emptyset \pop
                               \invREvV{\sigma}{N, X, m,F'}{\{N, X, F'\}}\\
\multicolumn{3}{r}{\hspace{105pt}\concatTr{\consTr{\sigma[n'\!\mapsto\!N,  N\!\mapsto\! *, X\!\mapsto\! *, F'\! \mapsto\! *]\!}{}}{\!\cont{\scope\subst{n}{n'}}}
\!\mid\! n',\! N,\! X,\! F'\!\!\not\in\! \dom(\sigma) \}}\\
 %
  \valP{O,F}{\{\  s_{\ref{ex:abs.add}}\ \}} & = & \valP{O,F}{s_{\ref{ex:abs.add}}} \\
  \valP{O,F}{n'\!\!:=\!n'\!+\!1} & = & \{\emptyset\pop\concatTr{\consTr{\update{\sigma}{n'}{\valP{O,F}{n'+1}}}{\langle\sigma\rangle}}{\cont{\zero}}\} \\
  \valP{O,F}{\Simple{return}~n'} &=&\{\emptyset\pop\concatTr{\compEv{\sigma}{F,\valP{O,F}{n'}}}{\cont{\zero}} \} \\              
\end{array}
\]

We now continue with the computation of a global concrete trace.  Let
$I_P$ be the initial state of~$P$.  We consider the state where the
statements in lines~\ref{ex:abs.declare}--\ref{ex:abs.assign} have
already been executed.  At this point, the concrete trace is
$
  sh=\consTr{\sigma}{\consTr{[a' \mapsto 0, x' \mapsto 0, f' \mapsto 0, y' \mapsto 0]}{\langle
      I_P\rangle}}
$
with $\sigma = [a' \mapsto 1, x' \mapsto 0, f' \mapsto 0, y' \mapsto 0]$, 
and the mapping $\Sigma$ is
$[\Omain \mapsto \{\contF{\Fmain}{s_{\ref{ex:abs.new}}}\}]$.
Only Rule~\eqref{eq:actor.progress} is applicable:

\gcondrule{ObjCreation-main}{
  \Sigma[\Omain]=\{\contF{\Fmain}{s_{\ref{ex:abs.new}}}\} \quad
  \sigma = \last(sh) 
  \\
  \emptyset\pop\concatTr{\consTr{{\sigma}[x'\mapsto X, X\mapsto *]}{\newEvsV{\sigma}{X,\_}{\{X\}}}}{\cont{s_{\ref{ex:abs.inv}}}}\in \valBP{\sigma}{\Omain,\Fmain}{s_{\ref{ex:abs.new}}}\\
  \rho=[X \mapsto o_m]  \quad
  \emptyset \text{ consistent}\quad
  \wf{sh'}
}{
  sh, \Sigma \to
  sh', 
  \Sigma[\Omain\mapsto \{\contF{\Fmain}{s_{\ref{ex:abs.inv}}}\}]
}
where $sh'$ is the concrete trace
\[
  \begin{array}{rl}
    &\chopTrSem{sh}{\rho(\consTr{{\sigma}[x'\mapsto X, X\mapsto *]}{\newEvsV{\sigma}{X,\varepsilon}{\{X\}}})}^{\Omain}\\
    =&\chopTrSem{sh}{\rho(\consTr{{\sigma}[x'\mapsto X, X\mapsto *]}{\consTr{\consTr{\sigma[X\mapsto *]}{\newEvsV{\sigma}{X,\varepsilon}{\{X\}}}}{\langle\sigma\rangle}})}^{\Omain}\\
    =&\langle I_P\rangle \curvearrowright \ldots \curvearrowright \sigma \curvearrowright \consTr{\sigma'}{\consTr{\sigma}{\newEvp{}{\Omain}{o_m,\varepsilon}}}   
  \end{array}
  \]
  with $\sigma' = \sigma[x'\mapsto o_m] = [a' \mapsto 1, x'\mapsto o_m, f' \mapsto 0, y' \mapsto 0]$.\footnote{We simplify the states in the concrete
    trace by removing the symbolic variables once they are instantiated.}
We must continue with $s_{\ref{ex:abs.inv}}$ by applying Rule~\eqref{eq:actor.progress}:
\gcondrule{Invocation-main}{
  \Sigma[\Omain]=\{\contF{\Fmain}{s_{\ref{ex:abs.inv}}}\} \quad
  \sigma' = \last(sh') \\
 \rho = [F' \mapsto f_m]  \quad
  \emptyset \text{ consistent} \quad 
  \wf{sh''}\\
  \emptyset\pop\concatTr{\consTr{{\sigma'}[f'\mapsto F',F'\mapsto
      *]}{\invEvV{\sigma'}{1,o_m,m,F'}{\{F'\}}}}{\contF{}{\Simple{await}\
      f'? ; s_{\ref{ex:abs.get}}}}\\
\multicolumn{1}{l}{\quad \in \valBP{\sigma'}{\Omain,\Fmain}{s_{\ref{ex:abs.inv}}}}
}{
  sh', \Sigma \to
  %
  sh'', 
  \Sigma[\Omain\mapsto \{\contF{\Fmain}{\Simple{await}\ f'? ; s_{\ref{ex:abs.get}}}\}]
}
where $sh''$ is the concrete trace
\[
  \begin{array}{rl}
    &\chopTrSem{sh'}{\rho(\consTr{{\sigma'}[f'\mapsto F',F'\mapsto *]}{\invEvV{\sigma'}{1,o_m,m,F'}{\{F'\}}})}^{\Omain}\\
    =&\langle I_P\rangle \curvearrowright \ldots \curvearrowright \sigma' \curvearrowright  \consTr{\sigma''}{\consTr{\sigma'}{\invEvp{}{\Omain}{1,o_m,m,f_m}}}\\
  \end{array}
\]
with $\sigma'' = {\sigma'}[f'\mapsto f_m] = [a' \mapsto 1, x'\mapsto o_m, f' \mapsto f_m, y' \mapsto 0]$.

At this point, one could in principle choose to continue with the
\Simple{await} statement in $main$ with Rule~\eqref{eq:actor.progress}
or to start executing method $m$ on object~$o_m$ with
Rule~\eqref{eq:actor.future.invREv2}.  However,
Rule~\eqref{eq:actor.progress} is not applicable, because it does not
result in a well-formed trace before method~$m$ returns.  Therefore,
we proceed with the second option:

\gcondrule{Body-m}{
  \Sigma[o_m] = \emptyset \quad
  \lookup{m, \Gi} = m(n)\ \scope 
    \quad
    \sigma'' = \last(sh'')
    \quad
    f_m\in\FId \\
\hspace*{-6cm}    \emptyset \pop \invREvV{\sigma''}{N, X', m,F''}{\{N, X', F''\}}\\
    \qquad\concatTr{\consTr{\sigma''[n'\mapsto N, N\mapsto *, X'\mapsto *, F'' \mapsto *]}{}}{\cont{\{\  s_{\ref{ex:abs.add}}\ \}}} \in \valBP{\sigma''}{o_m,f_m}{m(n)\ \scope} 
    \\
    \rho=[N \mapsto 1, X' \mapsto \Omain, F'' \mapsto f_m] \quad
    \emptyset \text{ consistent}
    \quad
    \wf{sh_m}
  }{
    sh'', \Sigma \to
    sh_m,
    \Sigma[o_m\mapsto \{\contF{f_m}{\{\  s_{\ref{ex:abs.add}}\ \}}\}]
}
where the mapping $\Sigma$ is updated to $[\Omain\mapsto
\{\contF{\Fmain}{\Simple{await}\ f? ; s_{\ref{ex:abs.get}}}\}, o_m \mapsto
\{\contF{f_m}{\{\  s_{\ref{ex:abs.add}}\ \}}\}]$, and $sh_m$ is the concrete trace
\[
  \begin{array}{rl}
    &\chopTrSem{sh''}{\rho(\consTr{\sigma''[n'\mapsto N, N\mapsto *, X'\mapsto *, F'' \mapsto *]}{\invREvV{\sigma''}{N, X', m,F''}{\{N, X',m, F''\}}})^{o_m}}\\
      =&\langle I_P\rangle \curvearrowright \ldots \curvearrowright \sigma'' \curvearrowright \consTr{\sigma_m}{\consTr{\sigma''}{\invREvp{}{o_m}{1, \Omain, m,f_m}}}
  \end{array}
\]
with $\sigma_m = \sigma''[n'\mapsto 1]$.

We have to continue with
$\{\ s_{\ref{ex:abs.add}}\ \}$ from the method $m$.  Executing
$n' := n'+1$ in~$s_{\ref{ex:abs.add}}$ updates $\sigma_m$ to
$\sigma_m' = \sigma_m[n' \mapsto \valBP{\sigma_m}{o_m,f_m}{n'+1}] = [a' \mapsto 1, x'\mapsto o_m, f' \mapsto f_m, y' \mapsto 0, n' \mapsto 2]$, which extends
the concrete trace to $sh_m'$ as follows:
\[
  \langle I_P\rangle \curvearrowright \ldots \curvearrowright \sigma_m \curvearrowright \sigma_m'\ .
\]
We proceed with the \Simple{return} statement in method $m$ by
applying Rule~\eqref{eq:actor.progress} once again:

\gcondrule{Return-m}{
  \Sigma[o_m] = \{\contF{f_m}{\Simple{return}\;n'}\} 
    \quad
    \sigma_m' = \last(sh_m')
    \\
    \emptyset\pop\concatTr{\compEv{\sigma_m'}{f_m,2}}{\cont{\zero}} \in \valBP{\sigma_m'}{o_m,f_m}{\Simple{return}\;n'} 
    \quad
    \rho=[] \quad
    \emptyset \text{ consistent}
    \quad
    \wf{sh_m''}
  }{
  sh_m', \Sigma \to
  sh_m'',
  \Sigma[o_m\mapsto \emptyset
]
}
where $\Sigma[o_m\mapsto \emptyset]$ is the result of simplifying $\Sigma[o_m\mapsto \emptyset + \contF{f_m}{\zero}]$. The resulting concrete trace $sh_m''$ is:
\[
  \begin{array}{rl}
    &\chopTrSem{sh_m'}{\rho(\compEv{\sigma_m'}{f_m,2})^{o_m}}\\
    =&  \langle I_P\rangle \curvearrowright \ldots \sigma_m'
       \curvearrowright \ \compEvp{}{o_m}{f_m,2} \curvearrowright
       \sigma_m'\ .
  \end{array}
\]

The mapping $\Sigma$ is now updated to
$[\Omain\mapsto \{\contF{\Fmain}{\Simple{await}\ f'? ;
  s_{\ref{ex:abs.get}}}\}, o_m \mapsto \emptyset] 
$.  Executing
the \Simple{await} statement in $main$ at this point produces a
well-formed trace by Rule~\eqref{eq:actor.progress}:

\gcondrule{Await-main}{ \Sigma[\Omain] = \{\contF{\Fmain}{\Simple{await}\ f'?
    ; s_{\ref{ex:abs.get}}}\} \quad \sigma_m' = \last(sh_m'') \\
  \emptyset\pop\concatTr{\consTr{{\sigma_m'}[V\mapsto\ast]}{\compREvV{\sigma_m'}{f_m,V}{\{V\}}}}{\cont{s_{\ref{ex:abs.get}}}}
  \in \valBP{\sigma_m'}{\Omain,\Fmain}{\Simple{await}\ f'? ;
    s_{\ref{ex:abs.get}}}
  \\
  \rho=[V \mapsto 2] \quad \emptyset \text{ consistent} \quad
    \wf{sh'''}
}{
  sh_m'', \Sigma \to
  sh''', 
  \Sigma[\Omain\mapsto \{\contF{\Fmain}{s_{\ref{ex:abs.get}}}\}]
}
with the concrete trace $sh'''$:
\[
  \begin{array}{rl}
    &\chopTrSem{sh_m''}{\rho(\consTr{{\sigma_m'}[V\mapsto\ast]}{\compREvV{\sigma_m'}{f_m,V}{\{V\}}})^{\Omain}}\\
    =&  \langle I_P\rangle \curvearrowright \ldots \curvearrowright \sigma_m' \curvearrowright \compREvp{}{\Omain}{f_m,2} \curvearrowright \sigma_m' \curvearrowright  \sigma_m'\ .
  \end{array}
\]
The final state $\sigma_m'$ is actually a simplification of
$\sigma_m'[V\mapsto 2]$ by removing the concretised symbolic variable
$V$ and thus identical to $\sigma_m'$.

Finally, we execute the \Simple{get} statement, again by Rule~\eqref{eq:actor.progress}:

\gcondrule{Get-main}{ \Sigma[\Omain] =
  \{\contF{\Fmain}{s_{\ref{ex:abs.get}}}\} \quad \sigma_m' = \last(sh''')
  \\
  \emptyset\pop\concatTr{\consTr{{\sigma_m'}[y'\mapsto V',V'\mapsto
      \ast]}{\compREvV{\sigma_m'}{\valBP{\sigma_m'}{\Omain,\Fmain}{f},V'}{\{V'\}}}}{\cont{\zero}}
  \in \valBP{\sigma_m'}{\Omain,\Fmain}{s_{\ref{ex:abs.get}}}
  \\
  \rho=[V' \mapsto 2] \quad \emptyset \text{ consistent} \quad
  \wf{sh_P} }{ sh''', \Sigma \to sh_P, \Sigma[\Omain\mapsto \emptyset
] } which results the mapping
$\Sigma = [\Omain\mapsto \emptyset 
, o_m \mapsto \emptyset 
]$, and produces the concrete trace~$sh_P$:
\[
  \begin{array}{rl}
    &\chopTrSem{sh'''}{\rho(\consTr{{\sigma_m'}[y'\mapsto V',V'\mapsto \ast]}{\compREvV{\sigma_m'}{    \valBP{\sigma_m'}{\Omain,\Fmain}{f}\},V'}{\{V'\}}})^{\Omain}}\\
    =&\consTr{\sigma'}{\consTr{\sigma}{\consTr{\newEvp{}{\Omain}{o_m,\_}}{\consTr{\sigma}{\consTr{[a' \mapsto 0, x' \mapsto 0, f' \mapsto 0, y' \mapsto 0]}{\langle I_P\rangle}}}}}\\
    &\curvearrowright \consTr{\sigma''}{\consTr{\sigma'}{\invEvp{}{\Omain}{1,o_m,m,f_m}}}%
    \curvearrowright \consTr {\sigma_m'}{\consTr{\sigma_m}{\consTr{\sigma''}{\invREvp{}{o_m}{1, \Omain, m,f_m}}}}\\
    &\curvearrowright \consTr{\sigma_m'}{\compEvp{}{o_m}{f_m,2}}%
      \curvearrowright    \consTr{\sigma_m'}{\consTr{\sigma_m'}{\compREvp{}{\Omain}{f_m,2}}}\\
    &\curvearrowright \consTr{\sigma_P}{\consTr{\sigma_m'}{\compREvp{}{\Omain}{f_m,2}}}
  \end{array}
\]
with $\sigma_P= [a' \mapsto 1, x'\mapsto o_m, f'\mapsto f_m, y'\mapsto 2, n' \mapsto
2]$. 
	

\subsubsection{ABS}
\label{sec:abs}

ABS is an actor-based executable modeling language
\cite{ABSFMCO10,Haehnle13,ABS-LangSpec} that falls in the class of
Active Object languages \cite{ActiveObjects17}. ABS is closely related
to the language in Figure~\ref{fig:actor.future}. Before we make this
relation precise, we mention the main features of ABS that have
\emph{not} been discussed in this paper:
\begin{description}
\item[Functional Expressions] ABS lets the user declare algebraic
  datatypes and provides a pure, functional language with pattern
  matching over these. An evaluation semantics for such a language is
  standard \cite{Mitchell96}. It is easily incorporated into states
  and their evaluation.
\item[Interfaces] ABS supports multiple implementations of
  interfaces, but no inheritance or dynamic dispatch. It is sufficient
  to equip the lookup table $\Gi$ with suitable selectors. 
\item[Modules, Traits] ABS has a simple module system. Modules do not
  have an operational semantics, but manage the name space. One can
  remove them by replacing relative with absolute names. ABS supports
  code reuse via traits: sets of method declarations that can be added
  to a class via a \emph{uses} clause. Like modules, traits can be
  assumed to have been resolved.
\item[Error Handling] ABS can throw and catch exceptions. The
  semantics of the corresponding statements is a combination of the
  case distinction rule~\eqref{eq:If} and the local jump
  rule~\eqref{eq:goto}. The design of the corresponding local
  evaluation rules is left as an exercise.
\end{description}

Otherwise, ABS is identical to the Active Object language discussed in
the present section. Specifically, ABS tasks are \emph{atomic by
  default}. Their execution is only suspended explicitly, either at
the end of a method, or with a \Simple{suspend} statement (equivalent to
``$\Simple{await}\,\Simple{true}$''), or by $\Simple{await}\,g$.

\ifadvancedActors

\subsubsection{Advanced Actor features}

Here we discuss two actors features that do not belong to ABS but LAGC is capable to encode.

\paragraph{Dataflow synchronisation on futures}
Dataflow explicit futures (see Godot paper for example) is a feature that changes the synchronisation pattern of futures. The principle is that a get on a future can only be resolved with a value that is not a future. If a future is resolved with another future, then we synchronise on the resolution of the other future before considering the future as resolved. This synchronisation requires a specific type system but allows to write easily programs that delegate the resolution of a future to another actor.

The change is as follows. We first suppose that we have an operator $\isfut(e)$ that returns true if and only if $e$ is a future. The well-formedness criteria for the completion reaction event is changed as follows:

\begin{multline*}
      \wf{\consTr{\compREvp{}{o}{v,v'}}{sh}}  =  \wf{sh} \land 
		\big( \neg \isfut(v) \land v=v'\\
		\lor (\exists o',  \compEvp{}{o'}{v,v''} \in sh  \land \wf{\consTr{\compREvp{}{o}{v'',v'}}{sh}} \big)
\end{multline*}

Here we benefit greatly from the decoupling introduced by our semantics: val states when to synchronise and what is the effect of synchronisation in terms of events, well-formedness states what it means for a computation to be completed, and thus when the synchronisation is resolved.

\paragraph{FIFO service of request}
In many actor and active object languages, there is a mailbox associated to each actor and it is guaranteed that messages are taken into account in the same order as they are received by the actor. One solution is to use the same approach as in Definition~\ref{def:wf-fifo}, instead of defining FIFO sending/reception of messages from a single source you can ensure that, for any active object, the order  in which methods are executed is the same as the order of invocation events targetting the given active object.

This corresponds to the following equation:
\begin{multline*}
      \wf{\consTr{\invREvp{}{o'}{\many{v},o,m,f}}{sh}}  =  \wf{sh} \land \ \invEvp{}{o}{\many{v},o',m,f} \in sh \land \\
\forall \invEvp{}{o''}{\many{v'},o',m',f'} \in sh .\,  
	\invEvp{}{o''}{\many{v'},o',m',f'} <_{sh} \invEvp{}{o}{\many{v},o',m,f} \\
		\implies \invREvp{}{o'}{\many{v'},o'',m',f'} \in sh
\end{multline*}
Strictly speaking, this ordering guarantees a ``mailbox behaviour'' where each message is pushed synchronously on a mailbox (there is one mailbox per actor) and asynchronously fetched from the mailbox. This ordering corresponds to what is implemented in most actor languages with FIFO service policy.

%

Note that these features are almost enough to have a semantics for the ASP language (or the ProActive library). The last thing to notice is that in these languages, futures are not explicitly typed and hus one should insert $\Simple{get}$ statements each time a variable that might contain a future is accessed. This method has already applied for the static analysis of ProActive programs. 

\fi


\section{Related Work}
\label{sec:related-work}

We position our work within denotational semantics. The discussion
focuses on three different strands of work on trace semantics:
Semantics based on execution traces, on communication traces, and on
hybrid traces combining execution states and communication events.  To
first motivate the use of traces semantics, we start by drawing some
lines from state-transformer semantics.  We then consider related work
on execution traces and communication traces. We finally discuss
related work on hybrid traces, combining these directions.

\paragraph{State-transformer semantics.}
State-transformer semantics explain a program statement as a
transition from one program state to another
\cite{Hoare69,apt09vscp}. This style of semantics elegantly hides
intermediate states in the transition, thereby providing an
abstraction from the inner, intermediate states visited in the actual
execution of statements. State-transformer semantics is fully abstract
for sequential while-programs, resulting in a compositional semantics
with respect to, for example, partial correctness behavior. However, for
parallel shared-variable languages the final outcome of a program may
depend on the scheduling of the different atomic sections, leading to
non-deterministic behavior. For a parallel statement, the state
transformer that results from the parallel composition of component
statements cannot be determined from the transformer for these
statements alone~\cite{Mosses06,brookes96ic}. Hence, state-transformer
semantics for parallel languages are not compositional; it is
necessary to capture the intermediate states at which interleaving may
occur. Mosses observes \cite{Mosses06} that by capturing these
intermediate states as so-called \emph{resumptions} \cite{hennessy79mfcs},
the resulting denotational semantics corresponds much more closely to
operational semantics \cite{Plotkin04}.

Early work on verification of concurrent systems extended
state-transformer semantics with additional side-conditions, and is as
such non-compositional; for example, interference freedom tests were
used for shared variable concurrency \cite{owicki76acta} and
cooperation tests for synchronous message passing
\cite{apt80toplas}. Compositional approaches were introduced for
shared variables in the form of rely-guarantee \cite{Jones81a} and for
synchronous message passing in the form of assumption-commitment
\cite{misra81tse}. Extending these principles for compositional
verification, object invariants can be used to achieve modularity
(for example, \cite{jacobs05sefm}).

\paragraph{Execution traces}
The use of execution traces to describe program behavior can be
motivated by the need for additional structure in the semantics of
statements, such that parallel statements can be described
compositionally from the semantics of their components.  Brookes
developed a trace semantics for a shared-variable parallel language
based on \emph{transition traces} \cite{brookes96ic}. In his work, a
command is described by a sequence of state transformers reflecting
the different atomic blocks in the execution of the command, but
abstracting from the inner states of these blocks. Transition traces
are thus more abstract than resumptions that are discussed above.  The
semantics of parallel commands then corresponds to the union of these
segments, obtained by stitching together the state transitions
corresponding to the different atomic sections in different
ways. Brookes shows that this semantics, which can be realised at
different levels of granularity, is fully abstract.  In fact, we
considered this solution in our work, but decided against it because
it constructs infinitely large mathematical objects in order to
include all possible states in which the next atomic block can
start. Instead, we opted for continuation markers, to be resolved in
the composition, as well as symbolic traces that are concretised on
demand during the composition phase.  This ensures that all semantic
elements are finite, which makes the LAGC semantics easy and effective
to compute.

\paragraph{Communication traces}
Communication traces first appeared in the object-oriented setting
\cite{Dahl77} and then for CSP \cite{hoare85}. Soundararajan developed
an axiomatic proof system for CSP using histories and projections
\cite{soundararajan84toplas}, which was compositional and removed the
need for cooperation tests. Zwiers developed the first sound and
complete proof system using communication traces
\cite{zwiers89comp}. Jeffrey and Rathke introduced a semantics for
object-oriented programs based on communication traces, and showed
full abstraction for testing equivalence
\cite{jeffrey05tcs,jeffrey05esop}. Reasoning about asynchronous method
calls and cooperative scheduling using histories was first done for
Creol \cite{dovland05swste} and later adapted to Dynamic Logic
\cite{ahrendt12scp}.  Din introduced a proof system based on four
communication events, significantly simplifying the proof rules
\cite{dinJLA12}, and extended the approach to futures
\cite{DinOwe14,DinOwe15}. This four-event proof system, which forms
the basis for KeY-ABS \cite{DinBH15}, was the starting point for the
communication events used in our paper. This way, communication events
can always be local to one semantic object, and their ordering is
captured by well-formedness predicates. Compared to this line of work,
we introduce continuation markers and locally symbolic traces. This
allows us to constructively derive global traces, in contrast to, for
example, Din's work which, building on Soundararajan's work, simply
eliminates global traces which do not project correctly to the local
trace sets of the components.

\paragraph{Hybrid traces.}
Brookes' \emph{action traces} \cite{Brookes02} bear some similarity to
our work.  He aims at denotational semantics using collecting
semantics, explicitly represents divergence, synchronises
communication using events, and captures parallel execution of two
components by means of a so-called \emph{mutex fairmerge} which
ensures that both components get the chance to be executed.  Action
traces were used to develop a semantics for concurrent separation
logic \cite{Brookes07}, where scheduling is based on access to shared
resources with associated invariants and data races are
exposed. Brookes' work elegantly develops a trace semantics for
low-level programming mechanisms with lock resources. However, it does
not cover the dynamic spawning of processes, procedure calls, method
invocations, and similar topics covered in our work, which led to the
locally abstract, globally concrete formulation of hybrid trace
semantics.  In previous work \cite{DHJPT17} we used \emph{scheduling
  events} and well-formedness properties over scheduling events to
capture all legal interleavings among concurrent objects at a
granularity similar to that of action traces.  In contrast, parallel
execution in the present paper is based on a more fine-grained
interleaving of processes, exploiting continuations, by means of
different composition rules (for example,~\eqref{eq:actor.progress}
and~\eqref{eq:actor.future.invREv2}) to capture global and local
interleaving of processes. The approach of Brookes~\cite{Brookes02,DHJPT17} is
easily expressible within our framework.
We obtain the equivalent of the fairmerge operator by collecting all
traces that can be constructed by the composition rules from a given
concrete initial state.

In trace semantics, the composition of traces from parallel executions
can be formalised in different ways; e.g., as continuation semantics,
erasure semantics or collecting semantics. We have opted for the first
approach and captured the interleaved execution by means of local
continuations, such that the global trace is obtained by gradually
unfolding traces that correspond to local execution by means of
continuations.
Kamburjan has explored erasure semantics in a hybrid trace semantics
developed as a foundation for behavioral program logic
\cite{Kamburjan20,Kamburjan19}; this work uses
an explicit representation of the heap which is exploited for
composition by introducing a fresh symbolic heap as a local
placeholder for execution in other objects at the scheduling points.
The introduction of fresh symbolic states for merging local behaviors
can equivalently be replaced by a set collecting all possible concrete
states after the scheduling point, such that the composition of local
behaviors can be done by selecting the appropriate, compatible
concrete traces rather than by instantiating a symbolic trace. This
approach has been studied in the context of hybrid trace semantics in
\cite{KDHJ20}.

Recent work \cite{ABHKMS17} on similar communication structures for
ASP/ProActive uses parameterised labelled transition systems with
queues to model interaction with futures in a fine-grained,
operational manner. In contrast, our work with traces allows futures
to be abstracted into communication events and well-formedness
conditions. We are unaware of previous work on programming language
semantics that captures different communication patterns as
well-formedness constraints over trace semantics; these constraints
are in fact \emph{orthogonal} to the building blocks of the trace
semantics and allow to study language behavior ranging over
architectures that support different forms of communication within a
single semantic framework.

More recently, \cite{Xia2020} \cite{Xia2020} defined a data
structure called interaction trees that enable a denotational and
somehow functional definition of the semantics for imperative programs
where interaction trees encode the interaction with the
environment. An interaction tree constitutes a form of trace that
expresses the effects of the program but abstract away from
states. The authors provide co-inductive principles for reasoning on
diverging programs, which is more precise than our current handling of
infinite traces.  Their semantics can be considered as more
compositional than ours, because no concretisation is needed, but the
approach is only valid for sequential programs.

The pure trace-based proof system of ABS \cite{DinOwe14,DinOwe15}
requires strong hiding of local state: the state of other objects can
only be accessed through method calls, so shared state is internal and
controlled by cooperative scheduling. Consequently, specifications can
be purely local. More expressive specifications require significantly
more complex proof systems, for example modifies clauses in Boogie
\cite{jacobs05sefm}, fractional permissions \cite{HeuleLMS13} in
Chalice~\cite{Leino09}, or dynamic frames in KeY \cite{Mostowski20}.
To specify fully abstract interface behavior these systems need to
simulate histories in an ad hoc manner, see
\cite[Figure~1]{jacobs05sefm}.  A combination of permission-based
separation logic~\cite{Amighi15} and histories has recently been
proposed for modular reasoning about multi-threaded
concurrency~\cite{Zaharieva-Stojanovski14}. The motivation for our
work stems from our aim to devise compositional proof systems to
verify protocol-like behavior for object-oriented languages
\cite{DHJPT17} by combining communication histories for ABS with
the trace modality formulas of Nakata et al.\
\cite{NakataUustalu15,BubelCHN15}.

\section{Future Work and Conclusion}
\label{sec:future-work-concl}

\subsection{Future Work}
\label{sec:future-work}

The original motivation for the present work was to create a modular
semantics that aligns well with the kind of program logics used in
deductive verification~\cite{HaehnleHuismann19}. In
Sections~\ref{sec:dynamic-logic-simple} and \ref{sec:interleaving} we
reached this goal by providing a program logic and calculus for the
shared variable language and by giving a compact soundness proof
(Thm.~\ref{thm:soundness}) for the sequential fragment. This result
can and should be extended to a state-of-art calculus for the language
in Section~\ref{sec:actor.future} \cite{KDHJ20}.

An obvious direction for future work is to fully mechanize the LAGC
semantics in a proof assistant \cite{Isabelle02,Coq04}. In fact, all
definitions and theorems from Sections~\ref{sec:basics},
\ref{sec:while}, and~\ref{sec:interleaving} have been mechanized and
proven\footnote{See \url{https://gitlab.com/Niklas_Heidler/mechanization-of-lagc-semantics-in-isabelle} and \cite{Heidler21}.} in Isabelle/HOL. The
mechanization is \emph{executable} in the sense that all traces in the
examples of the mentioned sections can be generated automatically from
the HOL theories. We plan to extend the Isabelle/HOL mechanization to
cover also Sections~\ref{sec:multicore}, \ref{sec:actor.class}, and
\ref{sec:actor.future}.

In Sections~\ref{sec:interleaving}--\ref{sec:case-studies} we showed
that a wide variety of parallel programming concepts can be formalised
in LAGC style with relative ease. It would be interesting to explore
how far this can be carried. On the one hand, one could try to
formalise the semantics of a major programming language, such as Java
or C. On the other hand, one could try to apply the LAGC framework to
weakly consistent memory models \cite{AdveHill93} or to the target the
language of concurrent separation logic \cite{Brookes07}. A further
possible extension concerns programs with non-terminating, atomic
segments, see Section~\ref{sec:atomic}.
A mechanized, executable LAGC semantics opens the possibility of
\emph{early prototyping} new semantics of concurrent and distributed
languages with relative ease.

We rely heavily on well-formedness of traces, but we did not discuss
how to define properties on possibly infinite sets of infinite traces
with events, for example, notions related to fairness. This is a
complementary problem to trace generation, and not the focus of the
present paper.

Our approach makes composition of local rules easy, and extension of
the language straightforward. Global trace composition on the contrary
relies on well-formedness criteria for the whole trace; even if this
criteria is defined modularly, it is not compositional in the sense
that a trace can become invalid by extension of the criteria or
parallel composition. This is often unavoidable because the semantics
of many concurrent models like CCS and $\pi$-calculus are by nature
very much sensitive to the execution context. However, depending on
the language and the form of concurrency, different compromises could
be found, defining a greater part of the concurrent semantics in a
symbolic and compositional way. For example our semantics of futures
is similar to message passing but one could probably specify new
composition tools that better take into account the single-assignment
property ensured by futures.

\subsection{Conclusion}
\label{sec:conclusion}

The semantics of concurrent programming languages is inherently a
technically demanding subject. In the best case, a formal semantics
can illuminate the design space created by the choice of different
concurrency concepts, it shows the consequences of these choices, and
it makes different version comparable. This is only possible in a
semantic framework that enforces uniform and modular definitions, and
it is what we strove to achieve with LAGC: the central design decision
is to strictly separate local, sequential computations and their
parallel composition. Also, we decided to render local evaluations
abstract. In this way, one and the same \emph{schematic} semantic
evaluation rule can be re-used for any initial state, executing
processor, and destiny of the result. While abstract local rules are
not a theoretical necessity, they drastically simplify the complexity
of definitions.

A central technical problem to address in a locally-globally divided
setup is to ensure that enough context information is available when
composing concurrent behavior. Instead of computing \emph{all}
possible states, in which an atomic segment can continue, we pass the
remaining code to be executed as a continuation. Again, this
constitutes a considerable technical simplification compared to the
former option. For example, in Section~\ref{sec:actor.future} we
characterised the behavior of active objects concisely with a suitable
definition of a continuation pool.

But mere continuation is not sufficient: one needs to orchestrate
different local evaluations within a global trace, for example, to
ensure a method is called before it is being executed. This is
achieved by suitable events emitted during local
evaluation. Orchestration of local computations by events leads to a
further separation of concerns: many concurrency models can be
characterised in a declarative manner by well-formedness of events
inside traces, as shown in Section~\ref{sec:communication-patterns}.

We believe the achieved separation of concerns, locally abstract
evaluation---trace composition with a continuation
pool---orchestration of computations by well-formedness, provides a
flexible and usable semantic framework to formalise and compare
concurrent languages in. It is close modern deductive verification
calculi and could even be fully mechanised.

\subsection*{Acknowledgment}

We would like to thank the following people for carefully reading
drafts of this paper and for spotting several omissions, errors,
typos, and inaccuracies: Lukas Gr\"atz, Dilian Gurov, Niklas Heidler,
Eduard Kamburjan, and Marco Scaletta.

\bibliographystyle{plain}

\end{document}